\pdfoutput=1
\newif\ifarxiv
\arxivtrue

\ifarxiv
  \documentclass[acmsmall,screen,nonacm]{acmart}  
\else
  \documentclass[conference]{IEEEtran}
  \usepackage{amsmath,amssymb,stmaryrd}
\fi
\usepackage{dutchcal} % Supports lowercase calligraphic glyphs 
\usepackage{xspace} % Command \xspace
\usepackage{booktabs} % Nice tabulars
\usepackage{bbding} % a dingbat font for fun
\usepackage{marginnote}
\usepackage{cleveref}
\usepackage{amsthm}
\usepackage{algorithm2e}
\usepackage{enumerate,hyperref}
\ifarxiv
  \usepackage[bibliography=common]{apxproof}
\else
  \usepackage[appendix=strip]{apxproof}
\fi
\ifarxiv
    \renewcommand\footnotetextcopyrightpermission[1]{}
\fi
\AtBeginDocument{%
  \providecommand\BibTeX{{%
    \normalfont B\kern-0.5em{\scshape i\kern-0.25em b}\kern-0.8em\TeX}}}

\usepackage[xcolor, hyperref, notion, quotation, electronic]{knowledge}
\usepackage{mathcommand}
\knowledgeconfigure{quotation, protect quotation={tikzcd}}
\knowledgeconfigure{diagnose line=true, diagnose bar=true}

\definecolor{Dark Ruby Red}{HTML}{7c1b1e}
\definecolor{Dark Blue Sapphire}{HTML}{004c5c} %lighter 18%
\definecolor{Dark Gamboge}{HTML}{be7c00}

\IfKnowledgePaperModeTF{
}{
    \knowledgestyle{intro notion}{color={Dark Ruby Red}, emphasize}
    \knowledgestyle{notion}{color={Dark Blue Sapphire}}
    \hypersetup{
        colorlinks=true,
        breaklinks=true,
        linkcolor={Dark Blue Sapphire}, % Links to sections, pages, etc.
        citecolor={Dark Blue Sapphire}, % Links to bibliography
        filecolor={Dark Blue Sapphire}, % Links to local file
        urlcolor={Dark Blue Sapphire},
    }
    \IfKnowledgeElectronicModeTF{
    }{
        \knowledgeconfigure{anchor point color={Dark Ruby Red}, anchor point shape=corner}
        \knowledgestyle{intro unknown}{color={Dark Gamboge}, emphasize}
        \knowledgestyle{intro unknown cont}{color={Dark Gamboge}, emphasize}
        \knowledgestyle{kl unknown}{color={Dark Gamboge}}
        \knowledgestyle{kl unknown cont}{color={Dark Gamboge}}
    }
}
\knowledge{url={https://complexityzoo.net/Complexity_Zoo:L\#logcfl}}
| LOGCFL

\knowledge{notion}
 | full version

\knowledge{notion}
 | notion@notice
 | definition@notice

\knowledge{notion}
 | $\omega$-regular tree satisfiability
 
 \knowledge{notion}
 | successful

 \knowledge{notion}
 | tree
 | trees

\knowledge{notion}
 | $\SigmaN$-labeled $\SigmaE$-tree
 | $\SigmaN$-labeled $\SigmaE$-trees
 | $\pset{P}$-labeled $A$-trees
 | $\SigmaN$-labeled tree
 | $\SigmaN$-labeled trees

\knowledge{notion}
 | two-way alternating parity tree automaton
 | TWAPTA
 | TWAPTAs

\knowledge{notion}
 | satisfiable with respect to@twapta

\knowledge{notion}
 | test transitions
 
 \knowledge{notion}
 | index

 \knowledge{notion}
 | non-deterministic finite automaton
 | NFA
 | NFAs
\knowledge{notion}
 | initial@NFA
 \knowledge{notion}
 | final@NFA
\knowledge{notion}
 | graph
 | graphs

\knowledge{notion}
 | vertex
 | vertices

\knowledge{notion}
 | edge
 | edges

\knowledge{notion}
 | minor
 | minors

 \knowledge{notion}
 | underlying graph@C

\knowledge{notion}
 | underlying graph
 | underlying graphs

\knowledge{notion}
 | bi-pointed multigraph
 | bi-pointed multigraphs
 | 2-graph
 | 2-graphs

\knowledge{notion}
 | label

\knowledge{notion}
 | variable
 | variables

\knowledge{notion}
 | formula
 | formulas
 | Formulas

 \knowledge{notion}
 | program
 | programs

 \knowledge{notion}
 | expression
 | expressions
 | Expressions

\knowledge{notion}
 | positive

\knowledge{notion}
 | $C$-satisfying assignment
 | $C'$-satisfying assignment

\knowledge{notion}
 | atomic propositions
 | atomic proposition

\knowledge{notion}
 | atomic programs
 | atomic program

\knowledge{notion}
 | program intersection

\knowledge{notion}
 | $\ICPDL$ program over $\allStates$
 | $\ICPDL$ programs over $\allStates$
 | $\ICPDL$ programs over
 | $\ICPDL$ program over

\knowledge{notion}
 | $a$-successor

\knowledge{notion}
 | $a$-predecessor

\knowledge{notion}
 | intersection width

 \knowledge{notion}
 | conjunctive width

 \knowledge{notion}
 | conjunctive programs
 | conjunctive program

 \knowledge{notion}
 | Kripke structure
 | Kripke structures
 | structure@kripke
 | structures@kripke

 \knowledge{notion}
 | finite degree

 \knowledge{notion}
 | distance

\knowledge{notion}
 | world
 | worlds

\knowledge{notion}
 | $\Tw$-model property
 | $\Tw[1]$-model property

 \knowledge{notion}
 | countable model property

 \knowledge{notion}
 | $\+G$-CPDL

\knowledge{notion}
 | CPDL

\knowledge{notion}
 | ICPDL

\knowledge{notion}
 |loop-CPDL

\knowledge{notion}
 | $k$-ary simulation relation
 | $k$-simulation
 | $k$-simulates

 \knowledge{notion}
 | $k$-ary bisimulation relation
 | $k$-bisimulation
 | $k$-bisimilar

 \knowledge{notion}
 | \quasibisim

\knowledge{notion}
 | graph database
 | graph databases
 | Graph databases
 | database
 | databases

\knowledge{notion}
 | evaluation
 | evaluated
 | evaluations
 | evaluates
 | evaluate

\knowledge{notion}
 | model checking

\knowledge{notion}
 | atom
 | atoms

\knowledge{notion}
 | output variables

\knowledge{notion}
 | Boolean

 \knowledge{notion}
 | satisfiability problem
 | satisfiability

\knowledge{notion}
 | satisfies
 | satisfy
 | satisfying

\knowledge{notion}
 | conjunctive query
 | conjunctive queries
 | Conjunctive queries
 | CQ
 | CQs

 \knowledge{notion}
 | CQPDL

 \knowledge{notion}
 | Regular Queries
 | Regular queries
 | regular queries

\knowledge{notion}
 | unions of CQs
 | union of CQs
 | UCQs

\knowledge{notion}
 | CRPQ
 | CRPQs
 | conjunctive regular path queries

\knowledge{notion}
 | C2RPQ
 | C2RPQs
 | conjunctive two-way regular path queries

\knowledge{notion}
 | UCRPQ
 | (U)CRPQ
 | UCRPQs
 | unions of CRPQs
 | union of CRPQs

\knowledge{notion}
 | C2RPQ
 | C2RPQs
 | conjunctive two-way regular path query
 | CRPQ with two-way navigation
 | two-way navigation

\knowledge{notion}
 | UC2RPQ
 | (U)C2RPQ
 | UC2RPQs
 | unions of C2RPQs
 | union of C2RPQs

\knowledge{notion}
 | connected

\knowledge{notion}
 | containment
 | contained
 | contains
 | contain

\knowledge{notion}
 | containment problem

\knowledge{notion}
 | equi-expressive
 
\knowledge{notion}
 | equivalent
 | equivalence
 | semantical equivalence
 | semantically equivalent

\knowledge{notion}
 | expansion
 | expansions

\knowledge{notion}
 | partial homomorphism
 | partial homomorphisms

 \knowledge{notion}
 | homomorphism
 | homomorphisms
 | Homomorphisms
 | homomorphic
 | C2RPQs homomorphisms
 
\knowledge{notion}
 | strong onto homomorphism
 | strong onto homomorphisms
 | strong onto

\knowledge{notion}
 | sublanguage
 | sublanguages

\knowledge{notion}
 | closed under sublanguages
 | closure under sublanguages

\knowledge{notion}
 | atom refinement
 | atom refinements

\knowledge{notion}
 | atom $m$-refinement
 | atom $m$-refinements
 | $m$-refinements@atom

\knowledge{notion}
 | equality atom
 | equality atoms

\knowledge{notion}
 | contraction
 | contractions
 | contract

\knowledge{notion}
 | refinements
 | refinement

 \knowledge{notion}
 | $m$-refinements
 | $m$-refinement

\knowledge{notion}
 | refinement length

\knowledge{notion}
 | tree decomposition
 | tree decompositions

\knowledge{notion}
 | A@treedec
\knowledge{notion}
 | B@treedec
\knowledge{notion}
 | C@treedec

\knowledge{notion}
 | full bag
 | full
 | full bags

\knowledge{notion}
 | path decomposition
 | path decompositions

\knowledge{notion}
 | bag
 | bags

\knowledge{notion}
 | contains@tw
 | contained@tw
 | containing@tw

\knowledge{notion}
 | width

\knowledge{notion}
 | tree-width
 | Tree-width

 \knowledge{notion}
 | tree-width@kripke
 | Tree-width@kripke

\knowledge{notion}
 | nice tree decomposition
 | niceness
 | nice
 | nice tree decompositions

\knowledge{notion}
 | tagged tree decomposition
 | tagged tree decompositions

\knowledge{notion}
 | tagged
 | tagging
 | tag

\knowledge{notion}
 | nice tagged tree decomposition
 | nice tagged tree-decom\-position

\knowledge{notion}
 | cyclic@path
 | acyclic@path
 | cyclic path
 | acyclic path
 | acyclic paths

\knowledge{notion}
 | locally acyclic
 | local acyclicity

 \knowledge{notion}
 | simple path

\knowledge{notion}
 | non-branching path
 | non-branching paths

\knowledge{notion}
 | Path induced
 | path induced
 | induced path
 | induces
 | induce
 | induced

\knowledge{notion}
 | leaves
 | leaving

\knowledge{notion}
 | semantic tree-width $k$ problem
 | semantic tree-width $1$ problem

\knowledge{notion}
 | semantic tree-width

\knowledge{notion}
 | maximal under-approximation by infinitary unions
 | maximal under-approximation
 | maximal under-approximations

\knowledge{notion}
 | approximations

\knowledge{notion}
 | infinitary unions
 | infinitary union
 | Infinitary unions

\knowledge{notion}
 | trio

\knowledge{notion}
 | key lemma
 | Key lemma

\knowledge{notion}
 | atomic
 | non-atomic
 | non-atomic bags

\knowledge{notion}
 | profile

\knowledge{notion}
 | type
 | types

 \knowledge{notion}
 | $k$-summary query
 | $k$-summary queries
 | summary query
 | summary queries

 \knowledge{notion}
 | simple regular expression
 | simple regular expressions
 | SRE
 | SREs

 \knowledge{notion}
 | path-$l$ approximation
 | path-$l$ approximations

\knowledge{notion}
 | recursive atom
 | recursive atoms
 | recursive@rat

\knowledge{notion}
 | non-recursive atom
 | non-recursive atoms
 | non-recursive@nrat

\knowledge{notion}
 | simple normal form

\knowledge{notion}
 | simple 

\knowledge{notion}
 | complete

\knowledge{notion}
 | labeling

\knowledge{notion}
 | at least as expressive as

\knowledge{notion}
 | $(k+1)$-clique formula
 | $3$-clique formula

\knowledge{notion}
 | extensible clique
 | extensible cliques

\knowledge{notion}
 | ML-simulation
 | ML-simulates

\knowledge{notion}
 | subexpression
 | subexpressions

\knowledge{notion}
 | is a clique in

\knowledge{notion}
 | composition

\knowledge{notion}
 | switching

\knowledge{notion}
 | $(s,u)$-run

\knowledgenewrobustcmd{\expfun}{\cmdkl{\textit{exp}}}

\newcommand{\fullverurl}{http://arxiv.org/xxx}

\knowledgenewrobustcmd{\uGraph}[1]{\cmdkl{\ensuremath{\textup{G}_{#1}}}}
\knowledgenewrobustcmd{\uGraphC}[1]{\cmdkl{\ensuremath{\textup{G}_{#1}}}}%
\knowledgenewrobustcmd{\SigmaE}{\cmdkl{\Sigma_{E}}}
\knowledgenewrobustcmd{\SigmaN}{\cmdkl{\Sigma_{N}}}
\knowledgenewrobustcmd{\tree}{\cmdkl{\textit{tree}}}
\knowledgenewrobustcmd{\posB}{\cmdkl{B^+}}
\knowledgenewrobustcmd{\translatedC}{\cmdkl{\hat \pi_{T_C,f}}}%
\knowledgenewrobustcmd{\negdepth}{\cmdkl{\textit{neg-depth}}}
\knowledgenewrobustcmd{\ICPDLp}{\ensuremath{\cmdkl{\textup{ICPDL}^{\!+}}}\xspace}
\knowledgenewrobustcmd{\allStates}{\cmdkl{\+Q}}

\knowledgenewrobustcmd{\kSimGame}[1][k]{\cmdkl{\mathbf{G}[\pebblesim{#1}\!]}}%
\knowledgenewrobustcmd{\kBisimGame}[1][k]{\cmdkl{\mathbf{G}[\pebblebisim{#1}\!]}}%

\renewcommand{\epsilon}{\varepsilon}
\newif\ifproofappendix

\newrobustcmd\labelwithproof[1]{%
\label{#1}%
\ifproofappendix%
\marginnote{\footnotesize{%
  \textnormal{First stated in page~\pageref{#1}.}%
}}
\else%
\marginnote{\footnotesize{%
  \textnormal{See the proof of \Cref{#1} in page~\pageref{proof-#1}.}%
}}%
\fi%
}

\newrobustcmd\introinrestatable[1]{%
\ifproofappendix%
\kl{#1}%
\else%
\intro{#1}%
\fi%
}

\newrobustcmd\introinrestatableopt[1]{%
\ifproofappendix%
\kl[#1]{#1}%
\else%
\intro[#1]{#1}%
\fi%
}

\newrobustcmd\recall[1]{
  \proofappendixtrue%
    #1*
  \proofappendixfalse%
}

\usepackage[backgroundcolor=orange!20, textcolor={Dark Ruby Red}, textsize=tiny, disable]{todonotes}
\definecolor{green}{RGB}{0,120,0}
\definecolor{hlyellow}{RGB}{250, 250, 190}
\definecolor{diegoeditcolor}{RGB}{210,210,255}
\definecolor{remieditcolor}{RGB}{210,255,210}
\definecolor{edwineditcolor}{RGB}{150,200,55}
\newcommand{\sidediego}[1]{\todo[backgroundcolor=diegoeditcolor, size=\tiny]{{\bf D:} #1}}
\newcommand{\sidesanti}[1]{\todo[backgroundcolor=remieditcolor, size=\tiny]{{\bf S:} #1}}
\newcommand{\sideedwin}[1]{\todo[backgroundcolor=edwineditcolor, size=\tiny]{{\bf E:} #1}}

\definecolor{light-gray}{gray}{0.9}

\newcommand{\proofcase}[1]{\noindent\colorbox{light-gray}{#1}~~}
\newcommand{\proofsubcase}[1]{\noindent$\rhd$~\underline{~#1~}~}

\newrobustcmd{\wrote}{\color{wrote}\scriptsize\text{wrote}}
\newrobustcmd{\advised}{\color{advised}\scriptsize\text{advised}}

\renewcommand{\phi}{\varphi}

\newcommand{\set}[1]{\{#1\}}
\newrobustcmd{\Nat}{\mathbb{N}}
\newcommand{\dcup}{\mathop{\dot\cup}} % disjoint union
\knowledgenewrobustcmd\pset[1]{\cmdkl{\wp}(#1)} % powerset

\newcommand{\resp}[1]{(resp.~#1)}
\newcommand{\ie}{\textit{i.e.}}
\newcommand{\aka}{\textit{a.k.a.}\ }
\newcommand{\eg}{\textit{e.g.}}
\newcommand{\cf}{\textit{cf.}}

\knowledgenewrobustcmd{\A}{\mathbb{A}} % Finite alphabet of labels
\knowledgenewrobustcmd{\Aext}{\cmdkl{\mathbb{A}^\pm}} % Finite alphabet of labels
\knowledgenewrobustcmd\vertex[1]{\cmdkl{V}(#1)}
\knowledgenewrobustcmd\edges[1]{\cmdkl{E}(#1)}

\knowledgenewrobustcmd{\Graphloop}{\cmdkl{\+G_\circlearrowleft}}
\knowledgenewrobustcmd{\Graphcap}{\cmdkl{\+G_\cap}}

\knowledgenewrobustcmd\worlds[1]{\cmdkl{W}\!(#1)} % set of worlds of kripke structure

\knowledgenewrobustcmd\qvar{\footnotesize\bullet} % quantified variable

\newcommand{\eqdef}{\mathrel{{\mathop:}}=}
\newcommand{\eqqdef}{\mathrel{{\mathop:}{\mathop:}}=}

\knowledgenewrobustcmd{\kHoms}[1][k]{\cmdkl{\textit{Hom}_{#1}}}%
\knowledgenewrobustcmd{\Prop}{\cmdkl{\mathbb{P}}}
\knowledgenewrobustcmd{\Prog}{\cmdkl{\mathbb{A}}}
\knowledgenewrobustcmd\Vars{\cmdkl{\mathbb{V}}}
\knowledgenewrobustcmd{\ICPDLg}[1]{\ensuremath{\cmdkl{\textup{ICPDL}^{\!+}}\!}(#1)}
\newcommand{\twticpdl}{\ICPDLg{\Tw[2]}}
\knowledgenewrobustcmd{\subexpr}{\ensuremath{\cmdkl{\textup{sub}}}}
\knowledgenewrobustcmd{\dbracket}[1]{\cmdkl{\llbracket} #1 \cmdkl{\rrbracket}}

\knowledgenewrobustcmd{\CPDL}{\ensuremath{\cmdkl{\textup{CPDL\xspace}}}}
\knowledgenewrobustcmd{\PDL}{\ensuremath{\cmdkl{\textup{PDL\xspace}}}}
\knowledgenewrobustcmd{\ICPDL}{\ensuremath{\cmdkl{\textup{ICPDL}}}\xspace}
\knowledgenewrobustcmd{\loopCPDL}{\ensuremath{\cmdkl{\textup{loop-CPDL}}}\xspace}
\knowledgenewrobustcmd{\CPDLp}{\cmdkl{\ensuremath{\textup{CPDL}^{\!+}}}}
\knowledgenewrobustcmd{\CPDLg}[1]{\ensuremath{\cmdkl{\textup{CPDL}^{\!+}}\!}(#1)}

\knowledgenewrobustcmd{\iwidth}{\cmdkl{\textup{iw}}}
\knowledgenewrobustcmd{\Iwidth}{\cmdkl{\textup{IW}}}

\knowledgenewrobustcmd{\cqsize}[1]{\cmdkl{\textup{cw}}(#1)}
\knowledgenewrobustcmd{\Cqsize}[1]{\cmdkl{\textup{CW}}(#1)}
\knowledgenewrobustcmd{\Cqsizealt}{\cmdkl{\textup{CQ}}}

\knowledgenewrobustcmd{\dbracketaut}[1]{\cmdkl{\llbracket} #1 \cmdkl{\rrbracket}}
\knowledgenewrobustcmd{\dbracketnfa}[1]{\cmdkl{\llbracket} #1 \cmdkl{\rrbracket}}
\knowledgenewrobustcmd{\dbrackett}[1]{\cmdkl{\llbracket} #1 \cmdkl{\rrbracket}}
\knowledgenewrobustcmd{\indexAut}[1]{\cmdkl{i(}#1\cmdkl{)}}

\knowledgenewrobustcmd{\pebblesim}[1]{\mathrel{\cmdkl{\rightharpoonup_{#1}}}}
\knowledgenewrobustcmd{\notpebblesim}[1]{\mathrel{\cmdkl{{\not\rightharpoonup}_{#1}}}}
\knowledgenewrobustcmd{\pebblebisim}[1]{\mathrel{\cmdkl{\rightleftharpoons_{#1}}}}
\knowledgenewrobustcmd{\pebblequasibisim}[1]{\mathrel{\cmdkl{\rightharpoondown^{\!\!\leftrightarrow}_{#1}}}}
\knowledgenewrobustcmd{\notpebblequasibisim}[1]{\mathrel{\cmdkl{\not\rightharpoondown^{\!\!\leftrightarrow}_{#1}}}}
\knowledgenewrobustcmd{\notpebblebisim}[1]{\mathrel{\cmdkl{{\not\rightleftharpoons}_{#1}}}}

\knowledgenewrobustcmd{\pebblesimneg}[1]{\mathrel{\cmdkl{\rightharpoonup^\lnot_{#1}}}}

\knowledgenewrobustcmd{\mlsim}{\mathrel{\cmdkl{\rightharpoonup}}}
\knowledgenewrobustcmd{\mlbisim}{\mathrel{\cmdkl{\rightleftharpoons}}}

\knowledgenewrobustcmd{\pathl}{\cmdkl{\mathbf{P}_{\!l}}} % path-l approximation query

\knowledgenewrobustcmd\subaut[3]{#1\cmdkl{[#2,#3]}}

\knowledgenewrobustcmd\bagmap{\cmdkl{\mathbf{v}}}
\knowledgenewrobustcmd\tagmap{\cmdkl{\mathbf{t}}}
\knowledgenewrobustcmd\tagmappath[1]{\cmdkl{\mathbf{t}[#1]}}
\newrobustcmd\tagmappathprime[1]{%
  \withkl{\kl[\tagmappath]}{%
    \cmdkl{\mathbf{t}'[#1]}%
  }%
}

\knowledgenewrobustcmd{\atom}[1]{\,\xrightarrow{\smash{#1}}\,}
\knowledgenewrobustcmd{\coatom}[1]{\,\xleftarrow{\smash{#1}}\,}
\knowledgenewrobustcmd{\atoms}[1]{\cmdkl{\textnormal{Atoms}}(#1)}
\knowledgenewrobustcmd{\contained}{\mathrel{\cmdkl{\subseteqq}}}
\newrobustcmd{\strcontained}{
  \mathrel{\withkl{\kl[\contained]}{\cmdkl{%
    \subsetneqq
  }}}
}
\disablecommand\equiv
\suggestcommand\equiv{Use instead \semequiv for semantical equivalence.}
\knowledgenewrobustcmd{\semequiv}{\mathrel{\cmdkl{\LaTeXequiv}}} %for formulas
\knowledgenewrobustcmd{\langsemequiv}{\mathrel{\cmdkl{\LaTeXequiv}}} %for languages
\knowledgenewrobustcmd{\lleq}{\mathrel{\cmdkl{\leqq}}}
\knowledgenewrobustcmd{\notlleq}{\mathrel{\cmdkl{\nleqq}}}
\knowledgenewrobustcmd{\lleqs}{\mathrel{\cmdkl{\lneqq}}}
\knowledgenewrobustcmd{\vars}{\cmdkl{\textit{vars}}} % set of variables of a query
\knowledgenewrobustcmd{\UCtwoRPQ}{\cmdkl{\textnormal{UC2RPQ}}}
\newrobustcmd{\CtwoRPQ}{%
  \withkl{\kl[\UCtwoRPQ]}{\cmdkl{%
    \textnormal{C2RPQ}
  }}
}

\knowledgenewrobustcmd{\UCRPQSRE}{\ensuremath{\cmdkl{\textup{UCRPQ}(\textup{SRE})}}}
\newrobustcmd{\CRPQSRE}{%
  \withkl{\kl[\UCRPQSRE]}{\cmdkl{%
    \textup{CRPQ}(\textup{SRE})
  }}
}

\newcommand{\xrightarrowdbl}[2][]{%
  \xrightarrow[#1]{#2}\mathrel{\mkern-14mu}\rightarrow
}
\knowledgenewrobustcmd\surj{%
    \mathrel{\cmdkl{%
      \xrightarrowdbl{\textit{\tiny hom}}
    }}
}
\knowledgenewrobustcmd{\fun}{f}
\knowledgenewrobustcmd{\homto}{\mathrel{\cmdkl{\xrightarrow{\textit{\tiny hom}}}}} %(normal) homomorphism

\knowledgenewrobustcmd{\class}{\mathcal{C}}
\knowledgenewrobustcmd{\Tw}[1][k]{\cmdkl{\textup{TW\!}_{#1\!}}}

\knowledgenewrobustcmd{\Refin}[1][]{\cmdkl{\textnormal{Ref}^{\smash{#1}}}}
\knowledgenewrobustcmd{\MUA}[2]{\cmdkl{\ensuremath{\textnormal{App}_{#2}(#1)}}}
\knowledgenewrobustcmd{\MUAHom}[2]{\cmdkl{\ensuremath{\textnormal{App}_{#2}^{\smash{\star}}(#1)}}}
\knowledgenewrobustcmd{\MUAHomBounded}[3]{\cmdkl{\ensuremath{\textnormal{App}_{#2}^{\smash{\star,#3}}(#1)}}}
\knowledgenewrobustcmd{\type}{\cmdkl{\textnormal{type}}}
\knowledgenewrobustcmd{\Qapp}{\cmdkl{\ensuremath{\textnormal{App}_{\Tw}^{\textup{zip}}(\gamma)}}}
\knowledgenewrobustcmd{\contract}[1]{\cmdkl{[}#1\cmdkl{]}}

\newcommand{\complexityclass}[1]{\textup{\textsc{#1}}\xspace}

\newcommand{\np}{\complexityclass{NP}}
\newcommand{\DP}{\complexityclass{DP}}

\newcommand{\ptime}{\complexityclass{PTime}}
\newcommand{\exptime}{\complexityclass{ExpTime}}

\newcommand{\wone}{\complexityclass{W[1]}}
\newcommand{\fpt}{\complexityclass{FPT}}

\newrobustcmd\pitwo{\ensuremath{\Pi^p_2}}

\newrobustcmd\sigmatwo{\ensuremath{\Sigma^p_2}}

\knowledgenewrobustcmd{\tw}{\ensuremath{\cmdkl{\text{tw}}}}

\newcommand{\lo}{{\rm loop}\xspace}

\knowledgenewrobustcmd{\mapcoord}{\cmdkl{\mapsto}}

\renewcommand{\phi}{\varphi}

\newtheorem{claim}{Claim}
\newtheoremrep{claim}{Claim}
\newtheoremrep{lemma}{Lemma}
\newtheoremrep{theorem}{Theorem}
\ifarxiv\else
  \newtheoremrep{corollary}{Corollary}
\fi

\newtheoremrep{proposition}{Proposition}

\newcommand{\tup}[1]{\langle #1 \rangle}
\newtheorem{remark}{Remark}
\newtheoremrep{remark}{Remark}

\begin{document}
\sloppy
\ifarxiv
\title[PDL on Steroids]{PDL on Steroids: on Expressive Extensions of PDL with Intersection and  Converse}
\else
\title{PDL on Steroids: \\on Expressive Extensions of PDL \\with Intersection and  Converse}
\fi
\ifarxiv
    \author{Diego Figueira}
    \email{diego.figueira@cnrs.fr}
    % \orcid{0000-0003-0114-2257}
    \affiliation{%
      \institution{Univ. Bordeaux, CNRS,  Bordeaux INP, LaBRI, UMR 5800}
      % \streetaddress{P.O. Box 1212}
      % \city{F-33400 Talence}
      \country{F-33400 Talence, France}
      \postcode{F-33400}
    }
    \author{Santiago Figueira}
    \email{santiago@dc.uba.ar}
    \affiliation{%
        \institution{Univ. Buenos Aires, FCEN, DC \& CONICET-UBA, ICC}
        \country{Argentina}}

      \author{Edwin Pin}
    \email{epin@dc.uba.ar}
      \affiliation{%
        \institution{Univ. Buenos Aires, FCEN, DM \& CONICET-UBA, ICC}
        \country{Argentina}}
    \renewcommand{\shortauthors}{Figueira, Figueira, and Pin}
\else
  \author{
    \IEEEauthorblockN{Diego Figueira} 
    \IEEEauthorblockA{Univ. Bordeaux, CNRS,  \\Bordeaux INP, LaBRI, UMR 5800 \\ F-33400 Talence, France} 
    \and 
    \IEEEauthorblockN{Santiago Figueira} 
    \IEEEauthorblockA{Univ. Buenos Aires, FCEN, DC \&\\CONICET-UBA, ICC\\Argentina}
    \and 
    \IEEEauthorblockN{Edwin Pin} 
    \IEEEauthorblockA{Univ. Buenos Aires, FCEN, DM \&\\CONICET-UBA, ICC\\Argentina}}
\fi

\ifarxiv
  \begin{abstract}
  %!TEX root = main.tex
%
We introduce $\CPDLp$, a family of expressive logics rooted in Propositional Dynamic Logic ($\PDL$).
In terms of expressive power, $\CPDLp$ strictly contains $\PDL$ extended with intersection and converse (\aka~$\ICPDL$) as well as Conjunctive Queries (CQ), Conjunctive Regular Path Queries (CRPQ), or some known extensions thereof (Regular Queries and CQPDL). 
We investigate the expressive power, indistinguishability via bisimulations, satisfiability, and model checking for $\CPDLp$.

We argue that natural subclasses of $\CPDLp$ can be defined in terms of the "tree-width" of the underlying graphs of the formulas.
We show that the class of $\CPDLp$ formulas of "tree-width" 2 is equivalent to $\ICPDL$, and that it also coincides with $\CPDLp$ formulas of "tree-width" 1. However, beyond "tree-width" 2, incrementing the "tree-width" strictly increases the expressive power. We characterize the expressive power for every class of fixed "tree-width" formulas in terms of a bisimulation game with pebbles. Based on this characterization, we show that $\CPDLp$ has a tree-like model property.
We prove that the satisfiability problem is decidable in 2\exptime on fixed "tree-width" formulas, coinciding with the complexity of $\ICPDL$. 
We also exhibit classes for which satisfiability is reduced to \exptime.
Finally, we establish that the "model checking" problem for fixed "tree-width" formulas is in \ptime, contrary to the full class $\CPDLp$.

  \end{abstract}
  \maketitle
\else
  \IEEEoverridecommandlockouts
  \IEEEpubid{\makebox[\columnwidth]{979-8-3503-3587-3/23/\$31.00~
  \copyright2023 IEEE \hfill} \hspace{\columnsep}\makebox[\columnwidth]{ }}
  \maketitle
  \begin{abstract}
  %!TEX root = main.tex
%
We introduce $\CPDLp$, a family of expressive logics rooted in Propositional Dynamic Logic ($\PDL$).
In terms of expressive power, $\CPDLp$ strictly contains $\PDL$ extended with intersection and converse (\aka~$\ICPDL$) as well as Conjunctive Queries (CQ), Conjunctive Regular Path Queries (CRPQ), or some known extensions thereof (Regular Queries and CQPDL). 
We investigate the expressive power, indistinguishability via bisimulations, satisfiability, and model checking for $\CPDLp$.

We argue that natural subclasses of $\CPDLp$ can be defined in terms of the "tree-width" of the underlying graphs of the formulas.
We show that the class of $\CPDLp$ formulas of "tree-width" 2 is equivalent to $\ICPDL$, and that it also coincides with $\CPDLp$ formulas of "tree-width" 1. However, beyond "tree-width" 2, incrementing the "tree-width" strictly increases the expressive power. We characterize the expressive power for every class of fixed "tree-width" formulas in terms of a bisimulation game with pebbles. Based on this characterization, we show that $\CPDLp$ has a tree-like model property.
We prove that the satisfiability problem is decidable in 2\exptime on fixed "tree-width" formulas, coinciding with the complexity of $\ICPDL$. 
We also exhibit classes for which satisfiability is reduced to \exptime.
Finally, we establish that the "model checking" problem for fixed "tree-width" formulas is in \ptime, contrary to the full class $\CPDLp$.

  \end{abstract}
\fi


\smallskip
%!TEX root = main.tex

\noindent
\raisebox{-.4ex}{\HandRight}\ \ This pdf contains internal links: clicking on a "notion@@notice" leads to its \AP ""definition@@notice"".\footnote{\url{https://ctan.org/pkg/knowledge}}
\ifarxiv
\smallskip

\noindent
\AP
\raisebox{-.4ex}{\HandRight}\ \ This manuscript is the full version of \cite{thispaper}.
\else
%!TEX root = main.tex

\noindent
\AP
\raisebox{-.4ex}{\HandRight}\ \ The ""full version"" of this article is available at \url{\fullverurl}.
\fi
%!TEX root = main.tex

\section{Introduction}

This paper can be seen as a proof of concept.
Since "Kripke structures" and "graph databases" (\ie, structures over a binary signature) are essentially the same: is it possible to amalgamate the expressive logics (for "Kripke structures") and query languages (for "graph databases") into one well-behaved framework? 

Our point of departure is Propositional Dynamic Logic, or $\PDL$, extended with intersection and converse. This is a well-studied logic, and amongst the most expressive decidable logics on "Kripke structures". But on the other hand, we want to be able to test for "conjunctive queries", or more generally "CRPQs", which are the basic building block for query languages for "graph databases". 
The outcome is an arguably natural and highly expressive logic, which inherits all good computational and model-theoretical behaviors of $\PDL$ and "CRPQs", which has the flavor of $\PDL$ but allows for richer ``conjunctive'' tests.
We name this logic ``$\CPDLp$''.\sidesanti{en estos párrafos de acá arriba dejé CRPQ}

$\PDL$ was originally conceived as a logic for reasoning about programs \cite{DBLP:journals/jcss/FischerL79}. However, variants of $\PDL$ are nowadays used in various areas of computer science, in particular in description logics, epistemic logics, program verification, or for querying datasets (see \cite{DBLP:journals/japll/Lange06,DBLP:conf/csl/GollerL06} for more applications).
One of the most studied extensions of $\PDL$ is the addition of converse navigation, "program intersection" and program complement. 
In particular, adding converse and "program intersection" operators results in a well-behaved logic, known as $\ICPDL$ (`I' for "intersection@program intersection", `C' for converse). $\ICPDL$ has decidable "satisfiability@satisfiability problem", polynomial time "model checking", and enjoys a ``tree-like'' model property \cite{DBLP:conf/csl/GollerL06}.
A simpler version of "program intersection" studied before is $\CPDL$ extended with "program looping@\loopCPDL", which can state that a program starts and ends at the same point, known as $\loopCPDL$, and it is strictly less expressive than $\ICPDL$.
However, adding the complement of programs results in a logic with a highly undecidable satisfiability problem \cite{harel2001dynamic,DBLP:conf/csl/GollerL06}, although its model checking is still polynomial time \cite{DBLP:journals/japll/Lange06}.

In this work we explore a family of logics which in particular generalizes "program intersection". 
The intuition is that the intersection of two $\PDL$ programs $\pi$ and $\pi'$ could be viewed as the conjunction of two atoms $R_\pi(x,y) \land R_{\pi'}(x,y)$ over the binary relations denoted by the programs $\pi$ and $\pi'$ (\ie, essentially a "conjunctive query"). Pursuing this idea, we introduce an extension of $\CPDL$ that we call $\CPDLp$ which allows for testing of any arbitrary number of atoms. For example, a formula can test for an $n$-clique of $\pi$-related elements with $\bigwedge_{1 \leq i < j \leq n} R_\pi(x_i,x_j)$. Further, these tests can be nested, composed, or iterated just like conventional $\PDL$ programs. 
The resulting logic seems appealing from an expressiveness point of view: it captures not only $\ICPDL$ but also several "graph database" query languages studied lately in the quest for finding well-behaved expressive query languages. These include "C2RPQ", "Regular Queries" and "CQPDL". See \Cref{fig:expressive-power} for a general idea of where $\CPDLp$ sits.\sidesanti{acá cambié a C2RPQ}

Further, subclasses of $\CPDLp$ can be naturally defined by restricting the allowed underlying graphs (\aka Gaifman graph) of these new kind of tests. Thus, $\CPDLg{\+G}$ is the restriction to tests whose underlying graphs are in the class of graphs $\+G$. In particular, for suitable (and very simple) classes we find $\ICPDL$ and $\loopCPDL$.

\begin{figure}
    \noindent
    \ifarxiv
        \includegraphics[width = .6 \textwidth]{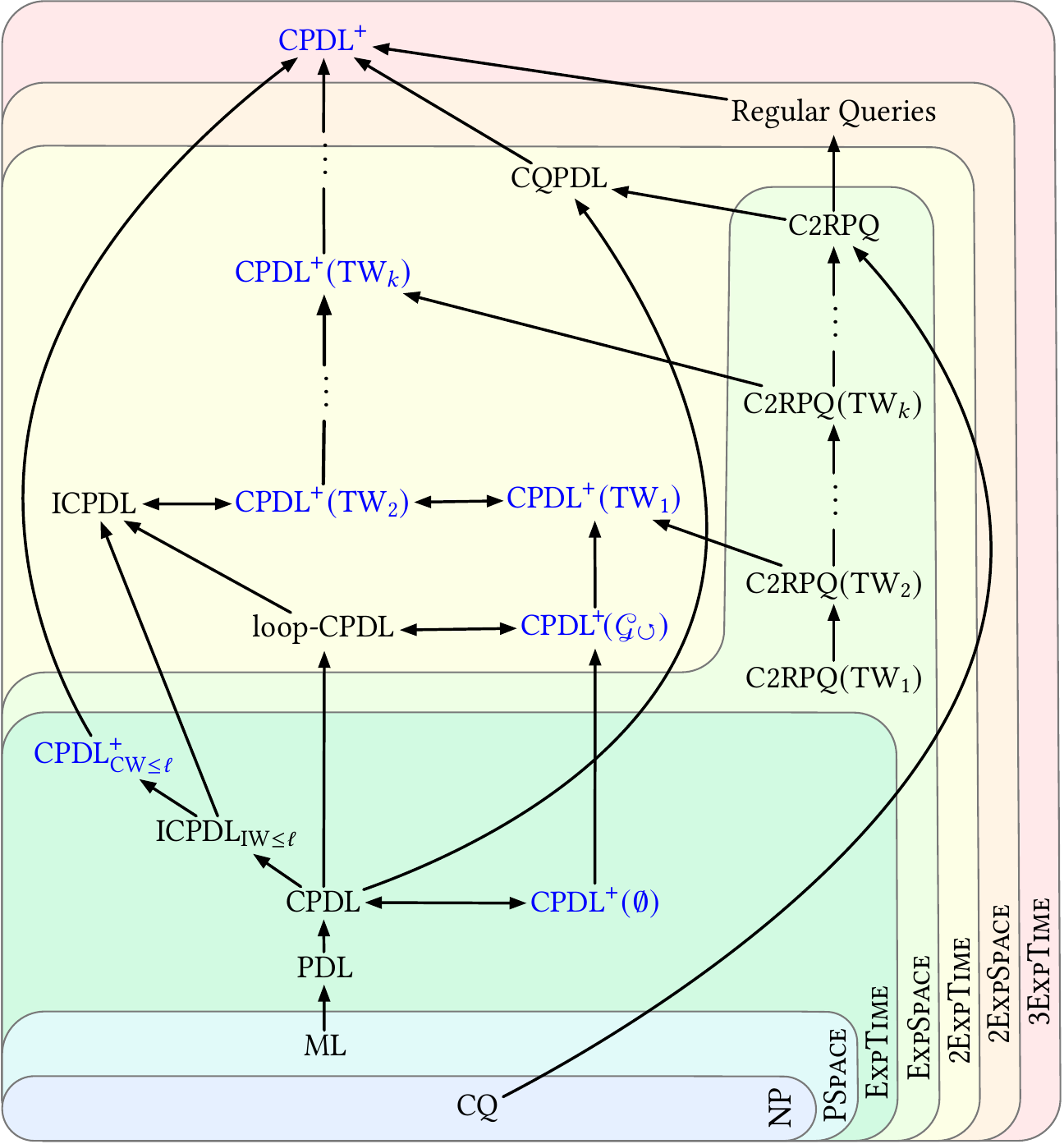}    
    \else
        \includegraphics[width = .47 \textwidth]{img/expressivity2.pdf}    
    \fi
    \caption{{\bf The landscape of expressiveness and complexity for PDL-like and CQ-like formalisms.} 
    { The arrow goes in the direction of the more expressive language, and all arrows are witnessed via polynomial time translations, except the one departing from "Regular Queries", which is exponential. For the query languages "CQ", "C2RPQ", "Regular Queries", and "CQPDL" we restrict to connected queries of arity 1 in order to be able to compare the expressive power. 
    By "C2RPQ"($\Tw$) we denote (connected) "C2RPQs" whose underlying graph has "tree-width" $\leq k$.
    $\ICPDL_{\Iwidth\leq\ell}$ for any $\ell$, is the set of all $\ICPDL$ "formulas" $\phi$ s.t.\ $\Iwidth(\phi) \leq \ell$, introduced in \cite{DBLP:journals/jsyml/GollerLL09} and defined in \Cref{sec:solving-omega-reg-sat}.
    $\CPDLp_{\Cqsizealt \leq \ell}$ for any $\ell$, is the set of all $\CPDLp$ "formulas" $\phi$ s.t.\ $\Cqsize{\phi} \leq \ell$, as defined in \Cref{sec:solving-omega-reg-sat}. The complexities correspond to the basic reasoning problem for each formalism: satisfiability for logics closed under negation, and containment for query languages of the "CQ"/"C2RPQ" family. In \color{blue}blue \color{black} we highlight the family of logics introduced in the present work.}
    }
    \label{fig:expressive-power}
\end{figure}

%!TEX root = main.tex

\subsection{Summary of results}

\paragraph{Expressivity}
As already mentioned, $\CPDLp$ contains some previously studied graph database query languages (see \Cref{fig:expressive-power}). Further, depending on the class $\+G$ of graphs $\CPDLg{\+G}$ can be identified with different previous formalisms. In particular, $\CPDL$ corresponds to $\CPDLg{\emptyset}$ on the empty class.

In \Cref{sec:relation-to-loop-icpdl} we show that, on the other hand, the logic $\loopCPDL$ can also be expressed as $\CPDLg{\Graphloop}$, where $\Graphloop$ is the class with just one "graph" having a self-looping vertex (\Cref{prop:loopCPDL-eq-CPDLgLoop}).
It is also easy to see that $\ICPDL$ and  $\CPDLg{\Graphcap}$ are "equi-expressive", where $\Graphcap$ is the class having just one "graph"  containing two vertices and one "edge" between them. What is more, we also show that $\ICPDL$ is "equi-expressive" to $\CPDLg{\Tw[2]}$. Hence, "formulas" and "programs" of $\CPDLg{\Tw[1]}$, $\CPDLg{\Tw[2]}$, $\CPDLg{\Graphcap}$, and $\ICPDL$ define the same unary and binary relations, respectively (\Cref{thm:ICPDL_equals_TW1_equals_TW2}).

However, as we show in \Cref{sec:caracterization}, beyond $\Tw[2]$ the expressive power starts increasing in an infinite hierarchy, where $\CPDLg{\Tw[k]}$ is strictly contained in $\CPDLg{\Tw[k+1]}$ in terms of expressive power, for every $k \geq 2$ (\Cref{thm:separation}).

\paragraph{Indistinguishability}
We characterize model-indis\-tin\-gui\-shability for $\CPDLg{\Tw[k]}$ in terms of a $k$-pebble local consistency game. That is, we study the conditions under which a pair of models cannot be distinguished by a "formula" or "program" from $\CPDLg{\Tw[k]}$.
In Section~\ref{sec:caracterization} we characterize both the indistinguishability via "positive" "formulas" (\Cref{thm:sim-char}) or arbitrary "formulas"  (\Cref{them:bisim-charac}).
These notions correspond, in some sense, to the lifting of the simulation relation from Modal Logics to $\CPDLg{\Tw[k]}$ \ifarxiv (though we defer this connection to \Cref{sec:MLbisim}). \else . \fi
% We make this connection clear in Section~\ref{sec:MLbisim} by showing how any pair of "structures@@kripke" can be modified in order to use the standard simulation relation.
We show that $\CPDLg{\Tw}$ enjoys the "$\Tw$-model property": if a "formula" of $\CPDLg{\Tw}$ is satisfiable, it is satisfiable in a "model@Kripke structure" of "tree-width" $k$ (\Cref{cor:treewidth-k-model-property}).

\paragraph{Satisfiability}
In \Cref{sec:sat} we study the "satisfiability problem", that is, whether for a given formula there exists a "model@Kripke structure" and "element@world" thereof which satisfy it. We show that, for any $k$, $\CPDLg{\Tw[k]}$ is decidable in $2$\exptime, \ie, the same complexity as $\ICPDL$ and $\loopCPDL$ \cite{DBLP:journals/jsyml/GollerLL09}. We also identify a hierarchy, based on what we call "conjunctive width" instead of "tree-width", so that any class of bounded "conjunctive width" is decidable in \exptime, \ie, the complexity of $\CPDL$.
For the unrestricted logic, namely $\CPDLp$, we obtain a $3$\exptime upper bound (\Cref{thm:sat-cpdlp}). These results exploit the "bounded tree-width model property@$\Tw$-model property" of $\CPDLp$ above.

\paragraph{Model checking}
Finally, we argue in \Cref{sec:modelchecking} that "model checking" on finite "structures@@kripke" is still in polynomial time for any class $\CPDLg{\Tw[k]}$ of bounded "tree-width", unlike the unrestricted logic $\CPDLp$ (\Cref{thm:modelchecking}).

All missing proofs are contained in the \ifarxiv
appendix.
\else
"full version".
\fi

\subsection{Related logics}

As mentioned earlier, $\CPDLp$ can be seen as an ``umbrella logic'' capturing both expressive extensions of $\PDL$, such as $\ICPDL$, and  expressive extensions of "conjunctive queries" studied in the context of "graph databases". We next detail some query languages related to $\ICPDLp$.

\AP
A ""graph database"" is often abstracted as a finite, edge-labeled graph, where labels come from some finite alphabet of relations~(see \eg~\cite{DBLP:conf/pods/Baeza13} for more details), but via trivial adaptations one can use "graph@graph database" query languages on "Kripke structures", or vice-versa.

\AP
% \paragraph{Conjunctive Queries}
""Conjunctive queries"", or \reintro{CQs}, are conjunctions of relational atoms, closed under projection. For example, the "CQ" ``$q(x,x,y) = R(x,y) \land S(y,z) \land T(z,x)$'' returns, on an edge-labelled graph, all triplets of vertices $(u,u,v)$ such that $u,v$ participates in an $R$-$S$-$T$-cycle.

\AP
% \paragraph{CRPQs}
In analogy to "conjunctive queries", the class of ""conjunctive regular path queries"", or \reintro{CRPQs}, allow for atoms of the form $x \xrightarrow{L} y$, where $L$ is a regular expression over the alphabet of the relations, with the semantics that there exists a path from $x$ to $y$ reading a word of relation names in $L$ \cite{DBLP:conf/kr/CalvaneseGLV00}. This language is often considered as the backbone of  any reasonable graph query language \cite{DBLP:conf/pods/Baeza13}. The extension of 
\AP
""conjunctive two-way 
regular path queries"", or ""C2RPQs"" allows $L$ to be built over the alphabet augmented with the converse of the relations.

\AP
% \paragraph{Regular queries}
""Regular queries"" \cite{DBLP:journals/mst/ReutterRV17}, also known as \emph{Nested Positive 2RPQ}~\cite{DBLP:conf/dlog/BourhisKR14}, can be seen as the closure under Kleene star, union and converse of "C2RPQs".Its containment problem  has been shown to be decidable in elementary time, contrary to some more ambitious generalizations of "C2RPQs"  \cite{DBLP:conf/pods/RudolphK13,DBLP:conf/dlog/BourhisKR14}.

\AP
% \paragraph{CQPDL}
Perhaps the language closest in spirit to this work is ""CQPDL"" \cite{DBLP:conf/lics/BenediktBB16}, denoting "conjunctive queries" whose atoms relations can be defined by $\CPDL$ "programs". It subsumes "C2RPQs" in expressive power, but it is incomparable with "Regular queries".

%!TEX root = main.tex

\section{Preliminaries}\label{sec:prelim}
We use the standard notation $\bar u$ to denote a tuple of elements from some set $X$ and $\bar u[i]$ to denote the element in its $i$-th component. 
\AP
If $\bar q=(q_1, \dotsc, q_k)$ is a tuple and $1\leq i\leq k$, by $\bar q[i\intro*\mapcoord r]$ we denote the tuple $(q_1, \dotsc, q_{i-1}, r, q_{i+1}, \dotsc, q_k)$.
For relations $R,S\subseteq X\times X$ we define the ""composition"" $R\circ S:=\set{(u,v)\mid \exists w. uRw, wSv}$.
For a function $f: X \to Y$, we often write $f(\bar u)$ to denote $(f(\bar u[1]), \dotsc, f(\bar u[k]))$, where $k$ is the dimension of $\bar u$. 
\AP
By $\intro*\pset{X}$ we denote the set of subsets of $X$.

\AP
A (simple, undirected) ""graph"" is a tuple $G=(V,E)$ where $V$ is a set of ""vertices"" and $E$ is a set of ""edges"", where each "edge" is a set of vertices of size at most two. We refer to $V$ and $E$ by $V(G)$ and $E(G)$, respectively. 
We henceforth assume that "graphs" have a finite set of "vertices" unless otherwise stated. 
\AP
A ""tree"" is a connected "graph" with no cycles (in particular no edges of the form $\set{v}$).
% We will often write $(x,y)$ to denote $\set{x,y}$. For example, both $(\set{v_1, v_2, v_3}, \set{v_1, v_2,v_3} \times \set{v_1, v_2,v_3})$ and $(\set{v_1, v_2, v_3}, \set{(v_1,v_1),(v_2,v_2),(v_3,v_3),(v_1,v_2), (v_1,v_3), (v_2,v_3)})$ denote the same (reflexive) $3$-clique.
% A graph $G'$ is a ""minor"" of $G$ if $G'$ is the result of deleting edges and vertices and by contracting edges from $G$, see \eg\ \cite{sparsity12} for a formal definition. 

% \paragraph{Tree-width}
\AP
A ""tree decomposition"" of a "graph" $G$ is a
pair $(T, \intro*\bagmap)$ where $T$ is a "tree" and $\bagmap: V(T) \to \pset{V(G)}$ is a function that associates to each node of $T$, called ""bag"",
a set of vertices of $G$. When $x \in \bagmap(b)$ we shall say that the bag $b \in V(T)$
""contains@@tw"" the vertex $x$. 
\AP
Further, it must satisfy the following three properties:
% \begin{enumerate}[A.]
    % \item \label{tw:cond:vertex} 
    ""A@@treedec"". each vertex $x$ of $G$ is "contained@@tw" in at least one "bag" of $T$;
    % \item \label{tw:cond:edge} 
    ""B@@treedec"". for each edge of $G$ between $x$ and $y$, there is at least one "bag" of $T$
        that "contains@@tw" both $x$ and $y$; and 
    % \item \label{tw:cond:connected} 
    ""C@@treedec"".
    for each vertex $x$ of $G$, the set of bags of $T$ "containing@@tw" $x$ is a 
        connected subtree of $V(T)$.
% \end{enumerate}

The "tree decomposition" of an infinite "graph" is an infinite tree and mapping satisfying the conditions "A@@treedec", "B@@treedec", "C@@treedec" above.
\AP
The ""width"" of $(T, \bagmap)$ is the maximum (or supremum if $T$ is infinite) of $|\bagmap(b)|-1$ when $b$ ranges over
$V(T)$. The ""tree-width"" of $G$, notated $\intro*\tw(G)$, is the minimum of the "width" of all "tree decompositions" of $G$.
\AP 
We denote by $\intro*\Tw$ the set of all (finite) "graphs" of "tree-width" at most $k$. See Figure \ref{fig:example_ugraph} (B and C) for "graphs" of different "tree-width" and Figure \ref{fig:example_treedecomp} for an example of a "tree decomposition".

% \knowledgenewrobustcmd\varSrc{\cmdkl{x_s}}
% \knowledgenewrobustcmd\varTgt{\cmdkl{x_t}}
\AP
Let $\intro*\Prop$, $\intro*\Prog$, and $\intro*\Vars$ be countably infinite, pairwise disjoint, sets of ""atomic propositions"", ""atomic programs"", and ""variables"", respectively. 
\AP
A ""Kripke structure"" is a tuple 
\begin{align}
    K = (X,\set{\rightarrow_a \mid a \in \Prog}, \set{X_p \mid p \in \Prop})    
    \label{eq:Kripke}
\end{align}
where $X$ is a set of ``""worlds""'', ${\rightarrow_a} \subseteq X \times X$ is a transition relation for each $a \in \Prog$, and $X_p \subseteq X$ is a unary relation for each $p \in \Prop$. 
\AP
The ""distance"" between two "worlds" of $K$ is simply the length of the shortest path in the underlying undirected "graph" of $K$.
\AP
We say that $K$ is of ""finite degree"" if every "world" of $K$ has a finite number of `neighbors' at "distance" $1$.
\AP
We use $\intro*\worlds K$ to denote the set $X$ of "worlds" of $K$.
\AP
A ""homomorphism"" from a "structure@@kripke" $K=(X,\set{{\rightarrow_a} \mid a \in \Prog}, \set{X_p \mid p \in \Prop})$ to another "structure@@kripke" $K'=(X',\set{{\rightarrow_a'} \mid a \in \Prog}, \set{X'_p \mid p \in \Prop})$ is a function $f : X \to X'$ such that for every $w,w' \in X$ we have that $w \in X_p$ implies $f(w) \in X'_p$ and $(w,w') \in {\rightarrow_a}$ implies $(f(w),f(w')) \in {\rightarrow'_a}$, for every $p \in \Prop$ and $a \in \Prog$. 
\AP
A function $f : \hat X \to X'$ with $\hat X \subseteq X$ is a ""partial homomorphism"" if it is a "homomorphism" from $\hat K$ to $K'$, where $\hat K$ is the substructure of $K$ induced by $\hat X$.
\AP
% \knowledgenewrobustcmd{\kHoms}[1][k]{\cmdkl{\textit{Hom}_{#1}}}%
Let $\intro*\kHoms(K,K')$ be the set of all $(\bar u, \bar v) \in \worlds{K}^k \times \worlds{K'}^k$ such that $\set{\bar u[i] \mapsto \bar v[i] \mid 1 \leq i \leq k}$ is a "partial homomorphism" from $K$ to $K'$. 
We will often abuse notation and write $(\bar u', \bar v') \in \kHoms(K,K')$ for $\bar u', \bar v'$ of dimension $k' < k$  to denote $(\bar u', \bar v') \in \kHoms[k']$.
\AP
The ""tree-width@@kripke"" of a "Kripke structure" is the "tree-width" of its underlying (possibly infinite) "graph".
%!TEX root = main.tex

\section{PDL and its extensions}

\paragraph{$\intro*\CPDL$}
\AP
""Expressions"" of $\CPDL$ can be either
""formulas"" $\phi$ or ""programs"" $\pi$, defined by the following grammar, where $p$ ranges over $\Prop$ and $a$ over $\Prog$:
\begin{align*}
    \phi &\eqqdef p \mid \lnot \phi \mid \phi\land\phi \mid \tup{\pi} \\
    \pi &\eqqdef \epsilon \mid a \mid \bar a \mid \pi \cup \pi \mid \pi¿ \circ \pi \mid \pi^* \mid \phi? 
\end{align*}
\AP
We define the semantics of "programs" $\intro*\dbracket{\pi}_K$ and of "formulas" $\intro*\dbracket{\phi}_K$ in a "Kripke structure" $K$ as \eqref{eq:Kripke}, where  $\dbracket{\pi}_K\subseteq X\times X$ and $\dbracket{\phi}_K\subseteq X$, as follows:
\begin{align*}
    \dbracket{p}_K \eqdef{}& X_p \text{ for $p \in \Prop$}, \qquad
    \dbracket{\lnot\phi}_K \eqdef{}  X\setminus \dbracket{\phi}_K,\\
    \dbracket{\phi_1\land\phi_2}_K \eqdef{}&  \dbracket{\phi_1}_K\cap\dbracket{\phi_2}_K,\\
    \dbracket{\tup{\pi}}_K \eqdef{}&  \{u\in X\mid \exists v\in X.(u,v)\in\dbracket{\pi}_K\},\\
    \dbracket{\epsilon}_K \eqdef{}&  \{(u,u) \mid u\in X\}, \qquad
    \dbracket{a}_K \eqdef{} {\rightarrow_a} \ ,\\
    \dbracket{\bar a}_K \eqdef{}&  \{(v,u)\in X^2\mid (u\rightarrow_av)\},\\
    \dbracket{\pi_1\star\pi_2}_K \eqdef{}&  \dbracket{\pi_1}_K\star\dbracket{\pi_2}_K \text{ for $\star \in \set{\cup,\circ}$},\\
    % \dbracket{\pi_1\circ\pi_2}_K \eqdef{}&  \{(u,v)\in X^2\mid\\ 
    %     &\exists w\in X.(u,w)\in \dbracket{\pi_1}_K, (w,v)\in\dbracket{\pi_2}_K\},\\
    \dbracket{\pi^*}_K \eqdef{}&  \mbox{the reflexive transitive closure of $\dbracket{\pi}_K$},\\
    \dbracket{\phi?}_K \eqdef{}&  \{(u,u) \mid u\in X,u\in\dbracket{\phi}_K\}.
    % }\{(u,v)\in X^2\mid u=v\vee \exists w_1,\dots,w_n \in X.\\ 
    %     & w_1=u,w_n=v,\forall i<n.(w_i,w_{i+1})\in\dbracket{\pi}_K\}\\
\end{align*}
We write $K,u\models\phi$ for $u\in\dbracket{\phi}_K$ and $K,u,v\models\pi$ for $(u,v)\in\dbracket{\pi}_K$ 
\AP
and we write $\phi_1\intro*\semequiv\phi_2$ \resp{$\pi_1\reintro*\semequiv\pi_2$} 
if $\dbracket{\phi_1}_K=\dbracket{\phi_2}_K$ \resp{$\dbracket{\pi_1}_K=\dbracket{\pi_2}_K$} 
\AP
for every "structure@@kripke" $K$, in which case we say that $\phi_1,\phi_2$ \resp{$\pi_1,\pi_2$} are ""equivalent"".
\AP
$\intro*\PDL$ is the fragment $\CPDL$ "expressions" which do not use "atomic programs" of the form $\bar a$.

\paragraph{$\intro*\ICPDL$}
\AP
$\ICPDL$ is defined by $\CPDL$ plus ""program intersection""
$\pi \eqqdef \pi \cap \pi$
% \begin{align*}
%     \pi &\eqqdef \pi \cap \pi
% \end{align*}
with semantics defined by
$    \dbracket{\pi_1\cap\pi_2}_K \eqdef \dbracket{\pi_1}_K\cap\dbracket{\pi_2}_K.
$
% \begin{align*}
%     \dbracket{\pi_1\cap\pi_2}_K & \eqdef \dbracket{\pi_1}_K\cap\dbracket{\pi_2}_K.
% \end{align*}

\paragraph{$\intro*\loopCPDL$}
\AP
$\loopCPDL$ is defined by $\CPDL$ plus
$    \phi  \eqqdef  \lo(\pi) $
% \begin{align*}
%     \phi & \eqqdef  \lo(\pi) 
% \end{align*}
with semantics defined by
    $\dbracket{\lo(\pi)}_K  \eqdef \{x \mid (x,x)\in \dbracket{\pi}_K \}.$
% \begin{align*}
%     \dbracket{\lo(\pi)}_K & \eqdef \{x \mid (x,x)\in \dbracket{\pi}_K \}.
% \end{align*}

\paragraph{Adding steroids}
\AP
An ``""atom""'' is an expression of the form $\pi(x,x')$, where $\pi$ is a {\CPDLp} "program" and $x,x' \in \Vars$. 
\AP
For an "atom" $\pi(x,x')$ we define $\intro*\vars(\pi(x,x')) \eqdef \set{x,x'}$, and for a set of "atoms" $C$ we define $\intro*\vars(C) \eqdef \bigcup_{A \in C} \vars(A)$.

We define $\intro*\CPDLp$ as an extension of $\CPDL$ allowing also "programs" of the form 
% \begin{align}
    \[\pi \eqqdef C[x_s,x_t]\]
    % \label{eq:cq-program}
    % \tag{$\dag$}
% \end{align}
where:
\AP
(1) $C$ is a finite set of "atoms";  
(2) $x_s,x_t\in\vars(C)$\footnote{Note that $x_s$ and $x_t$ may be equal or distinct "variables".};
% both $x_s$ and $x_t$ appear in $C$
and
(3) the ""underlying graph@@C"" 
$\intro*\uGraphC{C}$ of $C$ is connected, where $\uGraphC{C}$ is defined as $V(\uGraphC{C})=\vars(C)$, and $E(\uGraphC{C}) = \set{\vars(A) \mid A \in C}$.\footnote{The fact that we restrict our attention to ``connected'' "conjunctive programs" is unessential for any of our results, but it will allow us to easily relate to previously defined logics, such as $\ICPDL$ and $\loopCPDL$.}
\AP
We call these "programs" ""conjunctive programs"".
Observe that $\set{x_s,x_t}\subseteq{\vars(C)}$, and hence  we also define $\intro*\vars(C[x_s,x_t]) \eqdef {\vars(C)}$.

\AP
A function $f : \vars(C) \to \worlds{K}$, is a ""$C$-satisfying assignment"" if $(f(x),f(x')) \in \dbracket{\pi'}_K$ for every "atom" $\pi'(x,x')\in C$.
The semantics $\dbracket{C[x_s,x_t]}_K$ of a "conjunctive program" on a "Kripke structure" $K$ is the set of all pairs $(w_s,w_t) \in \worlds{K} \times \worlds{K}$ such that $f(x_s) = w_s$ and $f(x_t) = w_t$ for some "$C$-satisfying assignment" $f$.

\AP
For a class $\+G$ of "graphs" we define $\intro*\CPDLg{\+G}$ as the fragment of $\CPDLp$ whose "programs" have one of the shapes allowed in $\+G$. 
Formally, for any "conjunctive program" $\pi = C[x_s,x_t]$ we consider the ""underlying graph"" $\intro*\uGraph{C[x_s,x_t]}$ of $C[x_s,x_t]$ as having $V(\uGraph{C[x_s,x_t]}) = V(\uGraphC{C})$ and $E(\uGraph{C[x_s,x_t]}) = \set{\set{x_s,x_t}} \cup E(\uGraphC{C})$. 
Observe that $x_s$ and $x_t$ are always connected via an "edge" in $\uGraph{C[x_s,x_t]}$ but not necessarily in $\uGraphC{C}$. $\CPDLg{\+G}$ is the fragment of $\CPDLp$ whose only allowed "conjunctive programs" are of the form $C[x_s,x_t]$ where $\uGraph{C[x_s,x_t]}\in\+G$. See Figure \ref{fig:example_ugraph} for an example of "underlying graphs".

\begin{figure}
\ifarxiv
    \includegraphics[width=.7\textwidth]{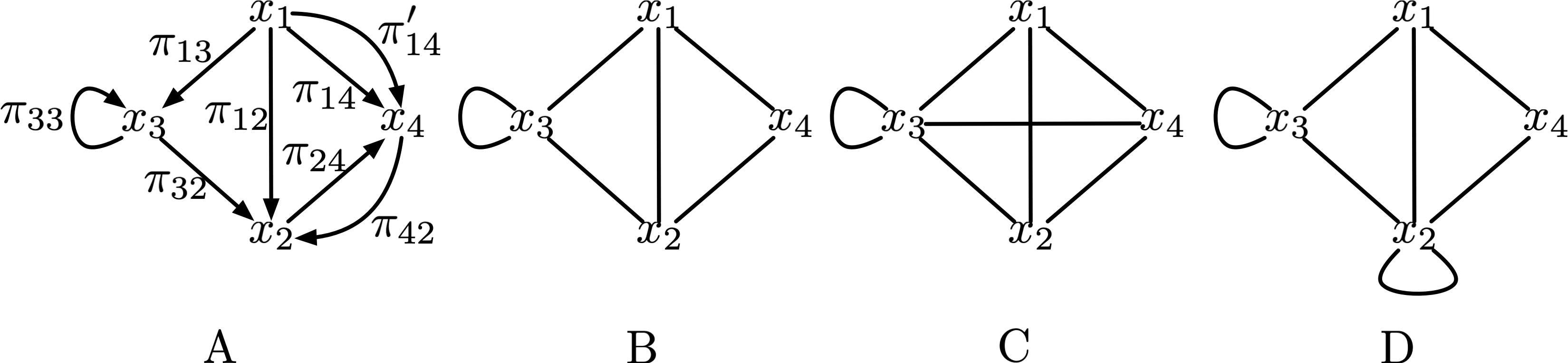}
\else
    \includegraphics[width=.47\textwidth]{img/ugraphs.png}
\fi
\caption{A. graphical representation of the "atoms" contained in the set 
$C=\{
\pi_{13}(x_1,x_3),$ $
\pi_{12}(x_1,x_2),$ $
\pi_{14}(x_1,x_4),$ $
\pi'_{14}(x_1,x_4),$ $
\pi_{33}(x_3,x_3),$ $
\pi_{32}(x_3,x_2),$ $
\pi_{24}(x_2,x_4),$ $
\pi_{42}(x_4,x_2)
\}$; B. $\uGraph{C[x_1,x_2]}$; C. $\uGraph{C[x_3,x_4]}$; D. $\uGraph{C[x_2,x_2]}$. Observe that $\uGraph{C[x_1,x_2]},\uGraph{C[x_2,x_2]}\in\Tw[2]$ but $\uGraph{C[x_3,x_4]}\in\Tw[3]\setminus\Tw[2]$.
}
\label{fig:example_ugraph}
\end{figure}

\AP
We say that an "expression" of $\CPDLg{\+G}$ is ""positive"" if it does not contain any subformula of the form $\lnot \psi$.

\AP
Sometimes we will need to use $\CPDLp$ extended with "program intersection" ($\cap$), that we denote by $\intro*\ICPDLp$. The semantics of $\pi_1\cap\pi_2$ is that of $\ICPDL$. Observe that any expression of $\ICPDL$ is also an expression of $\ICPDLp$ (with no "conjunctive programs") with equal semantics; hence $\ICPDLp$ extends $\ICPDL$.
\AP
Analogously, $\intro*\ICPDLg{\+G}$ is defined as $\CPDLg{\+G}$ extended with "program intersection". 
\begin{proposition}\label{prop:intersecion_does_not_add_expressive_power}
 $\ICPDLp$ and $\CPDLp$ are equivalent in terms of expressive power. Furthermore, for any class $\+G$ containing the graph which consists of an edge between two distinct nodes we have that $\ICPDLg{\+G}$ and $\CPDLg{\+G}$ are also equivalent in terms of expressive power.
\end{proposition}
\begin{proof}
$\pi_1\cap\pi_2$ is "equivalent" to $\set{\pi_1(x,y),\pi_2(x,y)}[x,y]$.
\end{proof}
\AP
For $e$ an $\ICPDLp$ "expression", let $\intro*\subexpr(e)$ be the set of all ""subexpressions"" of $e$ in the usual sense\ifarxiv ~(definition in \S\ref{app:subexpressions}). \else . \fi

\begin{toappendix}
\subsection{Definition of subexpressions}
\label{app:subexpressions}
For $e$ an $\ICPDLp$-"program" or an $\ICPDLp$-"formula", let $\intro*\subexpr(e)$ be the set of all ""subexpressions"" of $e$, that is, the smallest set satisfying the following:
\begin{itemize}
    \item $e \in \subexpr(e)$,
    \item if $\psi \land \psi' \in \subexpr(e)$ then $\set{\psi, \psi'} \subseteq \subexpr(e)$
    \item if $\lnot\psi \in \subexpr(e)$ then $\psi \in \subexpr(e)$
    \item if $\tup{\pi} \in \subexpr(e)$, then $\pi \in \subexpr(e)$
    \item if $\pi \star \pi' \in \subexpr(e)$, then $\set{\pi,\pi'} \subseteq \subexpr(e)$ for every $\star \in \set{\circ,\cup,\cap}$,
    \item if $\pi^* \in \subexpr(e)$, then $\pi \in \subexpr(e)$,
    \item if $\psi? \in \subexpr(e)$, then $\psi \in \subexpr(e)$,
    \item if $C[x_s,x_t] \in \subexpr(e)$ and $\pi(x,y) \in C$, then $\pi \in \subexpr(e)$.
\end{itemize}
\end{toappendix}

% \begin{remark}\label{rk:graph-programs-closed-minors}
%     For every program $\pi$ of the form \eqref{eq:cq-program} and "graph" $G'$ such that $\uGraph{\pi}$ is a "minor" of $G'$ there exists a program $\pi'$ such that $\dbracket{\pi'}_K = \dbracket{\pi}_K$ and $\uGraph{\pi'} = G'$. Indeed, for every deleted edge $(x,x')$ from $G'$, we have an atom $(\bigcup_{a \in \Prog} (a \cup \bar a))^*(x,x')$, for every contracted edge $(x,x')$ of $G'$ we have $\epsilon(x,x')$, and for every other edge we keep the same "atoms".
%     As a consequence, $\CPDLg{\+G}$ is semantically contained in $\CPDLg{\+G'}$ if every graph of $\+G$ is a "minor" of some graph of $\+G'$.
% \end{remark}

% \begin{lemma}[Normal form]\label{lem:normal-form}
%     For every "expression" ("program" or "formula") $e$ of $\CPDLp$ there exists an "equivalent" expression $e'$ such that 
%     \begin{enumerate}
%         \item there is no program ``$\epsilon$'' in $e'$;
%         \item for every "program" $\pi'$ of the form \eqref{eq:cq-program} in $e'$ there is an "equivalent" "program" $\pi$ of $e$ such that $\uGraph{\pi'}$ is the result of removing some (perhaps none) self-loops to $\uGraph{\pi}$ (in particular, $\uGraph{\pi'}$ is a "minor" of $\uGraph{\pi}$).
%     \end{enumerate}
%     Hence, for any "graph" class $\+G$ closed under "minors", if $e \in \CPDLg{\+G}$ then $e' \in \CPDLg{\+G}$.
% \end{lemma}

% \color{black}

%!TEX root = main.tex

\newcommand{\twtcpdl}{\CPDLg{\Tw[2]}}
\newcommand{\twocpdl}{\CPDLg{\Tw[1]}}
\newcommand{\twncpdl}{\CPDLg{\Tw[n]}}

\newcommand{\tr}{{\rm Tr}\xspace}

\section{Relation to $\loopCPDL$ and $\ICPDL$}
\label{sec:relation-to-loop-icpdl}
For logics $\+L_1$ and $\+L_2$ with expressions consisting of "formulas" and "programs", we say that $\+L_2$ is ""at least as expressive as"" $\+L_1$, notated $\+L_1{\intro*\lleq}\+L_2$, 
\phantomintro{\notlleq}%
if there is a translation 
% $\tr$ 
mapping $\+L_1$-"formulas" to $\+L_2$-"formulas" and $\+L_1$-"programs" to $\+L_2$-"programs" which preserves "equivalence".
% such that for any "Kripke structure" $K$ and worlds $u,v\in W(K)$: $u\in\dbracket{\varphi}_K$ iff $u\in\dbracket{\tr(\varphi)}_K$ and $(u,v)\in\dbracket{\pi}_K$ iff $(u,v)\in\dbracket{\tr(\pi)}_K$. 
% Observe that if $\+G\subseteq\+G'$ then $\CPDLg{\+G}\subseteq\CPDLg{\+G'}$.
\AP
We write $\+L_1 \intro*\lleqs \+L_2$ to denote that $\+L_1 \lleq \+L_2$ and $\+L_2 \not\lleq \+L_1$. We write $\+L_1 \intro*\langsemequiv \+L_2$ to denote that $\+L_1 \lleq \+L_2$ and $\+L_2 \lleq \+L_1$, in which case we say that $\+L_1$ and $\+L_2$ are ""equi-expressive"".

\AP
Let $\intro*\Graphloop \eqdef \set{G}$ be the class having just one "graph" $G=(\set{v},\set{\set{v}})$ consisting of just one self-looping "edge", and
\AP
$\intro*\Graphcap \eqdef \set{G}$ be the class having just one "graph" $G = (\set{v,v'}, \set{\set{v,v'}})$  containing one "edge". Observe that $\Graphloop$ and $\Graphcap$ are contained in $\Tw[1]$.

We next show that $\loopCPDL\langsemequiv\CPDLg{\Graphloop}$, and $\ICPDL\langsemequiv\CPDLg{\Graphcap}\langsemequiv\CPDLg{\Tw[1]}\langsemequiv\CPDLg{\Tw[2]}$ via polynomial time translations.

\subsection{$\loopCPDL$ and $\CPDLg{\Graphloop}$ are "equi-expressive"}
The translation from $\loopCPDL$ to $\CPDLg{\Graphloop}$ goes by simply recursively replacing every occurrence of $\lo(\pi)$ such that $\pi$ which does not use loop, with $\{\pi(x,x)\}[x,x]$\ifarxiv ~(details in \Cref{sec:from_loopCPDL}). \else . \fi

The translation in the other direction, from $\CPDLg{\Graphloop}$ to $\loopCPDL$, goes by first observing that every "conjunctive program" is of the form $\set{\pi_1(x,x),\dotsc, \pi_n(x,x)}[x,x]$. If the "program" is  so that no $\pi_i$ contains a "conjunctive program", then it is equivalent to $(\lo(\pi_1) \land \dotsb \land \lo(\pi_n))?$. Hence, a recursive replacement yields an equivalent $\loopCPDL$ "expression"\ifarxiv ~(details in \Cref{sec:to_loopCPDL}).\else . \fi

From these translations we then obtain:
\begin{proposition}\label{prop:loopCPDL-eq-CPDLgLoop}
    $\loopCPDL$ and $\CPDLg{\Graphloop}$ are "equi-expressive", via polynomial time translations.
\end{proposition}

\begin{toappendix}
    
\subsection{Translation from $\loopCPDL$ to $\CPDLg{\Graphloop}$}\label{sec:from_loopCPDL}

We define the following translation $\tr_1$ from $\loopCPDL$-"formulas" to $\CPDLg{\Graphloop}$-"formulas" and from 
$\loopCPDL$-"programs" to $\CPDLg{\Graphloop}$-"programs":
\begin{align}
\tr_1(\star)&=\star&\star\in\Prop\cup\{\epsilon\}\nonumber\\
\tr_1(\varphi\land\psi)&=\tr_1(\varphi)\land\tr_1(\psi)\nonumber\\
\tr_1(\lnot\varphi)&=\lnot\tr_1(\varphi)\nonumber\\
\tr_1(\langle\pi\rangle)&=\langle\tr_1(\pi)\rangle\nonumber\\
\tr_1(\varphi?)&=\tr_1(\varphi)?&\nonumber\\
\tr_1(\pi_1\star\pi_2)&=\tr_1(\pi_1)\star\tr_1(\pi_2)&\star\in\{\cup,\circ\}\nonumber\\
\tr_1(\pi^*)&=\tr_1(\pi)^*\nonumber\\
\tr_1(\lo(\pi))&=\{\tr_1(\pi)(x,x)\}[x,x]\label{eqn:translation_loop}
\end{align}

\begin{proposition}\label{prop:translation_from_loopCPDL}
$\loopCPDL\lleq\CPDLg{\Graphloop}$ via $\tr_1$.
\end{proposition}

\subsection{Translation from $\CPDLg{\Graphloop}$ to $\loopCPDL$}\label{sec:to_loopCPDL}

In $\CPDLg{\Graphloop}$ the only allowed "conjunctive programs" $C[x,y]$ are such that $x=y$ and all "atoms" in $C$ are of the form $\pi(x,x)$, for $\pi$ a $\CPDLg{\Graphloop}$-"program".
% , or of the form $\phi(x)$, for $\phi$ a $\CPDLg{\Graphloop}$-formula.

Consider the following translation $\tr_2$ from $\CPDLg{\Graphloop}$-"formulas" to $\loopCPDL$-"formulas" and from $\CPDLg{\Graphloop}$-"programs" to $\loopCPDL$-"programs": take all clauses of $\tr_1$ from \S\ref{sec:from_loopCPDL} except \eqref{eqn:translation_loop} replacing $\tr_1$ by $\tr_2$ and add
\begin{align*}
\tr_2(C[x,x])&=\left(\bigwedge_{\pi(x,x)\in C}\lo({\tr_2(\pi)})\right)? 
% \land \bigwedge_{\phi(x)\in C}\tr(\phi).
\end{align*}

\begin{proposition}\label{prop:translation_from_loopCPDL}
$\CPDLg{\Graphloop}\lleq\loopCPDL$ via $\tr_2$.
\end{proposition}
\end{toappendix}

\subsection{$\ICPDL$, $\CPDLg{\Tw[1]}$, and $\CPDLg{\Tw[2]}$ are "equi-expressive"}

\paragraph{Translation from $\ICPDL$ to $\CPDLg{\Graphcap}$}
The translation consists of recursively replacing every appearance of $\pi_1\cap\pi_2$ such that $\pi_1,\pi_2$ do not contain $\cap$-"subexpressions", with the equivalent expression $\set{\pi_1(x,y),\pi_2(x,y)}[x,y]$. \ifarxiv Details are in \Cref{sec:from_ICPDL}.\fi

\begin{toappendix}
\subsection{Translation from $\ICPDL$ to $\CPDLg{\Graphcap}$}\label{sec:from_ICPDL}

Consider the following translation $\tr_3$ from $\ICPDL$-"formulas" to $\CPDLg{\Graphcap}$-"formulas" and from $\ICPDL$-"programs" to $\CPDLg{\Graphcap}$-"programs": take all clauses of $\tr_1$ from \S\ref{sec:from_loopCPDL} replacing $\tr_1$ by $\tr_3$ except \eqref{eqn:translation_loop} and add
%
% We define the following translation from ICPDL-formulas to $\CPDLg{\Graphcap}$-formulas and from  ICPDL-programs to $\CPDLg{\Graphcap}$-programs
\begin{align*}
\tr_3(\pi_1\cap\pi_2)&=\{\tr_3(\pi_1)(x,y),\tr_3(\pi_2)(x,y)\}[x,y].
% \label{eqn:translation_intersection}
\end{align*}
Observe that the "underlying graph" of any  $C[x,y]\in\subexpr(\tr_3(\pi))$ is in $\Graphcap$, and hence $\tr_3(\pi)$ is in $\CPDLg{\Graphcap}$.
\end{toappendix}
\begin{propositionrep}\label{prop:translation_from_ICPDL}
$\ICPDL\lleq\CPDLg{\Graphcap}$.% via $\tr_3$.
% If $\varphi$ is an ICPDL-formula then $\tr(\varphi)$ is an $\CPDLg{\Graphcap}$-formula and
%  $u\in\dbracket{\varphi}_K$ iff $u\in\dbracket{\tr(\varphi)}_K$.
% If $\pi$ is an ICPDL-program then then $\tr(\pi)$ is an $\CPDLg{\Graphcap}$-program and $(u,v)\in\dbracket{\pi}_K$ iff $u\in\dbracket{\tr(\pi)}_K$.
\end{propositionrep}

\paragraph{Translation from $\twtcpdl$ to $\ICPDL$}\label{sec:to_ICPDL}
\AP
We start with a technical result. For a graph $G$ and $B\subseteq V(G)$, we say that $B$ ""is a clique in"" $G$ if 
$\set{\set{x,y}\mid x,y\in B,x\neq y}\subseteq E(G)$.
 % For a "tree decomposition" $(T, \bagmap)$ of a graph $G$, we say that a bag $b\in V(T)$ ""is a clique in"" $G$ if $\set{\set{x,y}\mid x,y\in \bagmap(b),x\neq y}\subseteq E(G)$. 
%
The following lemma --whose proof can be found in the appendix-- states that when dealing with $C'[x,y]$ in $\CPDLg{\Tw[n]}$, one may suppose that $\uGraph{C[x_1,x_2]}$ has a "tree decomposition" of "tree-width" $\leq n$ with the additional property that each bag "is a clique in" $\uGraph{C[x_1,x_2]}$. 
\begin{lemmarep}\label{lem:bags_are_cliques}
Given a $\CPDLg{\Tw[n]}$-"conjunctive program" $C'[x,y]$ one can compute in polynomial time a $\CPDLg{\Tw[n]}$-"conjunctive program" $C[x,y]\semequiv C'[x,y]$ with $\vars(C)=\vars(C')$, and a "tree decomposition" $(T, \bagmap)$ of "tree-width" $\leq n$ of $\uGraph{C[x,y]}$ such that for any $b\in V(T)$, $\bagmap(b)$ "is a clique in" $\uGraph{C[x,y]}$.
\end{lemmarep}
\begin{proof}
Let $(T, \bagmap)$ be a "tree decomposition" of "tree-width" at most $n$ of $\uGraph{C'[x_1,x_2]}$. Let $A = \Prog\cap\subexpr(C'[x,y])$ be the set of "atomic programs" in $C'$ or it "subexpressions". Define $C$ as $C'\cup D$, where 
\[
D=\{((\cup_{a \in A} a) \cup (\cup_{a \in A} \bar a))^*(z_1,z_2)\mid b\in V(T), z_1,z_2\in\bagmap(b),z_1\neq z_2\}.
\]
% \diego{Creo que acá hay un problema, porque $D$ no es finito (porque $\Prog$ no es finito). Alcanza con reemplazar $((\cup_{a \in \Prog} a) \cup (\cup_{a \in \Prog} \bar a))^*$ por $((\cup_{a \in P} a) \cup (\cup_{a \in P} \bar a))^*$ donde $P \subseteq \Prog$ es el conjunto de "atomic programs" que aparecen en $C'$?}
% \santi{sí. gracias. Ahí lo arreglé}
Notice that $(T, \bagmap)$ is a "tree decomposition" of $\uGraph{D}$ where every bag is a clique in $\uGraph{D}$. Since $\uGraphC{C'}$ is connected and since the semantics of $\CPDLp$ are compositional and hence indifferent for "atoms" not occurring in the expression, then $C[x,y]\semequiv C'[x,y]$ and it is clear that $\vars(D)\subseteq\vars(C')$. Furthermore, $(T, \bagmap)$ is also a "tree decomposition" of $\uGraph{C[x,y]}$. Since $D$ has at most $|\vars(C')|^2$ "atoms" and computing a "tree decomposition" of "tree-width" at most $n$ can be done in linear time \cite{bodlaender1993linear},\sideedwin{por alguna razón esta cita no siempre compila} then the whole construction remains in polynomial time.
\end{proof}

\begin{toappendix}
For an $\ICPDL$-"program" $\pi$ we define $\pi^{-1}$ as follows:
\begin{align*}
\star^{-1}&=\star&\star\in\{\epsilon,\phi?\}\\
a^{-1}& = \overline a &a\in\A\\
\overline a^{-1}& =a  &a\in\A\\
(\pi_1\star\pi_2)^{-1}&=\pi_2^{-1}\star\pi_1^{-1}&\star\in\{\cap,\cup,\circ\}\\
{\pi^*}^{-1}&={\pi^{-1}}^*
\end{align*}
\begin{proposition}
$\pi$ is an $\ICPDL$-"program" iff $\pi^{-1}$ is an $\ICPDL$-"program", and 
$(u,v)\in\dbracket{\pi}_K$ iff $(v,u)\in\dbracket{\pi^{-1}}_K$.
\end{proposition}
\end{toappendix}

The next lemma is the key ingredient of the translation from $\twtcpdl$ to $\ICPDL$: 
\begin{lemmarep}\label{lemita}
Let $C$ be a finite set of "atoms" of the form $\pi(z_1,z_2)$, 
where $\pi$ is a \ICPDL-"program" and $z_1,z_2 \in \Vars$, 
$\uGraphC{C}$ is a clique, and $|\vars(C)|\leq 3$. Then:
\begin{enumerate}
\item For any $x,y\in\vars(C)$, $x\neq y$, there is an $\ICPDL$-"program" $\pi_{C[x,y]}$ such that $\pi_{C[x,y]}\semequiv C[x,y]$.

\item For any $x\in\vars(C)$ there is an $\ICPDL$-"program" $\pi_{C[x,x]}$ such that $\pi_{C[x,x]}\semequiv C[x,x]$.
\end{enumerate}
Further, these translations are in polynomial time.
\end{lemmarep}
\begin{proof}
    For $z_1,z_2\in\vars(C)$, $z_1\neq z_2$, let
    \begin{align*}
    \Pi_{z_1z_2}&=\left(\bigcap_{\pi(z_1,z_2)\in C}\pi\right) \cap \left(\bigcap_{\pi(z_2,z_1)\in C} \pi^{-1}\right)\\
    % \Pi_{z_1}&=\epsilon\cap\bigcap_{\pi(z_1,z_1)\in C}\pi
    \Pi_{z_1}&=\left(\bigwedge_{\pi(z_1,z_1)\in C} \tup{\pi\cap\epsilon}\right)?
    \end{align*}
    % An intersection with empty range is defined by $((\cup_{a \in \Prog} a) \cup (\cup_{a \in \Prog} \bar a))^*$. 
    where a conjunction with empty range is defined by $\tup{\epsilon}$.
    
    For $x\neq y$ define the $\ICPDL$-"programs" $\pi_{C[x,y]}$ and $\pi_{C[x,x]}$ as follows:
    \begin{itemize}
    \item If $\vars(C)=\{x,y,z\}$ with $x\neq z \neq y$ then
    \begin{align*}
    \pi_{C[x,y]}&=\Pi_x\circ( \Pi_{xy}\cap(\Pi_{xz}\circ\Pi_z\circ\Pi_{zy}))\circ\Pi_y\\
    \pi_{C[x,x]}&=\tup{\pi_{C[x,y]}}?
    \end{align*}
    % \begin{align*}
    % \pi_{C[x,y]}&=\Phi_x?\circ( \Pi_{xy}\cap(\Pi_{xz}\circ\Phi_z?\circ\Pi_{zy}))\circ\Phi_y?\\
    % % \phi&=\tup{\pi_1}\\
    % \pi_{C[x,x]}&=\pi_{C[x,y]}\cap\epsilon
    % \end{align*}
    
    \item If $\vars(C)=\{x,y\}$ then
    \begin{align*}
    \pi_{C[x,y]}&=\Pi_x\circ \Pi_{xy} \circ\Pi_y\\
    \pi_{C[x,x]}&=\tup{\pi_{C[x,y]}}?
    \end{align*}
    % \begin{align*}
    % \pi_{C[x,y]}&=\Phi_x?\circ \Pi_{xy} \circ\Phi_y?\\
    % \pi_{C[x,x]}&=\pi_{C[x,y]}\cap\epsilon
    % \end{align*}
    
    \item If $\vars(C)=\{x\}$ then
    \begin{align*}
    \pi_{C[x,x]}&=\Pi_x
    \end{align*}
    \end{itemize}
    It can be shown that $\pi_{C[x,y]}\semequiv C[x,y]$ and $\pi_{C[x,x]}\semequiv C[x,x]$.
    \end{proof}

Figure \ref{fig:lemita} illustrates an example of $\pi_{C[x,y]}$ and of $\pi_{C[x,x]}$ of Lemma \ref{lemita}. The general case is more involved, as $C$ may contain --as illustrated in Figure \ref{fig:example_ugraph}.A-- several "atoms" $\pi_1(z_1,z_2),\dots\pi_n(z_1,z_2)$ for the same variables $z_1,z_2$ and also contain "atoms" of the form $\pi(z_1,z_1)$ not shown in the example. Furthermore, one would need sometimes to `reverse' the direction of the "programs" to obtain the desired $\ICPDL$-"program". The proof of \Cref{lemita} is deferred to the appendix.

Using \Cref{lemita} one can show the translation which, for technical reasons as it simplifies the proof, translates $\twticpdl$ "expressions" (instead of $\twtcpdl$) to $\ICPDL$.
\begin{figure}
\ifarxiv
    \includegraphics[width=.7\textwidth]{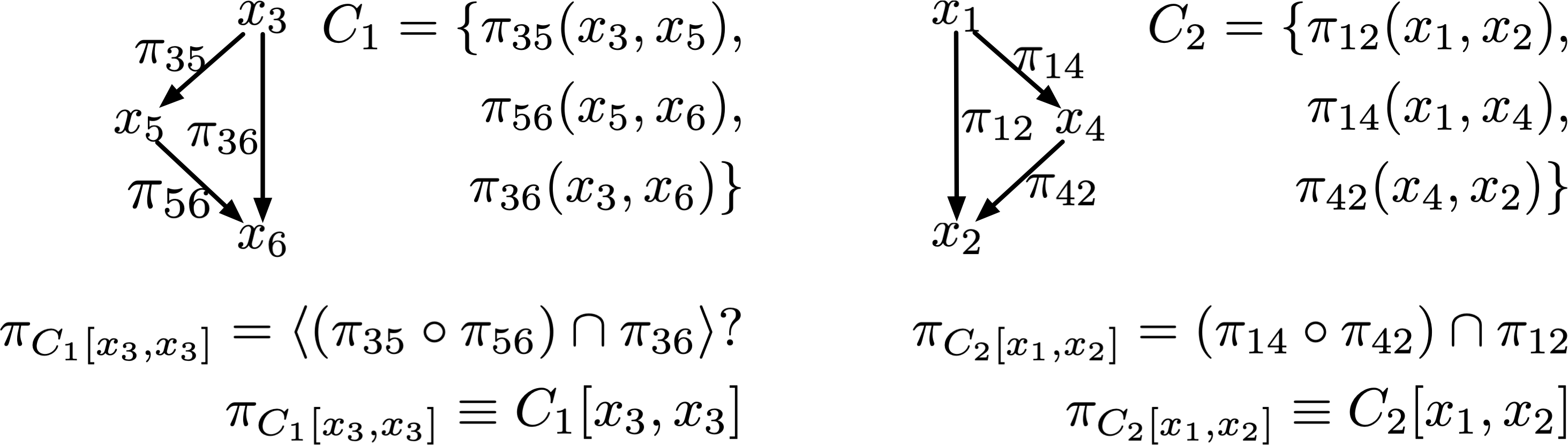}
\else
    \includegraphics[width=.47\textwidth]{img/lemita.png}
\fi
\caption{An example of translation of a "conjunctive program" made up of only $\ICPDL$ "atoms" and 3 variables to \ICPDL. The name of the $\ICPDL$-"programs" $\pi_{C[x_3,x_3]}$ and $\pi_{C[x_1,x_2]}$ are the one used in Lemma \ref{lemita}.}
\label{fig:lemita}
\end{figure}

\begin{figure}
\ifarxiv
    \includegraphics[scale=0.25]{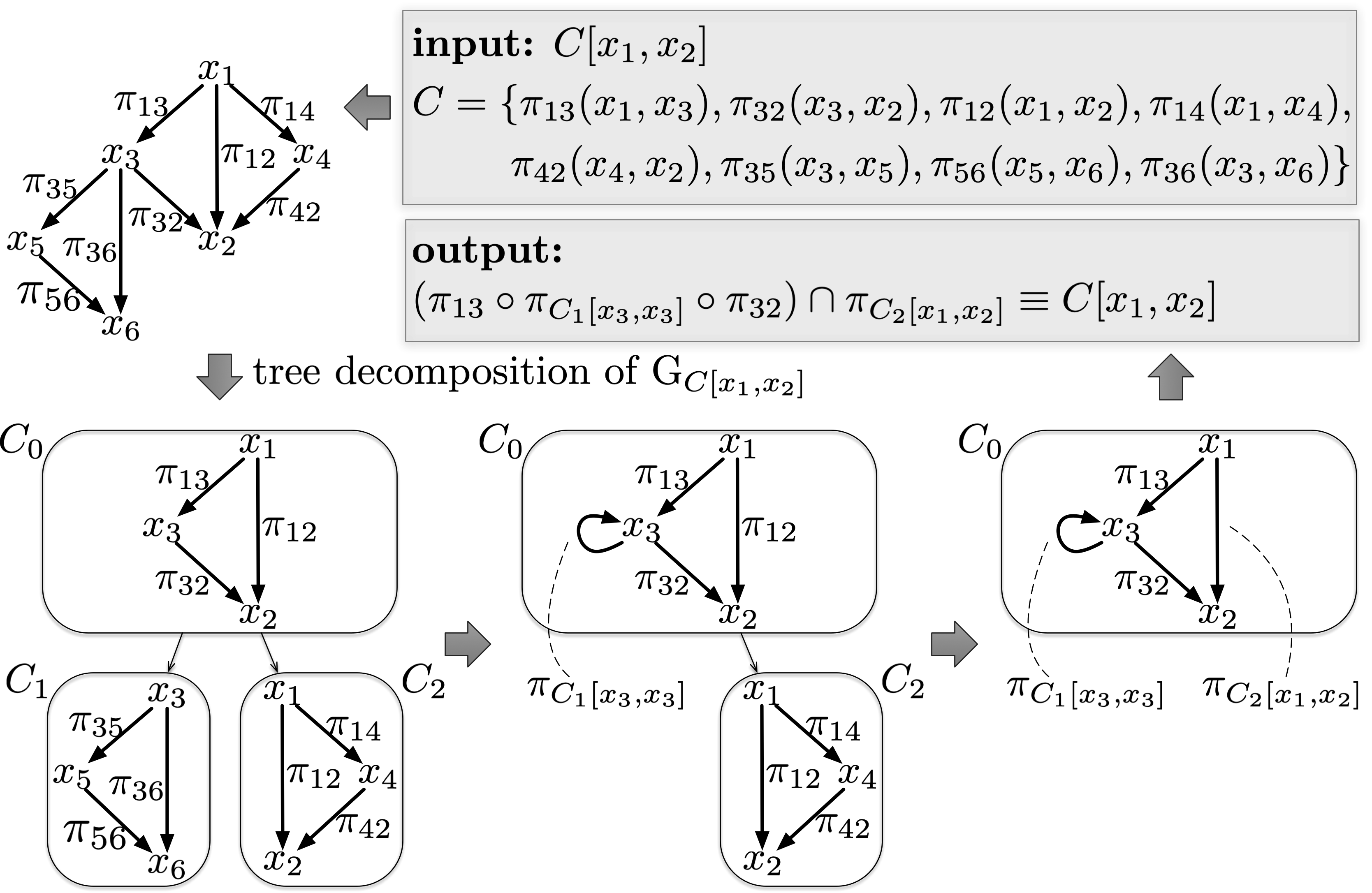}
\else
    \includegraphics[scale=0.20]{img/example_treedecomp.png}
\fi
\caption{A "conjunctive program" $C[x_1,x_2]$ made up of only $\ICPDL$ "programs" and the steps of the algorithm to construct a semantically equivalent $\ICPDL$ formula from a "tree decomposition" of "tree-width" 2 of $\uGraph{C[x_1,x_2]}$. Rounded boxes represent the bags of such "tree decomposition" where variables $x_1,\dots,x_6$ are grouped (additional information such as edges and labels are not part of the "tree decomposition"). See Figure~\ref{fig:lemita} for the definition of $\pi_{C[x_3,x_3]}$ and $\pi_{C[x_1,x_2]}$.}
\label{fig:example_treedecomp}
\end{figure}

\begin{propositionrep}\label{prop:translation_from_TW2}
$\twticpdl\lleq\ICPDL$ via a polynomial time translation.
\end{propositionrep}
\begin{proof}
    % For a $\twticpdl$ "expression" $e$, let $\sub(e)$ be the set of all the sub-"expressions" ("programs" or "formulas") of $e$, 
    Let $c_1(e)=\bigcup_{C[x,y]\in\subexpr(e)}\vars(C)$, $c_2(e)$ be the set of expressions in $\subexpr(e)$ which are not $\ICPDL$-expressions, and $c_3(e)$ be the syntactic complexity of $e$. Finally, let $c(e)=(|c_1(e)|+|c_2(e)|,c_3(e))$. We define a translation $\tr(e)$ from $\twticpdl$-"formulas" to $\ICPDL$-"formulas" and from $\twticpdl$-"programs" to $\ICPDL$-"programs" by recursion in the lexicographic order of $c$.

    For the cases other than "conjunctive programs", $\tr$ is defined as $\tr_1$ from \S\ref{sec:from_loopCPDL} replacing $\tr_1$ by $\tr$, and $\star\in\{\cup,\circ\}$ by $\star\in\{\cup,\circ,\cap\}$. The second component of $c$  ensures that if $\tr(e)=\dots\tr(e')\dots$ according to one of these rules, then $c(e')<c(e)$. 
    
    For the case of a "conjunctive program" $C[x,y]$, suppose
    \begin{align*}
    C&=\set{\pi_1[x_1,y_1],\dots,\pi_n[x_n,y_n]}.
    \end{align*}
    be such that $\uGraph{C[x,y]}\in\Tw[2]$.
    If all $\pi_i$ are $\ICPDL$-"programs" and $|\vars(C)|\leq3$, then $C$ satisfies the hypothesis of Lemma \ref{lemita}. In case $x=y$, there is an $\ICPDL$-"program" $\pi_{C[x,x]}\semequiv C[x,x]$ and we define $\tr(C[x,y])=\pi_{C[x,x]}$. In case $x\neq y$, there is an $\ICPDL$-"program" $\pi_{C[x,y]}\semequiv C[x,y]$ and we define $\tr(C[x,y])=\pi_{C[x,y]}$.
    
    Else there is some $\pi_i$ that is not a $\ICPDL$-"program" or $|\vars(C)|>3$. We define $\tr(C[x,y])=\tr(D[x,y])$ for a set of "atoms" $D$ to be defined next.
    
    \begin{enumerate}
    \item\label{trad:case:1} If there is $i$ such that $\pi_i$ is not a $\ICPDL$-"program", then $D=C\setminus\{\pi_i[x_i,y_i]\}\cup\{\tr(\pi_i)[x_i,y_i]\}$. It is straightforward that $\tr(C[x,y])\semequiv\tr(D[x,y])$. Observe that in this case 
    (1) $c_1(D[x,y])\subseteq c_1(C[x,y])$, since by construction we do not add new variables to $D$; and
    (2) $c_2(D[x,y])\subsetneq c_2(C[x,y])$ since in $D$ we removed at least one expression which was not an $\ICPDL$-"program" and we do not add new ones.
    Hence $c(D[x,y])<c(C[x,y])$.

    \item\label{trad:case:2} Else, all $\pi_i$ are already $\ICPDL$-"programs" and $|\vars(C)|>3$. Let $(T, \bagmap)$ be a "tree decomposition" of the "underlying graph" of $C$ of root containing $x$ and $y$ (observe that this is possible by the fact that $\uGraph{C[x,y]}$ contains the "edge" $\{x,y\}$), and of width $\leq 2$ such that any leaf of $T$ is a bag of size at least 2 (and at most 3), and is not included in its parent bag. If $T$ is a single bag $B$, then we are done: in case $x=y$, define $\tr(C[x,y])=\pi_{C[x,x]}$ and otherwise define $\tr(C[x,y])=\pi_{C[x,y]}$.
    If $T$ has at least two bags, pick any leaf $B$ from $T$ and let $B'$ be the parent of $B$ in $T$. Let $C_B$ be the set of all "atoms" $A\in C$ such that $\vars(A)\subseteq B$. Observe that $C_B$ satisfies the hypothesis of Lemma \ref{lemita}.
    \begin{itemize}
        \item  Suppose that $B'\cap B=\{z_1\}$ and $B\supseteq\{z_1,z_2\}$ ($z_1\neq z_2$). We define $D$ as the set of all "atoms" in $C$ except those 
        in $C_B$
        % of the form $\pi_i(y,z)$ or $\pi_i(z,y)$ for $y\in B\setminus B'$, $z\in\Vars$,
        plus the \ICPDL-"program" $\pi_{C_B[z_1,z_1]}\semequiv C_B[z_1,z_1]$ from Lemma \ref{lemita}. Notice that $z_2\in \vars(C)\setminus\vars(D)$.
    
        \item  Suppose that $B'\cap B=\{z_1,z_2\}$ ($z_1\neq z_2$) and $B=\{z_1,z_2,z_3\}$ ($z_1\neq z_3\neq z_2$). We define $D$ as the set of all "atoms" in $C$ except those 
        in $C_B$
        % of the form $\pi_i(z_3,z)$ or $\pi_i(z,z_3)$ for $z\in\Vars$ 
        plus the \ICPDL-"program" $\pi_{C_B[z_1,z_2]} \semequiv C_B[z_1,z_2]$ from Lemma \ref{lemita}. Notice that $z_3\in \vars(C)\setminus\vars(D)$.
    \end{itemize}
    It is straightforward that $\tr(C[x,y])\semequiv\tr(D[x,y])$. Furthermore, observe that 
    $c_1(D[x,y])\subsetneq c_1(C[x,y])$, and
    $c_2(D[x,y])\subseteq c_2(C[x,y])$.
    Hence $c(D[x,y])<c(C[x,y])$.
    \end{enumerate}
    This concludes the proof. We are only left with the task of showing that the translation is in polynomial time, which we defer to \Cref{prop:poly_translation} below.
\end{proof}

Figure \ref{fig:example_treedecomp} illustrates an example of a "conjunctive program" $C[x_1,x_2]$ 
where $C$ contains only "atoms" of the form $\pi(z_1,z_2)$ for $\pi$ an $\ICPDL$-"program". 
Given a "tree decomposition" $(T, \bagmap)$ of $\uGraph{C[x_1,x_2]}$ of "tree-width" at most 2 whose all bags are cliques in $\uGraph{C[x_1,x_2]}$ (this can be assumed without loss of generality by Lemma \ref{lem:bags_are_cliques}), one can successively remove its leaves until a single bag is obtained. Each time a leaf $B$ is removed, all "atoms" in $C$ using (only) variables from $B$ are removed, but --thanks to Lemma \ref{lemita}-- an $\ICPDL$-"program" with the information of $B$ is added instead. The general scenario is again more complex, since $C$ may contain "programs" $\pi$ which are not $\ICPDL$-"programs" (namely, "conjunctive programs"). In this case, we need to recursively eliminate such $\pi$s starting with those containing only $\ICPDL$-"programs". The proof of \Cref{prop:translation_from_TW2} is deferred to the appendix.

As a consequence of \Cref{prop:intersecion_does_not_add_expressive_power,prop:translation_from_ICPDL,prop:translation_from_TW2} and the fact that $\Graphcap\subseteq\Tw[1]$, we obtain the following corollary:
\begin{theorem}\label{thm:ICPDL_equals_TW1_equals_TW2}
$\ICPDL$, $\CPDLg{\Tw[1]}$ and $\CPDLg{\Tw[2]}$ are "equi-expressive", via polynomial time translations.
\end{theorem}

\begin{toappendix}
    
\begin{proposition}\label{prop:poly_translation}
The translation from $\twticpdl$ to $\ICPDL$ can be computed in polynomial time.
\end{proposition}
\begin{proof}
We first analyse the following algorithm \ref{algo_particular} emerging from case \ref{trad:case:2} of the proof of \Cref{prop:translation_from_TW2} that translates a $\twtcpdl$-"conjunctive program" of the form $C[x,y]$, where $C$ is a set of "atoms" of the form $\pi(z_1,z_2)$, for $\pi$ a \ICPDL-"program" and $z_1,z_2 \in \Vars$, into a semantically equivalent \ICPDL-"program".

\begin{algorithm}
\SetKwData{Left}{left}\SetKwData{This}{this}\SetKwData{Up}{up}
\SetKwFunction{Union}{Union}\SetKwFunction{FindCompress}{FindCompress}
\SetKwInOut{Input}{input}\SetKwInOut{Output}{output}
\Input{A $\twtcpdl$-"program" $C'[x,y]$, where $C'$ is a 
set of "atoms" of the form $\pi(z_1,z_2)$, 
for $\pi$ a \ICPDL-"program" and $z_1,z_2 \in \Vars$}
\Output{An \ICPDL-"program" $\pi$ such that $\pi\semequiv C'[x,y]$}
\BlankLine

Let $C[x,y]\semequiv C'[x,y]$ with $\vars(C)=\vars(C')$ and let $(T, \bagmap)$ be a "tree decomposition" of "tree-width" at most $2$ of $\uGraph{C[x,y]}$ such that any $b\in V(T)$, $\bagmap(b)$ "is a clique in" $\uGraph{C[x,y]}$ (Lemma \ref{lem:bags_are_cliques}).

% such that any leaf of $T$ is a bag of size at least 2 (and at most 3), and is not included in its parent bag.

\While{$T$ is not a single bag}
    {
    let $B$ be a leaf of $T$\; 

    let $B'$ be the parent of $B$ in $T$\;

    $C_B:=\{\pi(z_1,z_2)\in C\mid \{z_1,z_2\}\subseteq B\}$\;

    \If{$B'\cap B=\{z_1\}$ and $B\setminus B'\neq\emptyset$}
    {

    % $R:=\{\pi(y,z),\pi(z,y)\in C \mid y\in B\setminus B',z\in\Vars\}$\;
    % $R:=\{\pi(y,z)\in C \mid y,z\in B\}$\;

    % $C:=C\setminus R \cup\{\pi_{C_B[z_1,z_1]}\}$\;
    $C:=C\setminus C_B \cup\{\pi_{C_B[z_1,z_1]}\}$\;

    }
    \If{$B'\cap B=\{z_1,z_2\}$ ($z_1\neq z_2$) and $B=\{z_1,z_2,z_3\}$ ($z_1\neq z_3\neq z_2$)}
    {
    % $R:=\{\pi(z_3,z),\pi(z,z_3)\in C\mid z\in\Vars\}$\;
    % $R:=\{\pi(y,z)\in C \mid y,z\in B\}$\;

    % $C:=C\setminus R\cup\{\pi_{C_B[z_1,z_2]}\}$\;
    $C:=C\setminus C_B\cup\{\pi_{C_B[z_1,z_2]}\}$\;
    }
    Remove $B$ from $T$
    }
    {

\eIf{$x\neq y$}
    {\Return $\pi_{C[x,y]}$}
    {\Return $\pi_{C[x,x]}$}
}

\caption{Computing the translation from $\twticpdl$ to $\ICPDL$ in a special case}\label{algo_particular}
\end{algorithm}

Observe that the first line can be done in polynomial time by Lemma \ref{lem:bags_are_cliques}, and that $|C|=|C'|+O(|\vars(C')^2|)$.

For $\pi$ an \ICPDL-"program", let $c(\pi)$ be the syntactic complexity of $\pi$ and for a set $D$ of "atoms" of the form $\pi(z,z')$, where $\pi$ is a \ICPDL-"program", let $c(D)=\sum_{\pi(z,z')\in D}c(\pi)$.

The construction of the \ICPDL-"programs" $\pi_{C_B[z_1,z_1]}$ and $\pi_{C_B[z_1,z_2]}$ from Lemma \ref{lemita} can be done in time $O(c(C_B))$. Observe that the above algorithm removes from $C$ the set of "atoms" in $C_B$ but adds a new \ICPDL-"program". We check that the syntactic complexity of this "program" does not grow too much along the iterations of the cycle.
By inspecting the definition of $\pi_{C_B[z_1,z_1]}$ and $\pi_{C_B[z_1,z_2]}$ from Lemma \ref{lemita}, one can see that there is a constant $k$ such that for any $B$ and $C_B$ in any iteration of the cycle we have 
$
c(\pi_{C_B[z_1,z_1]}),c(\pi_{C_B[z_1,z_2]})\leq k|C_B|+c(C_B).
$
Hence the calculation of $\pi_{C_B[z_1,z_1]}$ and $\pi_{C_B[z_1,z_2]}$ can be done in time $O(k|C|+c(C))$; also the output $\pi$ of algorithm \ref{algo_particular} satisfies $c(\pi)\leq k|C|+c(C)$.

As mentioned in the proof of Lemma \ref{lem:bags_are_cliques}, the "tree decomposition" $T$ of "tree-width" at most $2$ can be computed in polynomial time and hence there are polynomially many bags in $T$. Since at each step of the cycle one bag of $T$ is removed, there are polynomially many iterations. Hence algorithm \ref{algo_particular} runs in polynomial time.

We next see how to use algorithm \ref{algo_particular} to compute the translation of any $\twtcpdl$-"program" into a semantically equivalent \ICPDL-"program". Given a $\twtcpdl$-"program" $\pi$ we proceed as follows. If there is no $C[x,y]\in\subexpr(\pi)$, then $\pi$ is already an \ICPDL-"program" and we are done. Else, pick $C[x,y]\in\subexpr(\pi)$ where $C$ is a set of \ICPDL-"programs". Apply algorithm \ref{algo_particular} to obtain a \ICPDL-"program" $\pi$ that is semantically equivalent to $C[x,y]$. Replace $C[x,y]$ by $\pi'$ in $\pi$ and repeat the procedure. This algorithm is clearly polynomial and computes the desired \ICPDL-"program".
\end{proof}

\end{toappendix}
%!TEX root = main.tex

\section{Indistinguishability}\label{sec:caracterization}
In this section we characterize the expressive power of $\CPDLg{\Tw}$ via a restricted form of (bi)simulation using pebbles, in what resembles the $k$-pebble game for characterizing finite variable first-order fragments.

\subsection{Simulation relation}
\AP
We will define the notion of "$k$-simulation" between pairs $(K,K')$ of "Kripke structures" via a two-player zero-sum graph game $\intro*\kSimGame$. 
\AP
The arena of the game has a set of positions $S \cup D$, %$S \cup D \cup \set{\bot_s,\bot_d}$, 
where
\begin{align*}
    S &= \set{s} \times \kHoms(K,K')\\
    D &= \set{d_1,\dotsc, d_k} \times (\worlds{K}^k \times \worlds{K'}^k)
\end{align*}
where Spoiler owns all positions from $S$ and Duplicator all positions from $D$.
The set of moves of $\kSimGame$ is the smallest set satisfying the following:\footnote{For ease of notation we write $(s,\bar u,\bar v)$ instead of $(s,(\bar u,\bar v))$ and the same for $(d_i,\bar u,\bar v)$.}
\begin{enumerate}
    \item There is a move from $(s,\bar u, \bar v)$ to $(d_i,\bar u',\bar v)$ 
    if $\bar u' = \bar u[i \mapcoord w]$, where $w$ is a "world" from $K$ at "distance" $\leq 1$ from $\bar u[j]$, for some $1 \leq j \leq k$ with $i\neq j$; and 
    % $\bar u'$ is  the result of replacing component $i$ from $\bar u$ with a "world" at "distance" $\leq 1$ from some $\bar u[j]$ in $K$ (with $j \neq i$); and
    \item There is a move from $(d_i,\bar u',\bar v)$ to $(s,\bar u',\bar v')$ 
    if $\bar v' = \bar v[i \mapcoord w]$, where $w$ is a "world" from $K'$ at "distance" $\leq 1$ from $\bar v[j]$, for some $1 \leq j \leq k$ with $i\neq j$.
\end{enumerate}
The winning condition for Duplicator is just any infinite play, which is a form of ``Safety condition'', which implies (positional) determinacy of the game.\footnote{These games are often presented as having infinite duration. Note that Spoiler is never `stuck' as he can always play by putting in component 1 the "world" contained in component 2, for example. However, Duplicator can get stuck.
\newcommand{\spoilerWin}{(s,\textit{Win})}%
The infinite duration variant would be adding a new position `$\spoilerWin$' owned by Spoiler, and having moves from $\spoilerWin$ to $\spoilerWin$, and from every position of $D$ to $\spoilerWin$. In this way, the winning condition of Duplicator are all plays which avoid going through position $\spoilerWin$. Since it is a Safety (hence parity) condition, it follows that the game is positionally determined \cite{DBLP:conf/focs/EmersonJ91}.}
For more information on this kind of games we refer the reader to~\cite[Chapter~2]{gradel2003automata}.

\AP
Given two "Kripke structures" $K, K'$, and tuples $\bar v\in W(K)^k$ and $\bar v'\in W(K')^k$, we say that $K',\bar v'$ ""$k$-simulates"" $K,\bar v$, notated $K,\bar v \intro*\pebblesim k K',\bar v'$, \phantomintro{\notpebblesim} if 1) $(s,\bar v, \bar v')$ is a valid position of $\kSimGame$ on $(K,K')$ (\ie, $\bar v, \bar v'$ induce a "partial homomorphism"), 2) all the "worlds" in $\bar v$ are in the same connected component of $K$, and 3) Duplicator has a winning strategy from $(s,\bar v, \bar v')$. See \Cref{fig:game} for an example. 
\begin{remark}\label{rem:connected}
Observe that any move of $\kSimGame$ preserves property 2) above.
\end{remark}

\begin{figure*}
\ifarxiv
    \includegraphics[width=1\textwidth]{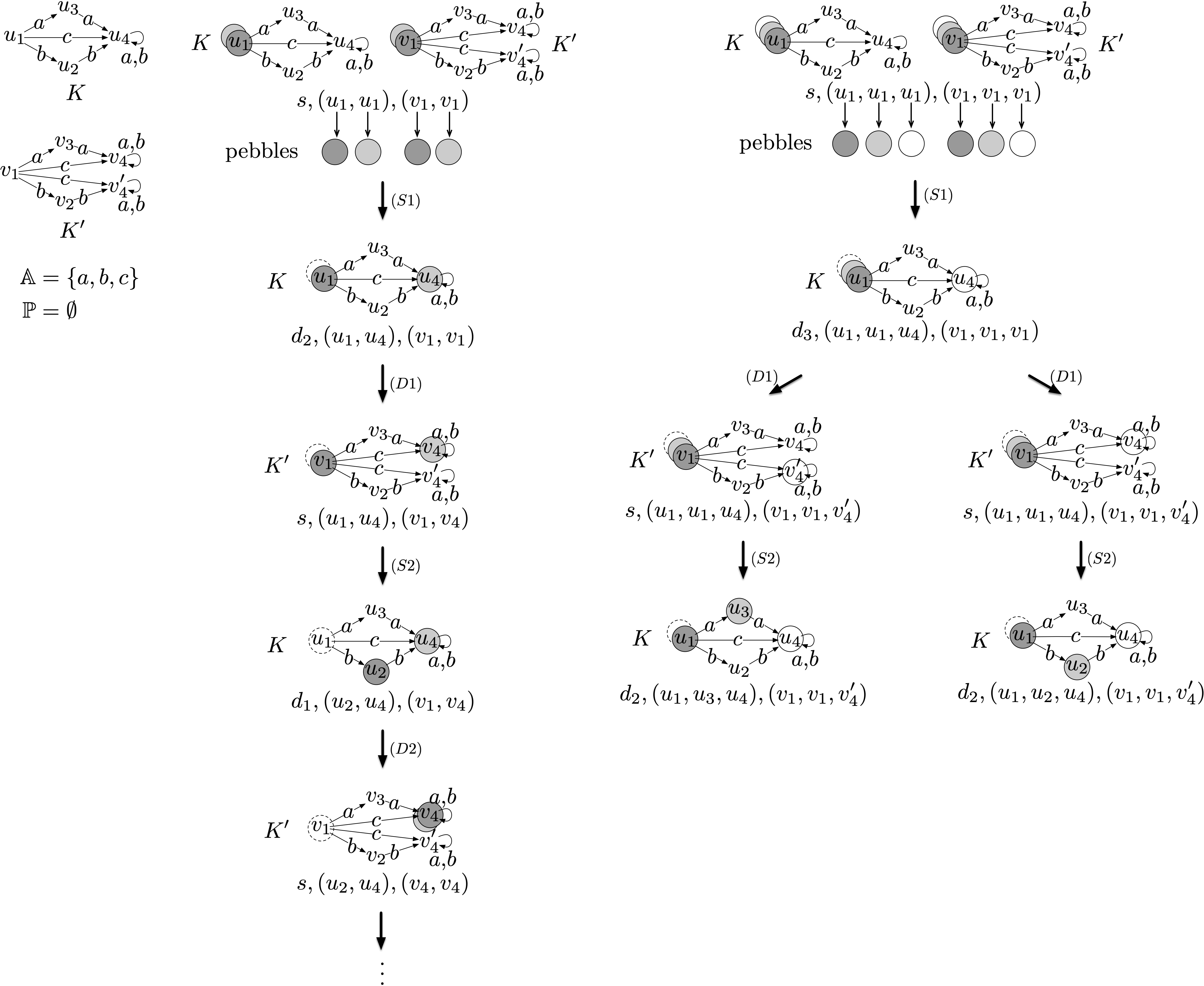}
\else
    \includegraphics[width=1\textwidth]{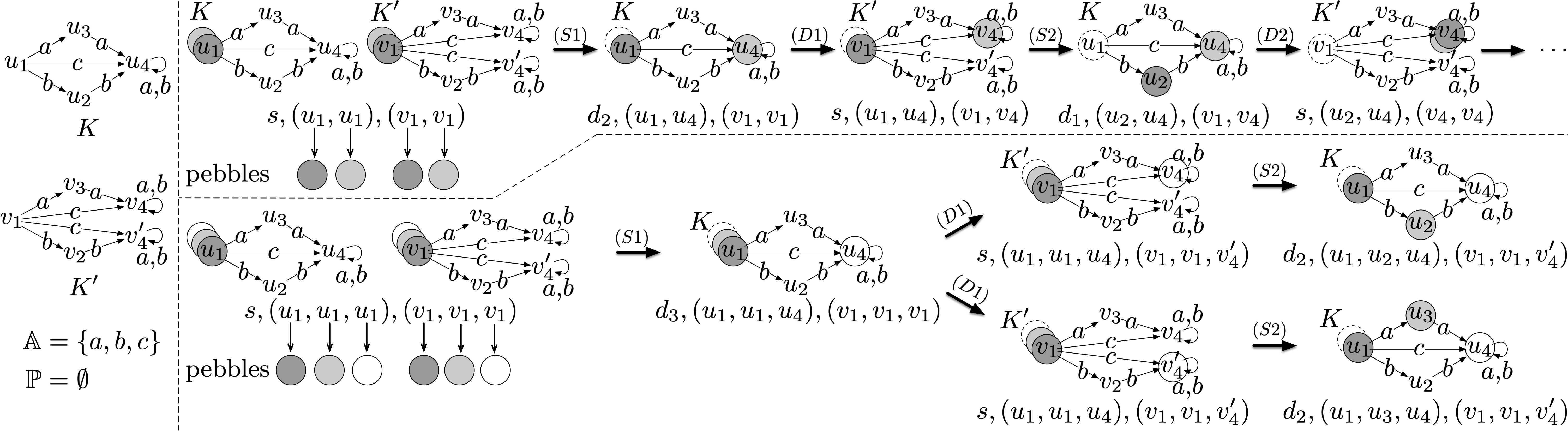}
\fi

    \caption{
    Left: Kripke structures $K$ and $K'$. 
    \ifarxiv
        Center: 
    \else
        Top:
    \fi
    an example of a game $\kSimGame[2]$ over $K$ and $K'$. The order of moves are: (S1), (D1), (S2), (D2). In fact, Duplicator can always answer to any Spoiler's move in $\kSimGame[2]$, and so $K,u_1 \pebblesim{2} K',v_1$. 
    \ifarxiv
        Right: 
    \else
        Bottom:
    \fi
    a tree representation of a winning strategy for Spoiler in $\kSimGame[3]$ over $K$ and $K'$. All possible answers by Duplicator are shown. Duplicator cannot answer after (S2) and so $K,u_1 \not\pebblesim{3} K',v_1$. Observe that $K,u_1,u_1\models\pi$ and $K,v_1,v_1\not\models\pi$ for $\pi=\{c(x_s,x_4),a(x_s,x_3),a(x_3,x_4),b(x_s,x_2),b(x_2,x_4)\}[x_s,x_s]$. Graphically the vectors in the positions can be thought as moving color pebbles on the "Kripke structures". Only "Kripke structures" that change pebble positions with respect to the previous move are shown.}

% Example of two games $\kSimGame$ for $k=2$ and $k=3$ over Kripke structures $K$ and $K'$. Both the formal definition and the intuitive pebble interpretation are shown. The order of moves are: (S1), (D1), (S2), (D2). Duplicator can always answer to any Spoiler's for $\kSimGame[2]$, showing that $K,u_1 \pebblesim{2} K',v_1$. On the contrary, Spoiler wins in two steps the game $\kSimGame[3]$ and hence $K,u_1 \not\pebblesim{3} K',v_1$. Observe that $K,u_1\models\varphi$ and $K,v_1\not\models\varphi$ for $\varphi=\tup{\{a(x_1,x_2),a(x_2,x_3),a(x_3,x_1)\}[x_1,x_2]}$.
\label{fig:game}
\end{figure*}

\begin{theorem}\label{thm:sim-char}
    Let $k\geq 2$.
    Given "Kripke structures" $K,K'$ where $K'$ is of "finite degree" and "worlds" $u,v\in\worlds{K}$ and $u',v'\in\worlds{K'}$, the following are equivalent
    \begin{enumerate}
        \item for every "positive" $\CPDLg{\Tw[k]}$-"formula" $\phi$, we have $K,v \models \phi$ implies $K',v' \models \phi$; and
        \item $K,v \pebblesim{k+1} K',v'$;
    \end{enumerate}
    and the following are equivalent
    \begin{enumerate}
        \item for every "positive" $\CPDLg{\Tw[k]}$-"program" $\pi$, we have $K,u,v \models \pi$ implies $K',u',v' \models \pi$; and
        \item $K,u,v \pebblesim{k+1} K',u',v'$.
    \end{enumerate}
Furthermore, the hypothesis of "finite degree" is only needed for the $1$-to-$2$ implications.
\end{theorem}
\begin{proof}
    
    %!TEX root = main.tex

\proofcase{1-to-2}
We first show the 1-to-2 implication.
    We remark first that the $k$-simulation relation is closely related to the  ``existential $k$-pebble game'' \cite{DBLP:journals/jcss/KolaitisV95,DBLP:journals/jcss/KolaitisV00} on first-order structures with binary and unary signatures. It is known \cite[Theorem~4]{DBLP:conf/cp/AtseriasKV04} that for every pair $K,K'$ of \emph{finite} "structures@@kripke", if for every Boolean "conjunctive query" $q$ of "tree-width" $<k$ we have that $K \models q$ implies that $K' \models q$, then Duplicator wins the existential $k$-pebble game on $K,K'$. Since "conjunctive programs" are basically "conjunctive queries", it is not surprising that one can adapt this result to our setting.

    By contrapositive, assume $K,u,v \notpebblesim{k+1} K',u',v'$. By ease of notation we write $(w_1,w_2)$ instead of the $k$-tuple $(w_1,w_2,\dots,w_2)$. 
    The cases where $((u,v), (u',v')) \not\in \kHoms(K,K')$ or where $u$ and $v$ belong to different connected components of $K$ are trivial, so let us assume the contrary.
    This means that Spoiler has a strategy to win in a bounded number of rounds in $\kSimGame[k+1]$ starting from $(s,(u,v),(u',v'))$. 
    Consider the winning strategy of Spoiler, described by a "tree" of finite height whose vertices are labelled with positions from $\kSimGame[k+1]$. In particular: (i) the root must be labeled with $(s,(u,v),(u',v'))$, (ii) any vertex labeled $(d_i,\bar u, \bar v)$ has a child labeled $(s,\bar u, \bar v')$ for every possible move of $\kSimGame[k+1]$, (iii) any vertex labeled $(s,\bar u, \bar v)$ has exactly one child 
    $(s,\bar u', \bar v)$,
    whose label is consistent with a move of $\kSimGame[k+1]$,
    (iv) for every leaf (labeled $(d_i,\bar u, \bar v)$) there is no possible move from that position in $\kSimGame[k+1]$.
    Observe that, since $K'$ is of "finite degree", there is a finite number of moves departing from any position of $\kSimGame[k+1]$ owned by Duplicator, and thus the branching of the "tree" is finite. Therefore, the strategy "tree" is finite.

    From such a winning strategy for Spoiler, consider the "tree" $T$ and labeling $\lambda: V(T) \to \worlds{K}^{k+1}$ resulting from: 
     (1) removing all vertices $y$ labeled by Spoiler positions except the root (\ie, the initial configuration, owned by Spoiler), and adding an edge between the parent of $y$ and the (sole) child of $y$, 
    and (2) projecting the labelling of each vertex onto its $\worlds{K}^{k+1}$ component (that is, $\lambda(x)=\bar u$ if the label of $x$ in the strategy "tree" was $(d_i,\bar u, \bar v)$). 
    % \sidediego{a picture would help}\santi{ver dibujo. 1) tomé los nodos de D en vez de los de S para construir $T$ sino perdia la info de las hojas 2) habria que agregar atomo $\epsilon(x_s,x_t)$ cuando sea necesario} \edwin{Creo que es la forma correcta de definir $T$}
    From this "tree" one can construct a "conjunctive program" $C[x_s,x_t]$ of "tree-width" $k$ containing all the homomorphism types of the tuples.
    Let $Q$ be the set of all pairs $(y,i)$ where $y$ is a vertex from $T$ and $1 \leq i \leq k+1$, and let $\approx$ be the finest equivalence relation on $Q$ such that: (a) $(y,i) \approx (y,i')$ if $\lambda(y)[i] = \lambda(y)[i']$, and (b) $(y,i) \approx (y',i)$ if $y'$ is a child of $y$ and $\lambda(y)[i] = \lambda(y')[i]$. The "conjunctive program" $C$ will use a variable for each equivalence class of $\approx$, which we denote by $[y,i]_\approx$, and it contains: 
    (i) an "atom" $a([y,i]_\approx,[y',j]_\approx)$ if $y=y'$ and $(\lambda(y)[i],\lambda(y)[j]) \in \dbracket{a}_K$, and
    (ii) an "atom" $p?([y,i]_\approx,[y,i]_\approx)$ if $\lambda(y)[i] \in \dbracket{p}_K$.
    The variable $x_s$ \resp{$x_t$} is the class $[y,i]_\approx$ where $y$ is the root of $T$ and $\lambda(y)[i] =u$ \resp{$\lambda(y)[i] = v$}. See Figure \ref{fig:1-to-2} for an illustration.

    \begin{claim}
        $C[x_s,x_t] \in \CPDLg{\Tw}$.
    \end{claim}
    \begin{proof}
        Due to \Cref{rem:connected} and the fact that $u$ and $v$ are in the same component of $K$, $\uGraph{C}$ is connected. On the other hand, $\uGraph{C[x_s,x_t]}$ has "tree-width" $k$ because $(T,\lambda)$ induces a "tree decomposition" $(T,\lambda')$ of "width" $k$ of $\uGraph{C[x_s,x_t]}$, where $\lambda'(y) = \set{[y,i]_\approx : 1 \leq i \leq k+1}$ for every $y \in V(T)$.
    \end{proof}
    % \color{gray}
    % Note that $K,u,v \models C[x_s,x_t]$, but it can't be that $K',u',v' \models C[x_s,x_t]$ because this would contradict the fact that we started with a winning strategy for Spoiler.
    % \color{black} 
    \begin{claim}
        $K,u,v \models C[x_s,x_t]$ and $K',u',v' \not\models C[x_s,x_t]$.
    \end{claim}
    \begin{proof}
    The fact that $K,u,v \models C[x_s,x_t]$ is a consequence of  $(y,i) \approx (y',i')$ implying $\lambda(y)[i] = \lambda(y')[i']$ by definition. Indeed, $f=\set{[y,i]_\approx \mapsto \lambda(y)[i] : y \in V(T), 1 \leq i \leq k+1}$ is a "$C$-satisfying assignment".

    We show $K',u',v' \not\models C[x_s,x_t]$ by means of contradiction. If $K',u',v' \models C[x_s,x_t]$, consider the corresponding "$C$-satisfying assignment" $f': \vars(C) \to \worlds{K'}$. Using $f'$ we can select a branch of Spoiler's strategy as follows: for every vertex $y$ labeled $(d_i, \bar u, \bar v)$ of the strategy tree, pick $(s, \bar u, \bar v')$ where $\bar v'[j] = f'([y',j]_\approx)$, where $y' \in V(T)$ is the (only) child of $y$ in the strategy tree. 
    It remains to show that there is a move from $(d_i, \bar u, \bar v)$ to $(s, \bar u, \bar v')$ in $\kSimGame[k+1]$.
    By definition of $C$, it contains all possible "atoms" derived from the substructure of $K$ induced by $\bar u$ using the variables $[y',1]_\approx, \dotsc, [y',k+1]_\approx$. Thus, $\set{\bar u[t] \mapsto f'([y',t]_\approx)}_t = \set{\bar u[t] \mapsto \bar v'[t]}_t$ witnesses a "partial homomorphism" from $K$ to $K'$, and hence $(d_i, \bar u, \bar v) \rightarrow (s, \bar u, \bar v')$ is a valid move of $\kSimGame[k+1]$. 
    This means that Spoiler's strategy tree cannot be finite since by descending the tree in this way we end up in a leaf $(d_i, \bar u, \bar v)$ of the strategy tree which has an outgoing move in $\kSimGame[k+1]$, contradicting the fact that the strategy was winning for Spoiler (more precisely point (iv)).
    \end{proof}
    Observe that, strictly speaking, $C$ may have an \emph{infinite} number of "atoms" (over the \emph{finite} set $\vars(C)$ of "variables"). This is because between two "worlds" of a "finite degree" "Kripke structure", there may be infinitely many transitions.
    However, from the fact that $K',u',v' \not\models C[x_s,x_t]$ and that $K'$ has "finite degree", one can extract a finite set of "atoms" $C' \subseteq C$ such that $K',u',v' \not\models C'[x_s,x_t]$ as we describe next.
    
    Consider the set $F$ of mappings $f: \vars(C) \to \worlds{K'}$ such that 
        (i) $f(x_s,x_t) = (u',v')$, and 
        (ii) $f(x)$ is at "distance" at most $|\vars(C)|$ from $u'$ or $v'$.
    Notice that, since $K'$ has "finite degree", $F$ is finite.
    For every $f \in F$, choose an "atom" $a_f(x_f,y_f)$ of $C$ which does not map to $K'$, \ie, such that $(f(x_f),f(y_f)) \not\in {\rightarrow_{a_f}}$ for the transition relation $\rightarrow_{a_f}$ of $K'$. Let $C'$ be the set of all such "atoms", plus some (finite) set of "atoms" from $C$ to assure (a) $\vars(C') = \vars(C)$ and (b) that $\uGraphC{C'}$ is connected. Notice that $C'$ is now finite.
    We claim that $K',u',v' \not\models C'[x_s,x_t]$. Indeed, if there was a "$C'$-satisfying assignment" $f: \vars(C') \to \worlds{K'}$, then $f \in F$ and $C'$ should include the "atom" $a_f(x_f,y_f)$ which is not mapped by $f$ to $K'$. Hence, there is no "$C'$-satisfying assignment" and $K',u',v' \not\models C'[x_s,x_t]$.
\begin{figure}\centering
\ifarxiv
    \includegraphics[width=.7\textwidth]{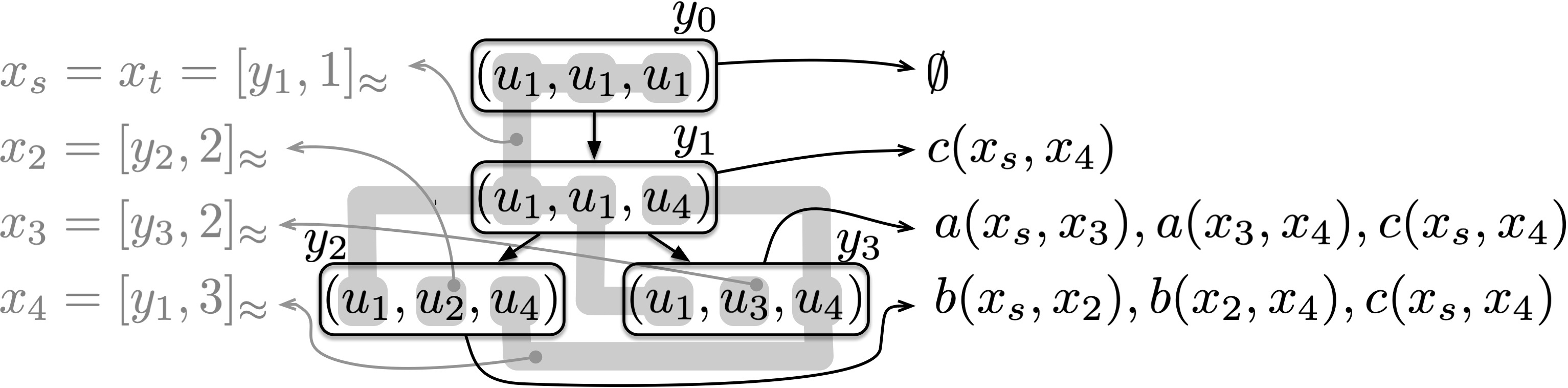}
\else
    \includegraphics[width=.45\textwidth]{img/1-to-2.png}
\fi
\caption{An example of the tree $T$ with nodes $y_0,y_1,y_2,y_3$ in the proof of Theorem \ref{thm:sim-char}\proofcase{1-to-2} built from Spoiler's strategy of Figure \ref{fig:game} (bottom). The variables involved and the "atoms" in $C$ such that $K,u_1,u_1\models\pi$ and $K',v_1,v_1\not\models\pi$ for $\pi=C[x_s,x_t]$. The grey area represents the partition induced by $\approx$.}
\label{fig:1-to-2}
\end{figure}
% \edwin{En la figura \ref{fig:1-to-2}, es $[y_0,1]$ en vez de $[y_1,0]$? Del gráfico no me queda muy claro cómo surgen estas clases. Se podrían colocar tres copias de este árbol y envolver las 4 clases que surgen con distinto color?}

% \santi{Gracias, sí, es $[y_0,1]$. Las clases son las zonas grises y las variables son las clases. No entendí lo de las 3 copias. Traté de no usar colores por si se imprime. Pero puedo intentar mejorar el dibujo si no queda claro.}

\smallskip
The case of $K,u \notpebblesim{k+1} K',u'$ is analogous. In this case, $C[x_s,x_t]$ is built just as before, but we set $x_t$ to be equal to $x_s$. The final "formula" is then $\tup{C[x_s,x_t]}$.

    \medskip

    %!TEX root = main.tex

\proofcase{2-to-1}
We now show the 2-to-1 implication by structural induction on the expression.
For any $(\bar u,\bar v)\in\set{(v,v'),((u,v),(u',v'))}$, assuming $K, \bar v \pebblesim {k+1} K', \bar v'$ %via $H$ 
and $K,\bar v \models e$ 
    for a "positive" $\CPDLg{\+G_{\Tw[k]}}$ expression $e$, we show $K',\bar v' \models e$.
    % We assume the normal form of \Cref{lem:normal-form}.\sidesanti{ya no hay normal form}

    \begin{toappendix}
        \subsection{Missing cases in the 2-to-1 proof of Theorem~\ref{thm:sim-char}}
        \label{apx:missing-cases-2-1-simchar}
    \end{toappendix}

    The only interesting case is the one where $e$ is a "conjunctive program", the remaining cases are easy and can be found in \ifarxiv \Cref{apx:missing-cases-2-1-simchar}.\else the "full version". \fi

    \begin{toappendix}

    \proofsubcase{$e=p$ for $p\in\Prop$} If $K,v \models p$ then it follows that since there is a "partial homomorphism" mapping $v$ to $v'$ we have $K',v' \models p$.

    \proofsubcase{$e=a$ for $a\in\Prog$} If $K,u,v \models a$ then it follows that since there is a "partial homomorphism" mapping $u\mapsto u'$ and $v\mapsto v'$ we have $K',u',v' \models a$.

    \proofsubcase{$e=\bar a$ for $a\in\Prog$} Analogous to the case $e=a$.

    \proofsubcase{$e=\phi_1\land\phi_2$} If $K,v \models \phi_1\land\phi_2$ then for $i=1,2$ we have $K,v \models \phi_i$ and by inductive hypothesis we have $K',v' \models \phi_i$ for $i=1,2$, which implies $K',v' \models \phi_1\land\phi_2$.

    \proofsubcase{$e=\pi_1\cup\pi_2$} Analogous to the previous case.

    \proofsubcase{$e=\phi?$} If $K,u,v \models \phi?$ then $u=v$ and $K,u \models \phi$. By inductive hypothesis $K',u' \models \phi$ and so $K',u',v' \models \phi$.

    \proofsubcase{$e=\tup{\pi}$} 
    If $K,v \models \tup{\pi}$ then there is some $\tilde v$ in $K$ such that $(v,\tilde v) \in \dbracket{\pi}_K$. Since we have: (1) $v$ and $\tilde v$ are part of the same connected component in $K$, (2) $(K, v) \pebblesim{k+1} (K', v')$, 
    %via $H$,
    \sidediego{saqué una H, acá y más abajo (comentado)} and (3) $k+1 \geq 3$, it follows that there is some $\tilde v'$ such that  
    $(K, v,\tilde v) \pebblesim{k+1} (K', v', \tilde v')$.
    % $(v,\tilde v) H (v', \tilde v')$. 
    This can be obtained by ``navigating'' from $v$ to $\tilde v$ in $K$ using the definition of $\pebblesim{k+1}$.
    We can now apply the inductive hypothesis obtaining that $K',v',\tilde v' \models \pi$, and thus that $K',v' \models \tup{\pi}$.

    \end{toappendix}

    \proofsubcase{$e = C[x_s,x_t]$} 
    We will actually work with the more general version of "programs" $C[\bar z]$ where $\bar z$ may have any number $1 \leq t \leq k+1$ of variables from $\vars(C)$, and the "underlying graph" $\uGraph{C[\bar z]}$ of $C[\bar z]$ is $\uGraph{C}$ plus the edges $\set{y,y'}$ for every pair of distinct variables $y,y'$ from $\bar z$. Hence,  $\uGraph{C[\bar z]} \in \Tw[k]$ means that all variables $\bar z$ appear in a bag of its "tree decomposition".
     The semantics $\dbracket{C[\bar z]}_{\tilde K}$ on a "Kripke structure" $\tilde K$ is as expected: the set of $t$-tuples $\bar u$ of "worlds" from $\tilde K$ such that $f(\bar z) = \bar u$ for some "$C$-satisfying assignment" $f$.

    Suppose we have $e = \hat C[x_s,x_t]$. We show the following:
    \begin{claim}
        For every $C \subseteq \hat C$ and every tuple $x_1, \dotsc, x_n$ of variables from $\vars(C)$ with $n \leq k+1$ such that $\uGraph{\hat C[x_1, \dotsc, x_n]} \in \Tw[k]$, for all $\bar v\in\worlds{K}^n$ and $\bar v'\in\worlds{K'}^n$: if $(K, \bar v) \pebblesim{k+1} (K', \bar v')$ and $\bar v  \in \dbracket{C[x_1, \dotsc, x_n]}_K$, then $\bar v'  \in \dbracket{C[x_1, \dotsc, x_n]}_{K'}$.
        %//TODO should they be **distinct** variables?
    \end{claim}
    First note that the claim above implies $K',u',v' \models \hat C[x_s,x_t]$. We now prove the claim by induction on the size $|C|$ of $C$. 
    % Assume $\pi=C[x_1, \dotsc, x_t]$ with $t \leq k$.
    % , and suppose $\set{x_1, \dotsc, x_t}$ are pairwise distinct variables.
    %
    Let $(T,\bagmap)$ be a "tree decomposition" of $\uGraph{\pi}$ of "width" $k$ for $\pi = C[x_1, \dotsc, x_n]$. Without loss of generality, suppose that the root bag $b$ of $T$ is so that $\bagmap(b) = \set{x_1, \dotsc, x_n}$. 
    % ,with the assumption that for every pair of vertices $v_1,v_2$ in the same "bag", there is an edge from $v_1$ to $v_2$. This is without loss of generality by Remark~\ref{rk:graph-programs-closed-minors}.
    
    Suppose $(K, \bar v) \pebblesim{k+1} (K', \bar v')$ and $\bar v  \in \dbracket{\pi}_K$. 
    We will show $\bar v'  \in \dbracket{\pi}_{K'}$.
    Let $h$ be a "$C$-satisfying assignment" on $K$ such that $\bar v = (h(x_1), \dotsc, h(x_n))$. We need to exhibit a "$C$-satisfying assignment" $h'$ on $K'$ such that $\bar v' = (h'(x_1), \dotsc, h'(x_n))$.

    Suppose that the root $b$ of $T$ has $\ell$ children $b_1, \dotsc, b_\ell$, and let $T_1, \dotsc, T_\ell$ be the corresponding "subtrees@trees". Let $\bar y_i$ be any vector of the (pairwise distinct) variables from $\bagmap(b_i) \cap \bagmap(b)$, that is, the variables which are both in the root "bag" of $T$ and in the "bag" of its $i$-th child. We can assume that $\bar y_i$ is of dimension $\leq k$ (otherwise parent and child would have exactly the same "bag").
    % in the same order as $(x_1, \dotsc, x_t)$. 
    Let $\bar v_0 = \bar v$, and $\bar v_i = h(\bar y_i)$ for every $1 \leq i \leq \ell$.
    Let $C_0$ be the set of all "atoms" of $C$ in the root bag (\ie, such that both variables are in $\bagmap(b)$). 
    Let $C_i \subseteq (C \setminus C_0)$ be the set of all "atoms" of $C \setminus C_0$ contained in "bags" of $T_i$ (\ie, such that both variables are in some "bag" of $T_i$). 
    
    It then follows that $K,\bar v \models C[x_1, \dotsc, x_n]$ if{f} $K,\bar v_i \models C_i[\bar y_i]$ for every $0 \leq i \leq \ell$.
    % \color{gray}
    % \begin{claim}
    %     Let $h(\bagmap(b)) = \set{h(x) : x \in \bagmap(b)}$.
    %     We have that there is a mapping $f : h(\bagmap(b)) \to \worlds{K'}$ such that
    %     (1) $K,\bar v_i \pebblesim{k+1} K',\bar v'_i$ for $\bar v'_i = f(\bar y_i)$ for every $i$, and (2) $\bar v' = (f(x_1), \dotsc, f(x_t))$. This is because of the "tree decomposition" assumption that every pair of elements of the root bag are at "distance" $\leq 1$,
    %     and of \Cref{rk:pebblesim-shrink}.
    % \end{claim}
    % \color{black}
    Now we have two possibilities, either $C_0 = C$ or $C_0 \subsetneq C$.
    % \begin{enumerate}
        % \item 

        Case $C_0 = C$:~~ If $C_0 = C$,  
        this means that all the "atoms" of $C$ use variables from $\set{x_1, \dotsc, x_n}$.
        Consider the function $h':\vars(C)\to K'$ such that $h'(x_i)=\bar u'[j]$ if $h(x_i)=\bar u[j]$. Let us show that $h'$ is a "$C$-satisfying assignment" on $K'$.
        For every "atom" $\pi'(x_i,x_j) \in C=C_0$, we have that $(h(x_i),h(x_j)) = (\bar v[i],\bar v[j])$ 
        is in $\dbracket{\pi'}_K$ by hypothesis on $h$.
        On the other hand we also have $K,\bar v[i], \bar v[j] \pebblesim{k+1} K',\bar v'[i],\bar v'[j]$. Hence, $(\bar v'[i], \bar v'[j]) \in \dbracket{\pi'}_{K'}$ by inductive hypothesis on $\pi'$. This means that $h'$ is a "$C$-satisfying assignment" on $K'$.

        Case $C_0 \subsetneq C$:~~
        % \item 
        If $C_0 \subsetneq C$, then $|C_i| < |C|$ for every $i$.\footnote{There is one degenerate case in which $C_0=\emptyset$, but since the "graph" $\uGraphC{C}$ of $C$ is connected, we can also assume that the "tree decomposition" is so that there is at least one "atom" in $C_0$.} We can then apply the inductive hypothesis on each $C_i$ since we have $K,\bar v_i \models C_i[\bar y_i]$ and $K,\bar v_i \pebblesim{k+1} K',\bar v'_i$, obtaining $K',\bar v'_i \models C_i[\bar y_i]$ for some "$C$-satisfying assignment" $h'_i$ on $K'$ sending $\bar y_i$ to $\bar v'_i$. This means that $h' \eqdef h'_0 \cup \dotsb \cup h'_\ell$ is a "$C$-satisfying assignment" on $K'$ and thus that $K',h'(x_1), \dotsc, h'(x_n) \models C[x_1, \dotsc, x_n]$.
    % \end{enumerate}
    \begin{toappendix}

    \proofsubcase{$e = \pi^*$} Suppose $K, u,v \models \pi^*$ 
    %//TODO \models está definido?
    and $K,u,v \pebblesim{k+1} K',u',v'$. We proceed by induction on the number of iterations. The base case $0$ corresponds to $K, u,v \models \pi^0$ which happens if and only if $u=v$, in which case $u'=v'$ by definition of $\pebblesim{k+1}$ and thus $K, u',v' \models \pi^0$.
    For the inductive case, $K, u,v \models \pi^n$. Note that $\pi^n$ is equivalent to $\tilde \pi = C[x_0,x_n]$ for $C = \set{\pi(x_i,x_{i+1}) \mid 0 \leq i < n}$. Since the "subexpressions" of $\tilde \pi$ are the same as those of $\pi^*$ (namely, $\pi$ and the "subexpressions" therein), we can replace $\pi^n$ with $\tilde \pi$ and invoke the previous case, obtaining  $K',u',v' \models \pi^n$.

    \proofsubcase{$\pi \circ \pi'$} This case is similar to the case $\pi^*$.
    \qed
    \end{toappendix}
\end{proof}

%!TEX root = main.tex

\subsection{Bisimulation relation}

We define here the notion of bisimulation on $(K,K')$, as before, using a two-player game
$\intro*\kBisimGame$, as was done for the simulation game $\kSimGame$. Assume, without loss of generality, that $K,K'$ have disjoint sets of "worlds".  This time the arena of $\kBisimGame$ has a set of positions $S \cup D $, where
\begin{align*}
    S &= \set{s} \times (\kHoms(K,K')\dcup \kHoms(K',K)),\\
    D &= \set{d_1,\dotsc, d_k} {\times}\big((\worlds{K}^k {\times} \worlds{K'}^k) \dcup (\worlds{K'}^k {\times} \worlds{K}^k)\big),
\end{align*}
Spoiler owns all positions from $S$, and Duplicator all positions from $D$. 
The set of moves of $\kBisimGame$ is the smallest set satisfying the following:
\begin{enumerate}
    % \item There is a move from $(s,\bar u, \bar v)$ to $(d_i,\bar u',\bar v)$ 
    % if $\bar u' = \bar u[i \mapcoord w]$, where $w$ is a "world" from $K$ at "distance" $\leq 1$ from $\bar u[j]$, for some $1 \leq j \leq k$ with $i\neq j$; and 
    % \sidediego{cambié para hacer más preciso}
    % \item There is a move from $(d_i,\bar u',\bar v)$ to $(s,\bar u',\bar v')$ 
    % if $\bar v' = \bar v[i \mapcoord w]$, where $w$ is a "world" from $K'$ at "distance" $\leq 1$ from $\bar v[j]$, for some $1 \leq j \leq k$ with $i\neq j$.
    \item There is a move from $(s,\bar u, \bar v)$ to $(d_i,\bar u',\bar v)$ 
    if $\bar u' = \bar u[i \mapcoord w]$, where $w$ is a "world" at "distance" $\leq 1$ from $\bar u[j]$, for some $1 \leq j \leq k$ with $i\neq j$; and 
    % if $\bar u'$ is the result of replacing component $i$ from $\bar u$ with a "world" at "distance" $\leq 1$ from some $\bar u[j]$ in the corresponding "structure@@kripke" (with $j \neq i$). 
    \item \label{it:kbisim:neg} There is a move from $(s,\bar u, \bar v)$ to $(s,\bar v,\bar u)$ if $\bar u[1] = \dotsb = \bar u[k]$ (and hence $\bar v[1] = \dotsb =\bar v[k]$).
    \item There is a move from $(d_i,\bar u',\bar v)$ to $(s,\bar u',\bar v')$ if $\bar v' = \bar v[i \mapcoord w]$, where $w$ is a "world" at "distance" $\leq 1$ from $\bar v[j]$, for some $1 \leq j \leq k$ with $i\neq j$.
    % if $\bar v'$ is the result of replacing component $i$ from $\bar v$ with a "world" at "distance" $\leq 1$ from some $\bar v[j]$ \sidesanti{in the corresponding "structure@@kripke" ?} (with $j \neq i$). 
\end{enumerate}
Again, the winning condition for Duplicator is just any infinite play.

We now define two related notions.
\AP
\newcommand{\quasibisim}{$k$-half-bisimulation}
The ""\quasibisim"" notion $\pebblequasibisim k$ is obtained as follows.
Given two "Kripke structures" $K, K'$ with disjoint set of "worlds" (this is without loss of generality), and given tuples $\bar v\in W(K)^k$ and $\bar v'\in W(K')^k$, we  say that there is a "\quasibisim" from $K,\bar v$ to $K',\bar v'$, notated $K,\bar v \intro*\pebblequasibisim k K',\bar v'$, 
$\phantomintro{\notpebblequasibisim}$%
if 1) $(s,\bar v, \bar v')$ is a valid position of $\kBisimGame$ on $(K,K')$  (\ie, they induce "partial homomorphisms"), 2) all the "worlds" in $\bar v$ are in the same connected component of $K$ and 3) Duplicator has a winning strategy from $(s,\bar v, \bar v')$. 
\AP
% If $(s,\bar v_2, \bar v_1)$ is also a position and Duplicator has also a winning strategy from $(s,\bar v_2, \bar v_1)$  
We say that there is a ""$k$-bisimulation"" between $K,\bar v$ and $K',\bar v'$, and we note it $K,\bar v \intro*\pebblebisim k K', \bar v'$ if $K,\bar v \pebblequasibisim k K',\bar v'$ and $K',\bar v' \pebblequasibisim k K,\bar v$.

% \color{gray}
% \AP
% Given two "Kripke structures" $K_1, K_2$ with disjoint set of "worlds" (this is without loss of generality), and given tuples $\bar v_1\in W(K_1)^k$ and $\bar v_2\in W(K_2)^k$, we  say that there is a "\quasibisim" from $K_1,\bar v_1$ to $K_2,\bar v_2$, notated $K_1,\bar v_1 \intro*\pebblequasibisim k K_2,\bar v_2$, 
% $\phantomintro{\notpebblequasibisim}$%
% if there is $H\subseteq (W(K_1)^k\times W(K_2)^k) \dcup (W(K_2)^k\times W(K_1)^k)$ such that $\bar v_1 H \bar v_2$ and for all $\bar u_1 \in W(K_1)^k$, $\bar u_2 \in W(K_2)^k$, $(\ell,\ell') \in \set{(1,2), (2,1)}$ and $\bar u_\ell H \bar u_{\ell'}$:
% \begin{enumerate}
%     \item\label{it:kbisim:hom} $(\bar u_\ell,\bar u_{\ell'})\in \kHoms(K_\ell,K_{\ell'})$;
%     \item\label{it:kbisim:travel} for every $1 \leq i,j \leq k$ with $i\neq j$, and every $\hat u_i$ in $K_\ell$ at "distance" $\leq 1$ from $\bar u_\ell [j]$, there is some $\hat u'_i$ in $K_{\ell'}$ such that $\bar u_\ell[i\mapcoord \hat u_i] H \bar u_{\ell'}[i\mapcoord \hat u'_i]$; and
%     \item\label{it:kbisim:neg} if $\bar u_\ell[1]= \dotsb = \bar u_\ell[k]$, then $\bar u_{\ell'} H \bar u_{\ell}$.
% \end{enumerate}
\begin{remark}\label{rem:bisim-imples-sim}
    If $K_1,\bar v_1 \pebblequasibisim k K_2,\bar v_2$ then $K_1,\bar v_1 \pebblesim k K_2,\bar v_2$.
\end{remark}

A similar characterization result as \Cref{thm:sim-char} can be shown for "\quasibisim", that is: (1) $K, \bar v \pebblequasibisim{k+1} K', \bar v'$ if, and only if, (2) for every $\CPDLg{\Tw[k]}$ "expression" $e$ ("formula" or "program" accordingly), we have $K,\bar v \models e$ implies $K',\bar v' \models e$.
The proof follows the same general blueprint. In the case 1-to-2, the strategy of Spoiler can be more complex, as he may change structures. But it can still be transformed into a "formula" of $\CPDLg{\Tw[k-1]}$, using negation each time it switches structure. For the 2-to-1 case, the proof goes by structural induction on the "expression" exactly as in \Cref{thm:sim-char}. The base cases are just like in \Cref{thm:sim-char} due to \Cref{rem:bisim-imples-sim}. If $\phi = \lnot \psi$, then by \Cref{it:kbisim:neg} we have $K',v' \pebblequasibisim{k+1} K,v$ and thus by inductive hypothesis we have $K',v' \models \psi$ implies $K,v \models e$. The latter is the same as $K,v \models \phi$ implies $K',v' \models \phi$.  The formal statement and proof details can be found in \ifarxiv \Cref{subsec:char-quasibisim}. \else the "full version". \fi

\begin{toappendix}
\subsection{Characterization in terms of \quasibisim}
\label{subsec:char-quasibisim}
\begin{lemma}\label{lem:quasibisim-char}
    Let $k\geq 2$.
    Given "Kripke structures" $K,K'$ of "finite degree" and "worlds" $u,v\in\worlds{K}$ and $u',v'\in\worlds{K'}$, the following are equivalent
    \begin{enumerate}
        \item  for every $\CPDLg{\Tw[k]}$ "formula" $\phi$, we have $K,v \models \phi$ implies $K',v' \models \phi$; and
        \item $K,v \pebblequasibisim{k+1} K',v'$;
    \end{enumerate}
    and the following are equivalent
    \begin{enumerate}
        \item for every $\CPDLg{\+G_{\Tw[k]}}$ "program" $\pi$, we have $K,u,v \models \pi$ implies $K',u',v' \models \pi$; and
        \item $K,u,v \pebblequasibisim{k+1} K',u',v'$.
    \end{enumerate}
Furthermore, the hypothesis of "finite degree" is only needed for the $1$-to-$2$ implications.
\end{lemma}
\begin{proof}
    % //TODO give out more details!! (in appendix)
    \proofcase{1-to-2}
    We proceed as in the proof of \Cref{thm:sim-char}, by contrapositive. Assume $K,u,v \notpebblequasibisim{k+1} K',u',v'$. 
    The case where $((u,v), (u',v')) \not\in \kHoms(K,K')$ is  trivial, so let us assume the contrary.
    This means that Spoiler has a strategy to win in a bounded number of rounds in $\kBisimGame[k+1]$ starting from $(s,((u,v),(u',v')))$.

    The winning strategy of Spoiler is again a finite "tree", whose vertices are labelled with positions from $\kBisimGame[k+1]$. In particular: (i) the root must be labeled with $(s,((u,v),(u',v')))$, (ii) any vertex labeled $(d_i,\bar u, \bar v)$ has a child labeled $(s,\bar u', \bar v)$ for every possible move of $\kBisimGame[k+1]$, (iii) any vertex labeled $(s,\bar u, \bar v)$ has exactly one child, whose label is consistent with a move of $\kBisimGame[k+1]$ (remember that Spoiler has two types of moves this time), (iv) for every leaf (labeled $(d_i,\bar u, \bar v)$) there is no possible move from that position in $\kBisimGame[k+1]$.
    Since both $K,K'$ are of "finite degree", there is a finite number of moves departing from any position of $\kBisimGame[k+1]$, and thus the branching of the "tree" is finite. Therefore, the strategy tree is finite.

    From such a winning strategy for Spoiler, consider the "tree" $T$ and labeling $\lambda: V(T) \to (\worlds{K}^{k+1} \dcup \worlds{K'}^{k+1})$ resulting from: (1) removing all vertices $v$ labeled by Duplicator positions, and adding edges between the parent of $v$ and all the children of $v$, and (2) projecting the labeling of each vertex onto the components of its first  "structure@@kripke" (that is, $\lambda(x)=\bar u$ if the label of $x$ in the strategy "tree" was $(s,\bar u, \bar v)$).
    \AP
    Let us call an edge ``""switching""'' if for some $\bar u, \bar v$, the parent is labeled $(s,\bar u, \bar v)$ and the child labeled $(s,\bar v, \bar u)$ (note that this implies $\bar u = (u, \dotsc, u)$ and $\bar v=(v, \dotsc, v)$ for some $u,v$ by definition of $\kBisimGame[k+1]$).
    Then, $T$ can be seen as a ``"tree" of "trees"'', each tree being a component of $T$ after removing "switching" edges. Consider all the "trees" $T_1, \dotsc, T_n$ resulting from removing the "switching" edges from $T$.
    We now proceed bottom-up: For each $T_i$ such that only its root may have a "switching" edge from (the leaf of) another "tree", we proceed as in the proof of \Cref{thm:sim-char}, and construct a "conjunctive program" $C_i[x,x]$. For the inductive case, suppose we have some tree $T_j$ whose leaves may have switching edges to other trees. Then we construct the "conjunctive program" for $T_j$ as before, but we add some extra "atoms": we add the "atom" $\lnot\tup{C_i[x,x]}? (z,z)$ for the variable $z$ of $C_j$ of the form $[v,\ell]_\approx$ where $v$ is the leaf of $T_j$ incident to the "switching" edge connecting $T_j$ with $T_i$.\footnote{Note that in this case $[v,\ell]_\approx = [v,\ell']_\approx$ for all $\ell,\ell'$, since the definition of $\kBisimGame[k+1]$ restricts these "switching" moves to be applied only to positions of the form $(s,(u,\dotsb, u), (v, \dotsb, v))$.} By a similar remark as in the proof of \Cref{thm:sim-char}, we can assume that each $C_i$ has a finite number of "atoms" (this time using "finite degree" of the `current' "structure@@kripke").

    Finally, one can show that the resulting "conjunctive program" $C[x_s,x_t]$ from processing all the $T_i$'s faithfully describes the winning strategy for Spoiler, and thus that it is not possible to have $K,u,v \models C[x_s,x_t]$.

    \proofcase{2-to-1} Assume $K,v \pebblequasibisim{k+1} K',v'$. The proof goes by structural induction on the "formula" like \Cref{thm:sim-char}. 
    % If $\phi$ is "positive" then since $K,v \pebblequasibisim{k+1} K',v'$ we have $K,v \pebblesim{k+1} K',v'$ by \Cref{rem:bisim-imples-sim} and thus $K,v \models \phi$ implies $K',v' \models \phi$.
    %
    The base cases are just like in \Cref{thm:sim-char} due to \Cref{rem:bisim-imples-sim}.
    If $\phi = \lnot \psi$, then by \Cref{it:kbisim:neg} we have $K',v' \pebblequasibisim{k+1} K,v$ and thus by inductive hypothesis we have $K',v' \models \psi$ implies $K,v \models \psi$. The latter is the same as $K,v \models \phi$ implies $K',v' \models \phi$.
    %
    % If $\phi = \tup{\pi}$, then by iterating \Cref{it:kbisim:travel} there must be some $u$ and $u'$ such that  $(v,u) \in \dbracket{\pi}_K$ and $K,v,u \pebblequasibisim{k+1} K',v',u'$. By inductive hypothesis $(v',u') \in \dbracket{\pi}_{K'}$ and thus $K',v' \models \phi$.
    %
    The remaining cases are exactly as in the proof of \Cref{thm:sim-char}, with the only difference that we use the inductive hypothesis on "formulas" which may contain negation.
\end{proof} 
\end{toappendix}

Observe that, since "formulas" are closed under negation, from the previous characterization \ifarxiv (\Cref{lem:quasibisim-char}) \fi  we obtain that $K,v \pebblequasibisim{k+1} K',v'$ if, and only if, for every $\CPDLg{\Tw[k]}$ "formula" $\phi$, we have $K,v \models \phi$ if{f} $K',v' \models \phi$. In order to have a similar result for "programs" we just need the symmetrical closure of this notion, which is the "$k$-bisimulation":

\begin{theorem}
    \label{them:bisim-charac}
    Let $k\geq 2$.
    Given "Kripke structures" $K,K'$ of "finite degree" and "worlds" $u,v\in\worlds{K}$ and $u',v'\in\worlds{K'}$, the following are equivalent
    \begin{enumerate}
        \item for every $\CPDLg{\Tw[k]}$ "formula" $\phi$, we have $K,v \models \phi$ if{f} $K',v' \models \phi$; and
        \item $K,v \pebblebisim{k+1} K',v'$;
    \end{enumerate}
    and the following are equivalent
    \begin{enumerate}
        \item for every $\CPDLg{\+G_{\Tw[k]}}$ "program" $\pi$, we have $K,u,v \models \pi$ if{f} $K',u',v' \models \pi$; and
        \item $K,u,v \pebblebisim{k+1} K',u',v'$.
    \end{enumerate}
Furthermore, the hypothesis of "finite degree" is only needed for the $1$-to-$2$ implications.
\end{theorem}

\subsection{Computing (bi)simulation}\label{sec:bisim-calc}
Observe that when $K,K'$ are finite "structures@@kripke", the games $\kSimGame$ and $\kBisimGame$ are finite as well, with $O(|\worlds{K}|^k \cdot |\worlds{K'}|^k)$ positions. The winning condition for Duplicator is in both cases a ``Safety condition''. Hence, the set of winning positions can then be computed in linear time in the number of moves and positions of the game, for example via the classical McNaughton-Zielonka computation of its attractor decomposition. 
This yields a polynomial algorithm, which runs in time $(|\worlds{K}|\cdot|\worlds{K'}|)^{O(k)}$, for computing the "$k$-simulation" and "$k$-bisimulation" relations $\pebblesim{k}$ and $\pebblebisim{k}$.

\begin{proposition}
    For every $k$, the relations $\pebblesim k$ and $\pebblebisim k$ on finite "structures@@kripke" $K,K'$ can be computed in polynomial time.
\end{proposition}

%!TEX root = main.tex
It is also possible (and useful) to reduce the notion of simulation $\pebblesim {k}$ for $\CPDLg{\Tw[k-1]}$ between "structures@@kripke" $K,K'$ to the classical simulation relation of Modal Logics (ML) on a $k$-power "Kripke structures", $S_k(K)$ and $S_k(K')$, where worlds of $S_k(K)$ are $k$-tuples of worlds from $K$, and the "atomic programs" reflect the different ways of changing a component (via a witness at "distance" $\leq 1$ in $K$) as in $\kSimGame$. The same can be done with bisimulation $\pebblebisim {k}$ and (one-way) simulation in both senses between "Kripke structures", each with the information of $K$ and $K'$. This would yield an alternative proof for showing that the (bi)simulation relations are computable in polynomial time. Full details on this connection with ML simulation are given in \ifarxiv the appendix, Section~\ref{sec:MLbisim}. \else the "full version".\fi

\begin{toappendix}
\subsection{A connection with simulation of basic modal logic}\label{sec:MLbisim}
\AP Let
\begin{align}
    K&=(X,\set{\rightarrow_a \mid a \in \Prog}, \set{X_p \mid p \in \Prop})\label{eq:kripke_structure}\\
    K'&=(X',\set{\rightarrow'_a \mid a \in \Prog}, \set{X'_p \mid p \in \Prop})\nonumber
\end{align}
be "Kripke structures". 
\AP
For $v\in X$ and $v'\in X'$ we say that $K',v'$ ""ML-simulates"" $K,v$ (notated $K,v\intro*\mlsim K',v'$) if there is $Z\subseteq X\times X'$ such that $vZv'$ and for all $uZu'$ we have
\begin{enumerate}
\item\label{it:MLsim:harmony} $u\in X_p$ implies $u'\in X'_p$ for all $p\in\Prop$; and
\item\label{it:MLsim:zig} if $u\to_a w$ then there is $w'\in X'$ such that $u'\to'_a w'$ and $wZw'$.
% \item\label{it:MLsim:zag} if $u'\to w'$ then there is $w\in X$ such that $u\to w$ and $wZw'$.
\end{enumerate}
% We say that $K,v$ is ""ML-bisimilar"" to $K',v'$ (notated $K,v\intro*\mlbisim K',v'$) if there is $Z\subseteq X\times X'$ such that $vZv'$ and for all $uZu'$ we have items \ref{it:MLsim:harmony} and \ref{it:MLsim:zig} plus
% \begin{enumerate}
% \item[$1'$.\ ] $u'\in X'_p$ implies $u\in X_p$ for all $p\in\Prop$; and
% \item[$2'$.\ ] if $u'\to_a w'$ then there is $w\in X$ such that $u\to_a w$ and $wZw'$.
% \end{enumerate}

It is well known that if $K'$ is of "finite degree" then $K,v\mlsim K',v'$ if{f} any positive modal formula true of $v$ is also true of $v'$.

\AP
We next show that our "$k$-simulation" $\pebblesim{k}$ can be reduced to the ""ML-simulation"" $\mlsim$. For a "Kripke structure" $K$ as in \eqref{eq:kripke_structure}
and $k>0$ define the following "Kripke structure":
\begin{align*}
    S_k(K) = (S_{K,k},
    \set{\rightarrow_i^{K,k}\mid i\in\Prog_k},
    \set{S_r^{K,k}\mid r\in\Prop_{K,k}})
    % \set{S_{(p,i)}^{K,k} \mid p \in \Prop_{K,k},i\in \Prog_k)} \cup
    %     \set{S_{a,i,j}^{K,k} \mid a \in \Prog,i,j\in \Prog_k)} \cup
    %     \set{S_{i,j}^{K,k} \mid i,j\in \Prog_k)}),
\end{align*}
where 
\begin{itemize}
\item $S_{K,k}=X^k$
\item $\Prog_k=\{1,\dots,k\}$
\item $\Prop_{K,k}=(\Prop\times\{1,\dots,k\}) \cup (\Prog\times \{1,\dots,k\}^2) \cup \{1,\dots,k\}^2$

\item $\bar u\to^{K,k}_i \bar v$ iff there is $1 \leq j \leq k$ with $i\neq j$, and $w\in X$ at "distance" $\leq 1$ from $\bar u[j]$ such that $\bar v = \bar u[i\mapcoord w]$

\item For $p\in\Prop$, $\bar v\in S_{(p,i)}^{K,k}$ iff $\bar v[i]\in X_p$; for $a\in\Prog$, $\bar v\in S_{(a,i,j)}^{K,k}$ iff $\bar v[i]\to_a \bar v[j]$; and $\bar v\in S_{(i,j)}^{K,k}$ iff $\bar v[i]=\bar v[j]$
\end{itemize}

\begin{theorem}\label{thm:reduce_sim_to_MLsim}
For "Kripke structures" $K,K'$, $\bar w\in W(K)^k$, and $\bar w'\in W(K')^k$:
$K,\bar w \pebblesim{k} K',\bar w'$ iff $S_k(K),\bar w\mlsim S_k(K'),\bar w'$.
\end{theorem}

\begin{proof}
% We show item \ref{thm:reduc:bisimML}, as item \ref{thm:reduc:simML} is analogous. 
It can be shown that for any $\bar u\in X^k$ and $\bar u'\in X'^k$
item (\ref{it:MLsim:harmony}) of $\mlsim$ over $S_k(K)$ and $S_k(K')$ is equivalent to $(\bar u, \bar u') \in \kHoms(K,K')$.

Suppose $K,\bar v \pebblesim{k} K',\bar v'$ via $Z$. First observe that $(s,\bar w, \bar w')$ is a valid position of $\kSimGame$ on $(K,K')$. We show that Duplicator has a winning strategy in $\kSimGame$ from $(s,\bar w, \bar w')$ which is the result of `copying' $Z$: suppose $\bar uZ\bar v$ and Spoiler moves from $(s,\bar u, \bar v)$ to $(d_i,\bar u',\bar v)$, where $\bar u' = \bar u[i \mapcoord w]$ and $w$ is a "world" from $K$ at "distance" $\leq 1$ from $\bar u[j]$, for some $1 \leq j \leq k$ with $i\neq j$. Since $\bar u\to^{K,k}_i\bar u'$ and $\bar uZ\bar v$, by item (\ref{it:MLsim:zig}) of $\mlsim$ over $S_k(K)$ and $S_k(K')$, there is $v'\in X'^k$ such that $\bar u'\to^{K',k}_i\bar v'$ and $\bar u'Z\bar v'$. By definition, $(s,\bar u',\bar v')$ is a valid position and Duplicator makes the move from $(d_i,\bar u',\bar v)$ to $(s,\bar u',\bar v')$. Since Duplicator can always make a move, we conclude he has a winning strategy.

For the other direction, suppose that Duplicator has a winning strategy in $\kSimGame$ from $(s,\bar w, \bar w')$, described by an "tree" $T$ without leaves whose vertices are labelled with positions from $\kSimGame[k+1]$. In particular: (i) the root must be labeled with $(s,\bar w,\bar w')$, (ii) any vertex labeled $(d_i,\bar u, \bar v)$ has exactly one child labeled $(s,\bar u, \bar v')$ whose label is consistent with a move of $\kSimGame[k+1]$, (iii) any vertex labeled $(s,\bar u, \bar v)$ has one child for every possible move of $\kSimGame[k+1]$. Let $Z\subseteq X^k\times X'^k$ defined by $\bar uZ\bar v$ iff $(s,\bar u,\bar v)$ is a label of some node of $T$. It is straightforward to see that $Z$ satisfies items (\ref{it:MLsim:harmony}) and (\ref{it:MLsim:zig}), and that $S_k(K),\bar v\mlsim S_k(K'),\bar v'$ via $Z$.
% It is clear that item \ref{it:MLsim:harmony} of $\mlsim$ over $S_k(K)$ and $S_k(K')$ holds. To see item \ref{it:MLsim:harmony}
%
% $(d_i,\bar u,)$
%
% Also, 
% (1) if $K,\bar v \pebblesim{k} K',\bar v'$ via $Z$ and $\bar uZ\bar u'$ then 
% item \ref{it:MLsim:zig} of $\mlsim$ over $S_k(K)$ and $S_k(K')$ is implied by 
% item \ref{it:ksim:zig} of $\pebblesim{k}$ over $K$ and $K'$ and 
% (2) if $S_k(K),\bar v \mlsim S_k(K'),\bar v'$ via $Z$ and $\bar uZ\bar u'$then 
% item \ref{it:ksim:zig} of $\pebblesim{k}$ over $K$ and $K'$ is implied by 
% item \ref{it:MLsim:zig} of $\mlsim$ over $S_k(K)$ and $S_k(K')$.
%
% From this it can be shown that $K,\bar v \pebblesim{k} K',\bar v'$ via $Z$ iff $S_k(K),\bar v\mlsim S_k(K'),\bar v'$ via $Z$.
\end{proof}

The "\quasibisim" relation can also be reduced to the "ML-simulation".
Given "Kripke structures" $K$ and $K'$ as in \eqref{eq:kripke_structure}, we can construct in polynomial time "Kripke structures" $Q_k(K,K')$ and $Q'_k(K,K')$ such that
\begin{theorem}\label{thm:reduce_quasisim_to_MLsim}
For "Kripke structures" $K,K'$, $\bar w\in \worlds{K}^k$ and $\bar w'\in \worlds{K'}^k$:
$K,\bar w \pebblequasibisim k K',\bar w'$ iff $Q_k(K,K'),\bar w\mlsim Q'_k(K,K'),\bar w'$.
\end{theorem}
% \begin{proof}[Proof sketch]
%         One can show that $Q_k(K,K'),\bar u\mlsim Q'_k(K,K'),\bar u'$ via $Z$ iff $K,\bar v\pebblequasibisim k K',\bar v'$ via  $Z$.
% \end{proof}
%
and hence we obtain
\begin{corollary}
For "Kripke structures" $K,K'$, $\bar w\in \worlds{K}^k$ and $\bar w'\in \worlds{K'}^k$:
$K,\bar w \pebblebisim k K',\bar w'$ iff $Q_k(K,K'),\bar w\mlsim Q'_k(K,K'),\bar w'$ and $Q_k(K',K),\bar w'\mlsim Q'_k(K',K),\bar w$.
\end{corollary}

% For a tuple $\bar u\in Y^k$ and $u\in Y$, we write $\bar u==u$ if $\bar u[i]=u$ for all $i$ and $eq(\bar u)$ if $\bar u==u$ for some $u$. 

\paragraph{Definition of $Q_k(K,K')$ and $Q'_k(K,K')$}
Suppose $K$ and $K'$ are "Kripke structures" as in \eqref{eq:kripke_structure} and further satisfy that $X$ and $X'$ are disjoint (this is without loss of generality). The "Kripke structures" $Q_k(K,K')$ and $Q'_k(K,K')$ are defined as follows:
\begin{align*}
    Q_k(K,K') = (Q,\set{\rightarrow_i\mid i\in \Prog_Q}, \set{Q_q \mid q \in \Prop_Q})
\end{align*}
where 
\begin{itemize}
\item $Q=X^k\cup X'^k$
\item $\Prog_Q=\{1,\dots,k\}\cup (X\times X')\cup(X'\times X)$
\item $\Prop_Q=(\Prop\times\{1,\dots,k\}) \cup (\Prog\times \{1,\dots,k\}^2) \cup \{1,\dots,k\}^2$

\item For each $1\leq i\leq k$, $\bar u\to_i \bar v$ iff there is $Y\in\set{X,X'}$, such that
$\bar u,\bar w\in Y^k$ and there is $1 \leq j \leq k$ with $i\neq j$, and $w\in Y$ at "distance" $\leq 1$ from $\bar u[j]$ and $\bar v = \bar u[i\mapcoord w]$

\item For each $(v_1,v_2)\in (X\times X')\cup(X'\times X)$, $\bar u_1\to_{(v_1,v_2)}\bar u_2$ iff $\bar u_1[i]=v_1$ and $\bar u_2[i]=v_2$ for all $1\leq i\leq k$ (notice that in case if $v_1\in X$ then $\bar u_1\in X^k$ and $\bar u_2\in X'^k$; and in case $v_1\in X'$ then $\bar u_1\in X'^k$ and $\bar u_2\in X^k$)

\item For $(\tilde X,\tilde \to)\in \{(X,\to),(X',\to')\}$, if $\bar v\in \tilde X^k$, then 
for $p\in\Prop$, $\bar v\in Q_{(p,i)}$   iff $\bar v[i]\in \tilde X_p$; 
for $a\in\Prog$, $\bar v\in Q_{(a,i,j)}$ iff $\bar v[i]\tilde\to_a \bar v[j]$; and 
$\bar v\in Q_{(i,j)}$                    iff $\bar v[i]=\bar v[j]$

% \item If $\bar v\in X^k$, then for $p\in\Prop$, $\bar v\in Q_{(p,i)}$ iff $v[i]\in X_p$; for $a\in\Prog$, $\bar v\in Q_{(a,i,j)}$ iff $v[i]\to_av[j]$; and $\bar v\in Q_{(i,j)}$ iff $v[i]=v[j]$

% \item If $\bar v\in X'^k$, then for $p\in\Prop$, $\bar v\in Q_{(p,i)}$ iff $v[i]\in X'_p$; for $a\in\Prog$, $\bar v\in Q_{(a,i,j)}$ iff $v[i]\to'_av[j]$; and $\bar v\in Q_{(i,j)}$ iff $v[i]=v[j]$
\end{itemize}
and
\begin{align*}
    Q'_k(K,K') = (Q',\set{\rightarrow'_i\mid i\in \Prog_{Q}}, \set{Q'_q \mid q \in \Prop_{Q}})
\end{align*}
where 
\begin{itemize}
\item $Q'=Q$%\cup\{*\}$
% \item $\Prog_{Q'}=\Prog_{Q}$
% \item $\Prop_{Q'}=\Prop_{Q}$

\item For each $1\leq i\leq k$, $\bar u\to'_i \bar w$ is defined as in $Q_k(K,K')$ replacing $\to_i$ by $\to'_i$

\item For each $(v_1,v_2)\in (X\times X')\cup(X'\times X)$, $\bar u_2\to'_{(v_1,v_2)}\bar u_1$ iff $\bar u_1\to_{(v_1,v_2)}\bar u_2$ in $Q_k(K,K')$

\item $Q'_{(p,i)}$, $Q'_{(a,i,j)}$ and $Q'_{(i,j)}$ are defined as in $Q_k(K,K')$ replacing $Q_{(p,i)}$ by $Q'_{(p,i)}$, $Q_{(a,i,j)}$ by $Q'_{(a,i,j)}$ and $Q_{(i,j)}$ by $Q'_{(i,j)}$
\end{itemize}
\end{toappendix}
%!TEX root = main.tex

\section{Separation}

In this section we show that, while $\CPDLg{\Tw[1]}$ and $\CPDLg{\Tw[2]}$ are "equi-expressive", we have a strict hierarchy $\CPDLg{\Tw[2]} \lleqs \CPDLg{\Tw[3]} \lleqs \dotsb$ from $\Tw[2]$.
\begin{theorem}\label{thm:separation}
    For every $k \geq 2$, $\CPDLg{\Tw[k]} \lleqs \CPDLg{\Tw[k+1]}$.
\end{theorem}
Essentially we show that the presence of a $(k+1)$-clique can be expressed in $\CPDLg{\Tw[k]}$ but not in $\CPDLg{\Tw[k-1]}$, for every $k\geq 3$. 
\AP
Consider the following ""$(k+1)$-clique formula"" 
\[
    \xi_{k+1} \eqdef \tup{C[x_1,x_{k+1}]} \land \lnot \tup{C'[x_1,y]}
\] 
of $\CPDLg{\Tw[k]}$ where $C = \set{a(x_i,x_j) \mid 1 \leq i < j \leq k+1}$ and $C' = \set{a(x_1,y), a(y,y)}$, for any fixed $a \in \Prog$. We will show that $\xi_{k+1}$ cannot be expressed in $\CPDLg{\Tw[k-1]}$, for every $k \geq 3$.
This inexpressivity result will be a direct consequence of the following ``tree-like model property'', which is of independent interest.
\begin{proposition}\label{prop:treewidth-k-unravelling}
    For every $k \geq 2$, "Kripke structure" $K$, and world $u \in \worlds{K}$, there exists a "Kripke structure" $\hat K$ of "tree-width" $\leq k-1$ and "world" $\hat u \in \worlds{\hat K}$ such that $K,u \pebblebisim{k} \hat K,\hat u$.
    Further, if $K$ is countable, $\hat K$ has a countable "tree decomposition" of "width" $\leq k-1$.
\end{proposition}
\begin{proof}
    Consider the set $\+S$ of all nonempty sets of at most $k$ "worlds" from $K$.

    Consider the infinite "tree" $T$ having $V(T) = \set{u} \cdot \+S^*$ (\ie, vertices are finite sequences on $\+S$ starting with $\set u$) and having an edge $\set{u,v}$ if $u$ is a prefix of $v$ of length $|v|-1$. Consider the root vertex $r$ of $T$ to be $r = \set{u}$.
    Let us denote by $\lambda(v)$ the last set of $v\in V(T)$ (\ie, $\lambda(v \cdot S) = S$ for every $S \in \+S$ and $v \in \+S^*$).

    % We show how to first define an infinite tree $T$ with a labeling function $\lambda: V(T) \to \+S$.

    % The root $r$ of $T$ has a label $\lambda(r) = \set{u}$. 
    % For every vertex $v$ of $T$ and set $S \in \+S$ there is exactly one child $v'$ of $v$ with label $\lambda(v') = S$.

    For a pair of worlds $u,u'$ of $K$ and a pair of vertices $v,v'$ of the "tree" $T$, let us write $\approx$ to denote the reflexive, symmetric and transitive closure of $\set{((u,v),(u',v'))\mid v' \text{ is child of } v \text{ in $T$, }u=u' \text{ and } u\in \lambda(v) \cap \lambda(v')}$. Let $[u,v]_\approx$ denote the equivalence class of $(u,v)$ with respect to $\approx$.

    We now define $\hat K$ as having 
    \begin{enumerate}
        \item \label{item:Kunravel:1} $\worlds{\hat K} = \set{[u,v]_\approx \mid u \in \worlds{K}, v \in V(T), u \in \lambda(v)}$;
        \item for every "atomic program" $a \in \Prog$ we have that $\dbracket{a}_{\hat K}$ consists of all pairs $([u,v]_\approx,[u',v]_\approx)$ of "worlds" such that $(u,u') \in \dbracket{a}_{K}$; and
        \item for every "atomic proposition"  $p \in \Prop$ we have that $\dbracket{p}_{\hat K}$ consists of all "worlds" $[u,v]_\approx$ such that $u \in \dbracket{p}_{K}$.
    \end{enumerate}
    Finally, consider the labeling $\bagmap : V(T) \to \pset{\worlds{\hat K}}$ mapping every vertex $v$ to $\set{[u,v]_\approx \mid u \in \lambda(v)}$. It follows that $(T,\bagmap)$ is a "tree decomposition" of $\hat K$ of "tree-width" at most $k-1$. It follows that $(T,\bagmap)$ is a "tree decomposition". 
    \begin{claimrep}
        $(T,\bagmap)$ is a "tree decomposition" of $\hat K$ of "tree-width" at most $k-1$.
    \end{claimrep}
    \begin{proof}
        $(T,\bagmap)$ has to verify the three conditions ("A@@treedec", "B@@treedec", "C@@treedec") of a "tree decomposition". 

        Condition "A@@treedec" holds  since every "world" $[u,v]_\approx$ of $\hat K$ is contained in the "bag" of $v$.

        Condition "B@@treedec" holds for any "atomic program" $a \in \Prog$ for the same reason: every pair $([u,v]_\approx,[u',v']_\approx)$ of "worlds" such that $v=v'$ have that both $[u,v]_\approx$ and $[u',v']_\approx$ are contained in the "bag" of $v$.
        
        For the ``connectedness'' condition "C@@treedec", consider any "world" $[u,v]_\approx$ of $\hat K$, and observe that $[u,v]_\approx \in \bagmap(v')$ for any vertex $v'$ from $V_{u,v} = \set{v'\in V(T) \mid (u,v) \approx (u,v')}$. Further, $V_{u,v}$ forms a connected "subtree@tree" of $T$, and any vertex $v'$ of $T$ with a "bag" containing $[u,v]_\approx$ must be such that $v' \in V_{u,v}$. Hence, $(T,\bagmap)$ verifies condition "C@@treedec".

        Finally, since every "bag" contains at most $k$ elements, the "tree-width" of $(T,\bagmap)$ is at most $k-1$.
    \end{proof}
    We show that, by construction, $\hat K$ is "$k$-bisimilar" to $K$.
    \begin{claim} 
        $K,u \pebblebisim{k} \hat K,\hat u$, where $\hat u = [u,r]_\approx$ and $r$ is the root of~$T$.
    \end{claim}
    \begin{proof}
        We show that, more generally, for any $k$-tuple $\bar u$ of "worlds" from $K$ 
        and any vertex of $T$ of the form $v = w \cdot \set{\bar u[1], \dotsc, \bar u[k]}$, we have
        $K,\bar u \pebblebisim{k} \hat K,\bar{\hat u}$,
        where $\bar{\hat u}[i] = [\bar u[i],v]_\approx$ for every $i$.

        Note that $\set{\bar u[i] \mapsto \bar{ \hat u}[i]}_{i\leq k}$ is a function, and further it is a bijection: if $\bar{\hat u}[i] = \bar{\hat u}[j]$ then $(\bar u[i],v) \approx (\bar u[j],v)$ which means that $\bar u[i] = \bar u[j]$. By definition of $\dbracket{{\rightarrow_a}}_{\hat K}$ in fact both $\set{\bar u[i] \mapsto \bar{\hat u}[i]}_i$ and $\set{\bar{\hat u}[i] \mapsto \bar u[i]}_i$ are "partial homomorphisms" from $\kHoms(K,\hat K)$ and $\kHoms(\hat K, K)$, respectively. In this way we
        verify the first conditions of 
        $K,\bar u \pebblequasibisim k \hat K,\bar{\hat u}$ and $ \hat K,\bar{\hat u}\pebblequasibisim k K,\bar{u}$.

        Let $u'_i$ in $K$ be at "distance" $\leq 1$ from some $\bar u[j]$. We can then show 
        $K,\bar u[i\mapcoord u'_i] \pebblebisim{k} \hat K,\bar{\hat u}[i\mapcoord \hat u'_i]$ for $\hat u'_i = [u'_i,v']_\approx$ and $v' = v \cdot (\set{\bar u[i'] \mid i'\neq i} \cup \set{u'_i}) \in V(T)$.
        Observe that for every $i' \neq i$ we have
        $\bar{\hat u}[i'] = [\bar u[i'],v]_\approx = [\bar u[i'],v']_\approx$
        and thus that $\bar {\hat u}[i\mapcoord \hat u'_i]$ is of the required form.

        We proceed symmetrically for any $\hat u'_i$ in $\hat K$ at "distance" $\leq 1$ from some $\bar{\hat u}[j]$: by definition of $\hat K$ there must be some $v$ and $u'_i$ such that $\hat u'_i = [u'_i,v]_\approx$, $\bar{\hat u}[j] = [\bar u[j],v]_\approx$, and  $u'_i$, $\bar u[j]$ are at "distance" $\leq 1$ in $K$. We then continue with 
        $K,\bar u[i\mapcoord u'_i]$ and
         $\hat K,\bar{\hat u}[i\mapsto \hat u'_i]$, of the required form.

        We can then repeat the same strategy ad aeternam, showing that $K,\bar u \pebblebisim{k} \hat K,\bar{\hat u}$.
    \end{proof}
    Finally, observe that $(T,\bagmap)$ and $\hat K$ are countable if $K$ is countable.
\end{proof}
\AP
As is the case of $\ICPDL$, $\CPDLp$ is definable in Least Fixed Point logic, and thus it inherits the Löwenheim-Skolem ""countable model property"": if a $\CPDLp$ "formula" is satisfiable, it is satisfiable in a countable "structure@@kripke".
We therefore obtain the following corollary from \Cref{prop:treewidth-k-unravelling} and \Cref{them:bisim-charac}.
\begin{corollary}\label{cor:treewidth-k-model-property}
    For every $k\geq 2$, $\CPDLg{\Tw}$ has the ``\,""$\Tw$-model property""'': if a formula $\phi \in \CPDLg{\Tw}$ is satisfiable, then it is satisfiable in a countable "Kripke structure" of "tree-width@@kripke" at most $k$. 
\end{corollary}

We can now proceed to the proof of \Cref{thm:separation}.
\begin{proof}[Proof of \Cref{thm:separation}]
Let $k\geq 2$.  Observe that the  "$(k+1)$-clique formula" $\xi_{k+1} \in \CPDLg{\Tw[k]}$  implies that there exists a directed $(k+1)$-clique of $a$'s starting from the current "world", and it is trivially satisfiable in a directed $(k+1)$-clique "structure@@kripke". 
However, there cannot be any $(k+1)$-cliques in a "structure@@kripke" of "tree-width" $\leq k-1$. Hence, in light of \Cref{cor:treewidth-k-model-property}, $\xi_k$ cannot be expressed in $\CPDLg{\Tw[k-1]}$.
\end{proof}

\begin{remark}
    While \Cref{prop:treewidth-k-unravelling} works for every $k \geq 2$, it does not follow that $\CPDLg{\Tw[1]}$ has the "$\Tw[1]$-model property", because the characterization between $\pebblebisim{k}$ and $\CPDLg{\Tw[k-1]}$  of \Cref{them:bisim-charac} holds only for $k\geq 3$. In fact, since $\CPDLg{\Tw[1]} \langsemequiv \CPDLg{\Tw[2]}$ by \Cref{thm:ICPDL_equals_TW1_equals_TW2}, the "$3$-clique formula" is actually expressible in $\CPDLg{\Tw[1]}$.
\end{remark}

\section{Satisfiability}
\label{sec:sat}
%//TODO: Acortar esta sección

The ""satisfiability problem"" for a fragment $\+C$ of $\CPDLp$ is the problem of, given a "formula" $\phi \in \+C$ whether there exists some "Kripke structure" $K$ and "world" $w$ thereof such that $K, w \models \phi$.

Observe that any countable "tree decomposition" can be turned into a "tree decomposition" whose underlying "tree" is binary, preserving the "width".
Decidability of the "satisfiability problem" for $\CPDLg{\Tw}$ and more generally for $\CPDLp$ follows readily from the previous \Cref{cor:treewidth-k-model-property}, the fact that $\CPDLp$ expressions can be effectively expressed in Monadic Second Order Logic (MSO), and the known result that satisfiability for MSO on structures of bounded "tree-width" is decidable by MSO interpretations onto "tree"-structures  (see eg~\cite[Fact~2]{DBLP:journals/apal/Seese91}) and Rabin's theorem. Since every $\CPDLp$ "formula" belongs to $\CPDLg{\Tw}$ for some $k$, it follows that $\CPDLp$ is decidable as well.

\begin{proposition}
    $\CPDLp$-"satisfiability@satisfiability problem" is decidable.
\end{proposition}

\paragraph*{Complexity}
Next, we will pin down the complexity of the satisfiability problem.
We follow closely the proof from \cite[Theorem~4.8]{DBLP:journals/jsyml/GollerLL09} of satisfiability for $\ICPDL$ and adapt it to the more general case of $\CPDLp$. 
In order to make clear the reuse of material, all definitions and notations are either verbatim or follow closely the ones from \cite{DBLP:journals/jsyml/GollerLL09}.
We will make clear the parts in which our proof diverges.

The proof goes by a series of reductions:
\begin{itemize}
    \item From the "satisfiability problem" of $\CPDLg{\Tw}$ on arbitrary "structures@@kripke" to "structures@@kripke" of "tree-width@@kripke" $k$ (trivial from \Cref{cor:treewidth-k-model-property}).
    \item From the "satisfiability problem" of $\CPDLg{\Tw}$ on "structures@@kripke" of "tree-width@@kripke" $k$ to the "satisfiability problem" of $\CPDLg{\Tw}$ on ``regular trees''.
    \item From the "satisfiability problem" of $\CPDLg{\Tw}$ on regular trees to the emptiness problem of an automata model on "trees".
\end{itemize}

\subsection{$\omega$-regular tree satisfiability}
The "$\omega$-regular tree satisfiability" is the problem of deciding, given an $\CPDLp$ "formula" $\phi$ and a "two-way alternating parity tree automaton" ("TWAPTA") $\+T$, whether there is an infinite "tree" in $L(\+T )$ that, when viewed as a "Kripke structure", is a model of $\phi$.

\AP
Let $\intro*\SigmaN$ be a finite node alphabet and $\intro*\SigmaE$ a finite edge alphabet. 
A ""$\SigmaN$-labeled $\SigmaE$-tree"" is a partial function $T : \SigmaE^* \to 
\SigmaN$ whose domain, denoted $dom(T)$, is prefix-closed. The elements of $dom(T)$ are 
the nodes of $T$. If $dom(T) = \SigmaE^*$, then $T$ is called ""complete"". 
\AP
\begin{toappendix}
% If a ""$\SigmaN$-labeled $\SigmaE$-tree"" $T$ is such that 
In the context of a "complete" tree $T$, 
a node $va \in dom(T)$ with $a \in \SigmaE$, is called the ""$a$-successor"" of $v$, and $v$ is the ""$a$-predecessor"" of $va$. 
\end{toappendix}
In the rest of the section we work with "complete" trees.
\AP
We use $\intro*\tree(\SigmaN, \SigmaE)$ to denote the set of all "complete" "$\SigmaN$-labeled $\SigmaE$-trees". 
If $\SigmaE$ is not important, we simply talk of \reintro{$\SigmaN$-labeled trees}.

The trees accepted by "TWAPTAs" are "complete" "$\pset{P}$-labeled $A$-trees", where $P \subseteq \Prop$ and $A \subseteq \Prog$ are finite sets of "atomic propositions" and "atomic programs", respectively. Such a tree $T$ can be identified with the "Kripke structure" $(A^*, \set{\rightarrow_a \mid a \in A}, \set{T_p \mid p \in P})$, where ${\rightarrow_a} = \set
{ (u, ua) \mid u \in A^*}$ for all $a \in A$ and $T_p = \set{ u \in A^* \mid  p \in T(u) }$ for $p \in P$. Observe that "Kripke structures" derived in this way are deterministic and total 
with respect to $A$, \ie,  the transition relation $\rightarrow_a$ is a total function 
for all $a \in A$.

"TWAPTAs" are walking automata which, at each transition, can read the current node label from $\SigmaN$ and move to a child checking an edge label from $\SigmaE$, move to the parent by checking an edge label from $\SigmaE$, or stay at the same node. Further, the automata model features alternation, so transitions are "positive Boolean combinations@\posB" of these kinds of moves. The acceptance is based on a "parity condition@successful". \ifarxiv Due to space constraints, and since our arguments do not need to handle details relative to its definition, we defer its definition to \Cref{sec:def:twapta}.\fi

For a "TWAPTA" $\+T$ and state $s$ thereof, we define $\dbracketaut{\+T , s}$ to be the set of all pairs $(T, u)$ such that $T \in \tree(\SigmaN , \SigmaE )$ and $u \in \SigmaE^*$ so that there exists a "successful" run of $\+T$ on $T$ starting from state $s$ and vertex $u$\ifarxiv \ (details in \Cref{sec:def:twapta}). \else . \fi
The language $L(\+T)$ accepted by $\+T$ is defined as $\set{ T \in \tree(\SigmaN , \SigmaE ) \mid (T, \epsilon) \in \dbracketaut{\+T,s_0} }$ where $s_0$ is the initial state of $\+T$.
\begin{toappendix}
\subsection{Definition of Two-Way Alternating Parity Tree Automata}
\label{sec:def:twapta}
\AP
To define "TWAPTAs", we need a few preliminaries. For a finite set $X$, we denote by $\intro*\posB(X)$ the set of all positive Boolean formulas where the elements of $X$ are used as variables. The constants \textit{true} and \textit{false} are admitted, \ie\ we have $\textit{true}$, $\textit{false}$ $\in \posB (X)$ for any set $X$. A subset $Y \subseteq X$ can be seen as a valuation in the obvious way by assigning $\textit{true}$ to all elements of $Y$ and $\textit{false}$ to all elements of $X \setminus Y$. For an edge alphabet $\SigmaE$, 
let $\overline{\SigmaE} = \set{ \bar a \mid a \in \SigmaE}$ be a disjoint copy of $\SigmaE$. For $u \in \Sigma^*$ and $d \in \SigmaE \cup \overline{\SigmaE} \cup \set \epsilon$ define:
\begin{align*}
u {\cdot} d &= \begin{cases}
    ud & \text{if $d \in \SigmaE$}\\
    u & \text{if $d = \epsilon$}\\
    v & \text{if there is $a \in \SigmaE$ with $d = \bar a$ and $u = va$}\\
    \textit{undefined} & \text{otherwise}
\end{cases}
\end{align*}

\AP
A ""two-way alternating parity tree automaton"" (""TWAPTA"") over "complete" $\SigmaN$-labeled $\SigmaE$-trees is a tuple $T = (S, \delta, s_0 , \textit{Acc})$, where

\begin{itemize}
    \item $S$ is a finite non-empty set of states,
    \item  $\delta : S \times \SigmaN \to \posB(mov(\SigmaE))$ is the transition function, where $mov(\SigmaE) = S \times (\SigmaE \cup \overline{\SigmaE} \cup \set{ \epsilon } )$ is the set of moves,
    \item  $s_0 \in S$ is the initial state, and
    \item  $\textit{Acc} \mid S \to \Nat$ is the priority function.
\end{itemize}

For $s \in S$ and $d \in \SigmaE \cup \overline{\SigmaE} \cup \set\epsilon$ , we write the corresponding move as $\tup{s, d}$. Intuitively, a move $\tup{s, a}$ , with $a \in \SigmaE$, means that the automaton sends a copy of itself in state $s$ to the "$a$-successor" of the current tree node. Similarly, $\tup{s,\bar a}$
means to send a copy to the "$a$-predecessor" (if existing), and $\tup{s,\epsilon}$ means to stay in the current node. Formally, the behaviour of "TWAPTAs" is defined in terms of runs. Let $\+T$ be a "TWAPTA" as above, $T \in \tree(\SigmaN , \SigmaE)$, $u \in \SigmaE^*$ a node, and $s \in S$ a state of $\+T$. \AP An ""$(s,u)$-run"" of $\+T$ on $T$ is a (not necessarily "complete") $(S \times \SigmaE^*)$-labeled tree $T_R$ such that
\begin{itemize}
    \item $T_R(\epsilon) = (s,u)$, and
    \item for all $\alpha \in dom(T_R)$, if $T_R (\alpha) = (p, v)$ and $\delta(p, T(v)) = \theta$, then there is a subset $Y \subseteq mov(\SigmaE)$ that satisfies $\theta$ and such that for all $(p', d) \in Y$, $v \cdot d$ is defined and there exists a successor $\beta$ of $\alpha$ in $T_R$ with $T_R(\beta) = (p', v \cdot d)$.
\end{itemize}

\AP
We say that an "$(s,u)$-run" $T_R$ is ""successful"" if for every infinite path $\alpha_1 \alpha_2 \dotsb$ in $T_R$ (which is assumed to start at the root), the following number is even:
\begin{align*}
    \min \{ \textit{Acc}(s) \mid {}& s \in S \text{ with } T_R (\alpha_i ) \in \set{ s } \times \SigmaE^* 
    \text{ for infinitely many } i \}
\end{align*} 

\AP
For $s \in S$ define
$\intro*\dbracketaut{\+T , s} \eqdef \{ (T, u) \mid T \in \tree(\SigmaN , \SigmaE ), u \in \SigmaE^*$ , and
there exists a "successful" "$(s,u)$-run" of $\+T$ on $T$ $\}$ and
$\dbracketaut{\+T} \eqdef \dbracketaut{\+T , s_0}$.

Now the language $L(\+T)$ accepted by $\+T$ is defined as
\begin{align*}
L(\+T) \eqdef \set{ T \in \tree(\SigmaN , \SigmaE ) \mid (T, \epsilon) \in \dbracketaut{\+T} }.
\end{align*}

\AP
For a "TWAPTA" $\+T = (S, \delta, s_0 , \textit{Acc})$, we define its size $|\+T | \eqdef | S |$ as its number of states and we define its ""index"" $\intro*\indexAut{\+T}$ as $\max \set{ \textit{Acc}(s) \mid s \in S }$. The size $| \delta |$ of the transition function $\delta$ is the sum of the lengths of all positive Boolean functions that appear in the range of $\delta$.
\end{toappendix}
We use the following upper bound for the emptiness problem for "TWAPTAs", where the "index" $\indexAut{\+T}$ refers to the number of indices used in the parity condition.
\begin{theorem}(\cite{Var89}, \cite[Theorem~3.1]{DBLP:journals/jsyml/GollerLL09})\label{thm:TWAPTA-complexity}
For a given "TWAPTA" $\+T$ with transition function $\delta$, it can be checked in time $\expfun( | \+T | + \indexAut{\+T}) \cdot | \delta |^{O(1)}$ whether $L(\+T ) \neq \emptyset$. Where $\intro*\expfun(n)$ means $2^{n^{O(1)}}$.
\end{theorem}

We can now formally define "$\omega$-regular tree satisfiability". Let $\phi$ be an $\CPDLp$ "formula", let $A = \set{ a \in \Prog \mid a \text{ occurs in }\phi}$ and $P = \set{ p \in \Prop \mid p \text{ occurs in } \phi}$. 
\AP
The "formula" $\phi$ is ""satisfiable with respect to@@twapta"" a "TWAPTA" $\+T$ over "$\pset{P}$-labeled $A$-trees" if there is $T \in L(\+T )$ such that $\epsilon \in \dbracketaut{\phi}_T$. 
%//TODO knowledge L(T)
\AP
Finally, ""$\omega$-regular tree satisfiability"" is the problem to decide, given such $\phi$ and $\+T$, whether $\phi$ is "satisfiable with respect to@@twapta" $\+T$.

\subsection{Solving $\omega$-regular tree satisfiability}
\label{sec:solving-omega-reg-sat}

We prove that the "$\omega$-regular tree satisfiability" problem is in $2$\exptime by an exponential time reduction to the non-emptiness problem for "TWAPTAs". 
% \knowledgenewrobustcmd{\ICPDLp}{\ensuremath{\cmdkl{\textup{ICPDL}^{\!+}}}\xspace}
We will actually work with $\ICPDLp$ instead of $\CPDLp$ to simplify the reuse of the construction of \cite{DBLP:journals/jsyml/GollerLL09} for $\ICPDL$.

The main ingredient of the reduction is a mutual inductive translation of (i) $\ICPDLp$ "formulas" into "TWAPTAs" and (ii) of $\ICPDLp$ "programs" into a certain kind of non-deterministic automata ("NFAs") introduced in \cite{DBLP:journals/jsyml/GollerLL09}. The latter are standard NFAs on words which navigate a "complete" "$\SigmaN$-labeled tree" reading symbols from $\set{ a, \bar a \mid a \in \SigmaE }$. They can also make conditional $\epsilon$-transitions, which are executable only if the current tree node is accepted by a given "TWAPTA". The definition of "NFAs" over "TWAPTAs" is relegated to \ifarxiv \Cref{sec:def:NFA-over-TWAPTA}. \else the "full version".\fi

\begin{toappendix}
\subsection{Definition of "Non-deterministic Finite Automata@NFA" over "TWAPTAs"}
\label{sec:def:NFA-over-TWAPTA}
\AP
Formally, a ""non-deterministic finite automaton"" (""NFA"") $\+A$ over a "TWAPTA" $\+T$  with statespace $S$
% $\+T = (S, \delta, s_0 , \textit{Acc})$ 
is a tuple $(Q, p_0, q_0, \rightarrow_{\+A})$, where $Q$ is a finite set of states, $p_0$ and $q_0$ are two selected states, called ""initial@@NFA"" and ""final@@NFA"" respectively, and $\rightarrow_{\+A}$ is a set of labeled-transitions of the following forms, where $q, q' \in Q$ and $a \in \SigmaE$: (i) $q \xrightarrow{a}_{\+A} q'$, (ii) $q \xrightarrow{\bar a}_{\+A} q'$ or (iii) $q \xrightarrow{\+T,s} q'$  with $s \in S$.

\AP
Transitions of the third kind are called ""test transitions"". 
"NFAs" define binary relations on the set of nodes of a "complete" "$\SigmaN$-labeled $\SigmaE$-tree". To make this explicit, let $T \in \tree(\SigmaN , \SigmaE )$ and define ${\Rightarrow_{\+A,T}} \subseteq (\SigmaE^* \times Q) \times (\SigmaE ^* \times Q)$ as the smallest relation such that for all $u \in \SigmaE^*$ , $a \in \SigmaE$ , $p, q \in Q$, and $s \in S$, we have:
%
% \begin{align*}
    (i) $(u,p) \Rightarrow_{\+A,T} (ua,q) \text{ if } p \xrightarrow{a}_{\+A} q$,
    (ii) $(ua,p) \Rightarrow_{\+A,T} (u,q) \text{ if } p \xrightarrow{\bar a}_{\+A} q$,
    (iii) $(u,p) \Rightarrow_{\+A,T} (u,q) \text{ if } p \xrightarrow{\+T,s}_{\+A} q \text{ and }
    (T,u) \in \dbracketaut{\+T,s}$.
% \end{align*}

\AP
Define
$\intro*\dbracketnfa{\+A} \eqdef \set{ (T, u, v) \mid T \in \tree(\SigmaN , \SigmaE )\text{, }u, v \in \SigmaE^*, \text{ and } (u, p_0) \Rightarrow_{\+A,T}^* (v, q_0)}$.

\bigskip
\end{toappendix}

For a "NFA" $\+A$ over a "TWAPTA" $\+T$, we define $\dbracketnfa{\+A}$ as the set of all triplets $(T,u,v)$ such that $T \in \tree(\SigmaN , \SigmaE )$, $u, v \in \SigmaE^*$, and there is an accepting run from $u$ to $v$ of $\+A$ on $T$ using $+T$ "tests@test transitions" from $\+T$\ifarxiv \ (details in \Cref{sec:def:NFA-over-TWAPTA}).\else . \fi

% When considering an "NFA" over a certain "TWAPTA" $\+T$ , the initial state of $\+T$ is obviously useless. Thus, in the context of "NFAs", a "TWAPTA" will henceforth only be a 3-tuple.

\paragraph{From $\ICPDLp$ to automata.} 
Fix a finite set of "atomic propositions" $P \subseteq \Prop$ and "atomic programs" $A \subseteq \Prog$ over which $\ICPDLp$ "formulas" and "programs" are built. In our context of trees, for $\ICPDLp$ "formulas" $\psi$ and "programs" $\pi$, let us define 
\AP%
$\phantomintro\dbrackett{}$%
$\dbrackett{\psi}$ and $\dbrackett{\pi}$ as follows:
\ifarxiv
\begin{align*}
    \intro*\dbrackett{\psi} &\eqdef \set{ (T, u) \mid T \in \tree(\pset{P} , A), u \in A^*, u \in \dbracket{\psi}_T },\\
    \reintro*\dbrackett{\pi} &\eqdef \set{ (T, u, v) \mid T \in \tree(\pset{P} , A), u, v \in A^*, (u, v) \in \dbracket{\pi}_T }.
\end{align*}
\else
\begin{align*}
    \intro*\dbrackett{\psi} &\eqdef \set{ (T, u) \mid T \in \tree(\pset{P} , A), u \in A^*, u \in \dbracket{\psi}_T },\\
    \reintro*\dbrackett{\pi} &\eqdef \set{ (T, u, v) \mid T \in \tree(\pset{P} , A), u, v \in A^*, \\
    &\hspace{6.4em} (u, v) \in \dbracket{\pi}_T }.
\end{align*}
\fi

The aim of this section is to show how to convert
\begin{itemize}
    \item each "formula" $\psi$ into a "TWAPTA" $\+T ( \psi )$ such that $\dbracketaut{\+T ( \psi )} = \dbrackett{\psi}$ and
    \item each "program" $\pi$ into a "TWAPTA" $\+T (\pi)$ and an "NFA" $\+A(\pi)$ over $\+T (\pi)$ such that $\dbracketnfa{\+A(\pi)} = \dbrackett{\pi}$.
\end{itemize}

All automata work over "$\pset{P}$-labeled $A$-trees". The construction is by induction on the structure of $\psi$ and $\pi$. 
The definition of the "TWAPTA" $\+T( \psi )$ for each "formula" $\psi$ is as in \cite[\S 3.2]{DBLP:journals/jsyml/GollerLL09}.
We only need to describe the inductive construction of $\+A(\pi)$ and $\+T (\pi)$ for a $\ICPDLp$ "program" $\pi$. If $\pi$ is of the form $a$, $\bar a$, $\psi ?$, $\pi_1 \cup \pi_2$, $\pi_1 \cap \pi_2$, $\pi_1 \circ \pi_2$, or $\chi^*$, we produce the construction as in \cite[\S 3.2]{DBLP:journals/jsyml/GollerLL09}.

\subsubsection{The case of "conjunctive programs"}
\label{subsec:sat:conjunctiveprog}
We are left with the case of a "conjunctive program" of the form $\pi = C[x_s,x_t]$. By inductive hypothesis, for every "atom" $\pi'(x,x') \in C$ there is a "TWAPTA" $\+T (\pi')$ and an "NFA" $\+A(\pi')$ over $\+T (\pi')$ such that $\dbracketnfa{\+A(\pi')} = \dbrackett{\pi'}$. Let us assume, without loss of generality, that the sets of states of the automata for the "atoms" are pairwise disjoint. Hence, we will rather refer to them as $\+T (\pi'(x,x'))$ and $\+A(\pi'(x,x'))$. 
\AP
Let $\intro*\allStates$ be the (disjoint) union of the states $\+A(\pi'(x,x'))$ for every $\pi'(x,x') \in C$.

\knowledgenewrobustcmd{\progStates}[2]{\cmdkl{#1 {\dashrightarrow
} #2}}%
For any $\pi'(x,x') \in C$  and states $q,q'$ of $\+A(\pi'(x,x'))$, we will consider a new sort of "program" of the form 
\AP
\[\intro*\progStates{p}{q}\]
with the semantics $\dbracket{\progStates{p}{q}} \eqdef \dbracketnfa{\+A_{q,q'}(\pi'(x,x'))}$, where $\+A_{q,q'}(\pi'(x,x'))$ is the result of setting $q$ and $q'$ to be the "initial@@NFA" and "final@@NFA" states, respectively, of the "NFA" $\+A(\pi'(x,x'))$ over the "TWAPTA" $\+T(\pi'(x,x'))$. Note that, formally, $\progStates{p}{q}$ depends also on $\+A(\pi'(x,x'))$ and $\+T(\pi'(x,x'))$, but since we have assumed state-disjointness, the automata can be unequivocally recovered, and are then implicit in the definition. 
\AP
An ""$\ICPDL$ program over $\allStates$"" is any "program" described by the syntax 
\begin{align*}
    \pi &\eqqdef \pi \star \pi \quad \text{ for } \star \in \set{\cap,\cup,\circ}\\
    \pi &\eqqdef \progStates p q \quad \text{ for } p,q \in \allStates
\end{align*}
with the expected semantics.

We show next how to translate $C[x_s,x_t]$ into an "equivalent" "$\ICPDL$ program over $\allStates$". 
% The resulting program will be an exponential union of polynomial $\ICPDL$ programs over the programs of the form $\progStates{q}{q'}$. 
From there, it suffices to apply the translation of the previous cases ($\cap$, $\circ$, $\cup$) \cite[\S 3.2]{DBLP:journals/jsyml/GollerLL09} into a suitable "NFA".

\AP
\knowledgenewrobustcmd{\pathinT}{\mathrel{\cmdkl{\leadsto_{T_C}}}}%
Consider any binary "tree" $T_C$ having $\vars(C)$ as vertices. For any two variables $x,y \in \vars(C)$, let us write $x \intro*\pathinT y$ to denote the (unique) simple path from $x$ to $y$ in $T_C$. 
Let $f : \vars(C) \to \pset{\allStates}$ be a function mapping each "variable" $y$ to a set of states $Q_y \subseteq \allStates$ so that for every "atom" $\pi'(x,x') \in C$:
\begin{itemize}
    \item if the path $x \pathinT x'$ contains $y$, then $Q_y$ contains exactly one state from $\+A(\pi'(x,x'))$,
    \item otherwise, $Q_y$ contains no state from $\+A(\pi'(x,x'))$.
\end{itemize}
The intuition is that with $T_C$ we guess the `shape' in which $C$ will be mapped to the tree, and with $f$ we guess intermediate states in the paths witnessing each "atom". 
\AP
\knowledgenewrobustcmd{\ShapesC}{\cmdkl{\+S}}%
Let $\intro*\ShapesC$ be the set of all pairs $(T_C,f)$ verifying the conditions above.

For any two variables $y,y' \in \vars(C)$ such that $y'$ is a child of $y$ in $T_C$, and any two states $q  \in f(y)$ and $q' \in f(y')$ belonging to some common "NFA" $\+A(\pi'(x,x'))$, we have that the path $x \pathinT x'$ either (a) visits first $y$ and then $y'$, or (b) visits first $y'$ and then $y$. We then associate to such tuple $(y,y',q,q')$ the "atom" $\rho(y,y')$ with $\rho = \progStates{q}{q'}$ (if (a) holds), or the "atom" $\rho(y',y)$ with $\rho = \progStates{q'}{q}$ (if (b) holds).

For a given choice of $T_C$ and $f$, consider $C_{T_C,f}$ as the set of all the "atoms" described in this way. Observe that 
\begin{enumerate}
    \item $\vars(C_{T_C,f}) = \vars(C)$, 
    \item the number of "atoms" of $C_{T_C,f}$ is at most polynomial in the number of "atoms" of $C$ (in fact, quadratic), and \label{it:quadratic-atoms} 
    \item $\bigcup_{(T_C,f) \in \ShapesC}\dbrackett{C_{T_C,f}[x_s,x_t]} = \dbrackett{C[x_s,x_t]}$,
    \item the "underlying graph@@C" of $C_{T_C,f}$ is of "tree-width" 1.
\end{enumerate}
Since for each $\hat C = C_{T_C,f}$ with $(T_C,f) \in \ShapesC$ we have that $\uGraphC{\hat C}$ is of "tree-width" 1, it follows that  $\uGraph{\hat C[x_s,x_t]} \in \Tw[2]$. Hence, we can construct in polynomial time an "equivalent" "$\ICPDL$ program over $\allStates$", 
\AP
call it $\intro*\translatedC$, of polynomial size  (\cf~\Cref{thm:ICPDL_equals_TW1_equals_TW2}). 
% The translation is effective and in polynomial time (even in \logspace).
%
Therefore, $C[x_s,x_t]$ is equivalent to an exponential union of polynomial-sized "$\ICPDL$ programs over $\allStates$", one $\translatedC$ for each $(T_C,f) \in \ShapesC$.
% the "NFAs" for the $\+A(\pi(x,x'))$'s \sidesanti{$\+A(\pi'(x,x'))$?}atoms of $C$. 
We can then apply the construction for $\ICPDL$ \cite[\S 3.2]{DBLP:journals/jsyml/GollerLL09} to obtain the desired "TWAPTA" $\+T (C[x_s,x_t])$ and an "NFA" $\+A(C[x_s,x_t])$ over $\+T (C[x_s,x_t])$ such that $\dbracketnfa{\+A(C[x_s,x_t])} = \dbrackett{C[x_s,x_t]}$.

\subsubsection{Automata size}
\AP
The ""intersection width""  $\intro*\iwidth(\pi)$ of an $\ICPDL$ "program" $\pi$ is defined in \cite{DBLP:journals/jsyml/GollerLL09} as follows:
\begin{align*}
\iwidth(a) &= \iwidth(\bar a) \eqdef 1 \text{ for all } a \in \Prog, %\\
\qquad \iwidth(\theta?) \eqdef 1,\\
\iwidth(\pi_1 \cup \pi_2) &=
\iwidth(\pi_1 \circ \pi_2) \eqdef
\max \set{ \iwidth(\pi_1), \iwidth(\pi_2)},\\
\iwidth(\pi^*) &\eqdef \iwidth(\pi),%\\
\qquad
\iwidth(\pi_1 \cap \pi_2) \eqdef \iwidth(\pi_1) + \iwidth(\pi_2).
\end{align*}
\AP
This is generalized to any "program" or "formula" $\alpha$ by defining $\intro*\Iwidth(\alpha)$ to be the maximum "intersection width" of a "program" therein (or 1 if it contains no "program").
The interest of $\Iwidth$ stems from the fact that, as shown in \cite{DBLP:journals/jsyml/GollerLL09}, for every $\ICPDL$ "expression" $\alpha$, the constructed "TWAPTA" is exponential only in $\Iwidth(\alpha)$:
\begin{lemma}(\cite[Lemma 3.7]{DBLP:journals/jsyml/GollerLL09})\label{lem:ICPDL-automata-size}
    % For every $\ICPDL$ "program" $\pi$ and $\ICPDL$ "formula" $\psi$:
    % \begin{itemize}
    %     \item $| \+T (\pi) | \leq 9 \cdot | \pi |^{2\cdot\Iwidth(\pi)+1} \text{ and }\indexAut{\+T (\pi)} \leq | \pi |$.
%
    %     \item $| \+T ( \psi )| \leq  9 \cdot |\psi|^{2\cdot\Iwidth( \psi )+1}\text{ and }\indexAut{\+T ( \psi )} \leq |\psi|$.
    % \end{itemize}
    For every $\ICPDL$ "expression" $e$, 
    $| \+T (e) | \leq 9 \cdot | e |^{2\cdot\Iwidth(e)+1}$ and $\indexAut{\+T (e)} \leq | e |$.
\end{lemma}
 
\begin{remark}\label{rem:index-better}
    \AP
    By simple inspection of the construction of \cite{DBLP:journals/jsyml/GollerLL09}, the previous bound can be improved to $\indexAut{\+T (e)} \leq {\negdepth(e)}$, where $\intro*\negdepth(e)$ denotes the maximum number of nested negations in the "expression" $e$ (plus one).
\end{remark}

In our construction for a $\ICPDLp$ "expression" we obtain an "intersection width" polynomially bounded since at each step we translate a "conjunctive program" into a union of polynomial-sized "$\ICPDL$ programs over" the already built automata.
Concretely, the "intersection width" of each $\translatedC$ above is at most the number of "atoms" of $C_{T_C,f}$, hence quadratic in the number of "atoms" of $C$.
Further, the construction does not change the number of nested negations. Therefore, by \Cref{lem:ICPDL-automata-size}, we obtain, for each $\ICPDLp$ "expression" $\alpha$, a "TWAPTA" $\+T(\alpha)$ which is of single exponential size.

\AP
Let us denote by the ""conjunctive width"" $\intro*\cqsize{\pi}$ of a $\CPDLp$ "program" $\pi$ to the following measure:
\begin{align*}
    \cqsize{a} &= \cqsize{\bar a} \eqdef 1 \text{ for all } a \in \Prog,
    \qquad
    \cqsize{\theta?} \eqdef 1,\\
    \cqsize{\pi_1 \cup \pi_2} &= 
    \cqsize{\pi_1 \circ \pi_2} \eqdef
    \max \set{ \cqsize{\pi_1}, \cqsize{\pi_2}},\\
    \cqsize{\pi^*} &\eqdef \cqsize\pi,
    \qquad
    \cqsize{C[x_s,x_t]} \eqdef \sum_{\pi(x,y) \in C}\cqsize\pi.
\end{align*}
We define $\intro*\Cqsize{\alpha}$
$\phantomintro*\Cqsizealt$%
for any "expression" $\alpha$ to be the maximum "conjunctive width" of a "program" therein (or 1 if it contains no "program").

In the previous reduction from a "conjunctive program" to a union of "$\ICPDL$ programs over $\allStates$", observe that each "$\ICPDL$ program over $\allStates$" is of polynomial size, and thus polynomial in the "conjunctive width" of the original "formula". Further, the "intersection width" is polynomial in the "conjunctive width". Combining this observation with the size bounds of \Cref{lem:ICPDL-automata-size} plus \Cref{rem:index-better}, and the complexity of "TWAPTA" of \Cref{thm:TWAPTA-complexity} we obtain the following\ifarxiv \ (details in \Cref{sec:sat-auotmatabounds}).\else . \fi

\begin{toappendix}
\subsection{Automata bounds}
\label{sec:sat-auotmatabounds}
First we show that from the reduction from a "conjunctive program" to a union of "$\ICPDL$ programs over $\allStates$", we get the following bounds.
\begin{claim}\label{cl:size-CPDLplus-ICPDL}
    Let $\pi' = \bigcup_{T_C,f} \translatedC$ be the "$\ICPDL$ program over" previously built automata produced from an $\ICPDLp$ "conjunctive program" $\pi$. Then: (1) the size\footnote{By `size' we mean without taking into account the sizes of the "NFAs"/"TWAPTAs" over which it is based.} of $\pi'$ is in $\expfun(\cqsize{\pi})$; (2) $\iwidth(\pi') \leq \cqsize{\pi}^{O(1)}$; and (3) ${\negdepth(\pi')} = {\negdepth(\pi)}$.
\end{claim}
\begin{proof}
    Note that if $\pi = C[x_s,x_t]$, each $\translatedC$ is of polynomial size $\cqsize{\pi}^{O(1)}$ where remember $\cqsize{\pi}$ of "atoms" of $C$, as the result of the polynomial time translation of \Cref{thm:ICPDL_equals_TW1_equals_TW2} applied to the "conjunctive program" $C_{T_C,f}$ of quadratic size.  Then, there are at most $2^{\cqsize{\pi}^{O(1)}}$ such "programs" $\translatedC$, each "program" being of polynomial size $\cqsize{\pi}^{O(1)}$. This yields a bound of $\cqsize{\pi}^{\cqsize{\pi}^{O(1)}}$ which is in $\expfun(\cqsize{\pi})$.
    The fact that $\iwidth(\pi')$ is polynomial in $\pi$ follows from the fact that each $\translatedC$ is of polynomial size, and that $\iwidth(\bigcup_{T_C,f} \translatedC)$ is $\max_{T_C,f}\set{\iwidth(\translatedC)}$.
\end{proof}

Combining the size bound of \Cref{lem:ICPDL-automata-size} with \Cref{rem:index-better}, \Cref{cl:size-CPDLplus-ICPDL} and the complexity of "TWAPTA" of \Cref{thm:TWAPTA-complexity} we obtain:
\end{toappendix}
\begin{propositionrep}\label{thm:omega-tree-sat}
    For a "TWAPTA" $\+T$ with transition function $\delta$ and a $\CPDLp$ "formula" $\phi$, we can decide in time
    \[\expfun\big(|\+T| + \indexAut{\+T} + |\phi|^{\Cqsize{\phi}} \big) \cdot |\delta|^{O(1)}\] 
    whether there exists some $T \in L(\+T)$ such that $\epsilon \in \dbracket{\phi}_T$.
    Hence, the "$\omega$-regular tree satisfiability" problem for $\CPDLp$ is in $2$\exptime, and the "$\omega$-regular tree satisfiability" problem for $\set{\phi \in \CPDLp \mid \Cqsize{\phi} \leq n}$ is in \exptime, for any $n \in \Nat$.
\end{propositionrep}
\begin{proof}
    We first show, by induction on $\phi$, the following bound on the constructed "TWAPTA":
    \begin{claim}\label{cl:automatabound}
        $| \+T (e) | \leq 9 \cdot | e |^{2\cdot \Cqsize{e}+1}$, and $\indexAut{\+T (e)} \leq  \negdepth(e) $.
    \end{claim}

\proofcase{Base case}
    For any "expression" $e$ which no "conjunctive programs", then $e$ is a $\CPDL$ "expression". Then the bound on the size of the equivalent "TWAPTA" of \Cref{lem:ICPDL-automata-size} applies where $\Iwidth(e)=1$.

    We obtain $| \+T (e) | \leq  9 \cdot | e |^{3}$, and $\indexAut{\+T (e)} \leq  \negdepth(e) $.

    \proofcase{Inductive case}
    If $e$ is a "conjunctive program" $C[x_s,x_t]$, for $C= \set{\pi_i(x_i,y_i) \mid 1 \leq i \leq n}$. 

    We have
    $| \+T (\pi_i) | \leq  9 \cdot | \pi_i |^{2\cdot \Cqsize{\pi_i}+1}$, and $\indexAut{\+T (\pi_i)} \leq  \negdepth(\pi_i) $ by inductive hypothesis, and
    $| \+T (e) | \leq (\sum_i | \+T (\pi_i) |) + |\+T(\hat e)|$ where $\hat e$ is the "$\ICPDL$ program over" the $\+T(\pi_i)$'s equivalent to $e$, as defined in \Cref{subsec:sat:conjunctiveprog}. By the bounds of \Cref{cl:size-CPDLplus-ICPDL}, $|\hat e| \leq 2^{\cqsize{e}^c}$ for some $c \in\Nat$. Hence, applying the automata bound of \Cref{lem:ICPDL-automata-size}  we have $|\+T(\hat e)| \leq 9 \cdot  2^{2 \cdot \cqsize{e}^c+1} \leq 9 \cdot  2^{2 \cdot \Cqsize{e}^c+1}$.

    Thus,
    \begin{align*}
        | \+T (e) | &\leq \big(\sum_i | \+T (\pi_i) |\big) + 9 \cdot  2^{2 \cdot \Cqsize{e}^c+1}    \\
        &\leq \big(\sum_i 9 \cdot | \pi_i |^{2\cdot \Cqsize{\pi_i}+1}\big) + 9 \cdot  2^{2 \cdot \Cqsize{e}^c+1}\\
        &\leq 9 \cdot \big(\sum_i  \cdot | \pi_i |^{2\cdot \Cqsize{e}+1} + \cdot  2^{2 \cdot \Cqsize{e}^c+1}\big)\\
        &\leq 9 \cdot \big(\sum_i  | \pi_i |\big)^{2\cdot \Cqsize{e}+1}\\
        &= 9 \cdot |e|^{2\cdot \Cqsize{e}+1}
    \end{align*}
    %//TODO En realiadad hay que hacerlo también para los NFAs... Ver cómo se prueba en Goller&al.
    This concludes the proof of \Cref{cl:automatabound}. Now, applying the complexity of "TWAPTA" of \Cref{thm:TWAPTA-complexity} with the bounds of \Cref{cl:automatabound} we obtain that emptiness for $\+T(\phi)$ can be checked in
    \begin{align*}
        \expfun( | \+T(\phi) | + \indexAut{\+T(\phi)}) \cdot | \delta |^{O(1)} &
        \leq \expfun( 9 \cdot |\phi|^{2\cdot \Cqsize{\phi}+1} + \negdepth(\phi)) \cdot | \+T(\phi) |^{O(1)} \\
        &\leq \expfun(|\phi|^{\Cqsize{\phi}})
    \end{align*}
    From this, it follows (again by \Cref{thm:TWAPTA-complexity}) that for a "TWAPTA" $\+T$ with transition function $\delta$, and a $\CPDLp$ "formula" $\phi$, testing emptiness for $\+T \cap \+T(\phi)$ can be decided in $\expfun\big(|\+T| + \indexAut{\+T} + |\phi|^{\Cqsize{\phi}} \big) \cdot |\delta|^{O(1)}$.
\end{proof}

\subsection{Reducing to $\omega$-regular tree satisfiability}
%//TODO Say something about the binary vs infinite branching of tree decompositions
This is essentially the reduction shown in \cite[\S 4.2]{DBLP:journals/jsyml/GollerLL09}.
The exact same reduction as \cite[Lemma 4.7]{DBLP:journals/jsyml/GollerLL09} (with a trivial adaptation changing "tree-width" $2$ to "tree-width" $k$) yields that for every $\CPDLp(\Tw)$ "formula" $\phi$ there is a polynomially computable $\CPDLp(\Tw)$ "formula" $\phi'$ and "TWAPTA" $\+T$ such that the following holds.
\begin{lemma}\label{lem:redux-omega-tree-sat}
    $\phi$ is satisfiable on a "structure@@kripke" of "tree-width" $k$ if{f} $\phi'$ is "satisfiable with respect to@@twapta" $\+T$. Moreover, $\Cqsize{\phi}=\Cqsize{\phi'}$.
\end{lemma}

The reduction is polynomial only if the "tree-width" is considered to be fixed, otherwise the reduction is exponential.

As a consequence of \Cref{lem:redux-omega-tree-sat}, \Cref{thm:omega-tree-sat}, and \Cref{cor:treewidth-k-model-property} we obtain the main result of this section:

\begin{theorem}\label{thm:sat-cpdlp}
    For every $k \geq 1$,
    \begin{enumerate}
        \item $\CPDLp$-"satisfiability@satisfiability problem" is decidable in $3$\exptime,
        \item $\CPDLg{\Tw}$-"satisfiability@satisfiability problem" is $2$\exptime-complete,
        \item $\set{ \phi \in \CPDLp \mid \Cqsize{\phi} \leq k}$-"satisfiability@satisfiability problem" is  \exptime-complete. 
    \end{enumerate}
\end{theorem}
The lower bound of item 2 of \Cref{thm:sat-cpdlp} follows from the known $2$\exptime lower bound of $\ICPDL$ satisfiability \cite[Theorem~2]{DBLP:journals/jsyml/LangeL05} combined with \Cref{thm:ICPDL_equals_TW1_equals_TW2}. The lower bound of item 3 follows from the \exptime lower bound of $\PDL$ \cite{DBLP:journals/jcss/FischerL79}.

Due to the following remark, the last item of \Cref{thm:sat-cpdlp} can be seen as a generalization of the result \cite[Theorem~4.8]{DBLP:journals/jsyml/GollerLL09} stating that the satisfiability problem for $\ICPDL$ "formulas" of bounded "intersection width"  is in \exptime.
\begin{remark}
    There is a polynomial time translation of $\ICPDL$ "formulas" of "intersection width" $k$ to $\CPDLg{\Tw[2]}$ "formulas" of "conjunctive width" $k$\ifarxiv \ (\eg, the translation $\tr_3$ from \S \ref{sec:from_ICPDL}).\else . \fi
\end{remark}

%!TEX root = main.tex

\section{Model Checking}
\label{sec:modelchecking}

We consider the (finite) ""model checking"" problem, as the problem of, given a finite "Kripke structure" $K$, a "world" $w \in \worlds{K}$, and a $\CPDLp$ "formula" $\phi$, whether $K,w \models \phi$. This problem is known to be \ptime-complete for $\PDL$, $\ICPDL$, and many other variants \cite{DBLP:journals/japll/Lange06,DBLP:journals/jcss/FischerL79}. 
In the unrestricted setting, it is easy to see that the "model checking" problem for $\CPDLp$ is hard for the complexity class $\DP= \set{ A \setminus B : A,B \in \np}$.

\begin{propositionrep}
    The "model checking" problem for $\CPDLp$ is \DP-hard, and in $\textup{P}^{\np}$.
\end{propositionrep}
\begin{proof}
    The lower-bound follows from a straightforward reduction from the "model checking" problem for Boolean combinations of "conjunctive queries". \sidediego{cite sth?}
    For the upper-bound, simply note that the algorithm in the proof of \Cref{thm:modelchecking} (\cref{it:modelch:ptime}) runs in $\textup{P}^{\np}$.
\end{proof}

However, provided the "tree-width" of $\CPDLp$ "formulas" is bounded, "model checking" remains \ptime.
\begin{theoremrep}\label{thm:modelchecking}
    For any class $\+G$ of connected "graphs":
    \begin{enumerate}
        % \item The "model checking" problem for $\CPDLp$ is \np-complete.
        \item \label{it:modelch:ptime}
        if $\+G \subseteq \Tw$ for some $k$, then the "model checking" problem for $\CPDLp(\+G)$ is \ptime-complete;
        % \item The "model checking" problem for $\CPDLp(\Tw[1])$ is in linear time and \ptime-complete.
        \item \label{it:modelch:notptime}
         otherwise, the "model checking" problem for $\CPDLp(\+G)$ is not in \ptime, under the hypothesis that $\+G$ is recursively enumerable and $\wone \neq \fpt$. This holds even for the "positive" fragment of $\CPDLp(\+G)$.
    \end{enumerate}
\end{theoremrep}
\begin{proofsketch}
    \proofcase{\ref{it:modelch:ptime}.}
The procedure for showing that $\CPDLp(\Tw)$ "model checking" is in \ptime is a classical dynamic programming algorithm. Given a "formula" $\phi$ and a finite "Kripke structure" $k$, we iteratively label all "worlds" of $K$ with "subexpressions" $\psi$ of $\phi$ based on the labeling of "subexpressions" of $\psi$. 

For processing a "program" $\pi \in sub(\phi)$ we exploit the fact that the evaluation of "conjunctive queries" of "tree-width" $\leq k$ is in polynomial time \cite[Theorem~3]{DBLP:journals/tcs/ChekuriR00} (more precisely "LOGCFL"-complete \cite[Theorem~6.12]{DBLP:journals/jacm/GottlobLS01}). This yields an algorithm which is linear in $\phi$ and polynomial in $K$, where the degree of the polynomial is $k+1$ if $\phi \in \CPDLg{\Tw}$. 

The lower bound comes from \ptime-hardness of "model checking" for modal logic (see, eg, \cite[Proposition 5]{DBLP:journals/japll/Lange06}).

\proofcase{\ref{it:modelch:notptime}.} This follows from the fact that a similar statement is known for "conjunctive queries": the "model checking" problem  for the class of Boolean "conjunctive queries" on binary relations whose underlying graph is in $\+G$ is not in \ptime unless $\fpt = \wone$ \cite[Corollary~19]{DBLP:conf/stoc/GroheSS01}. 
Then the result follows from a polynomial-time reduction from the evaluation of Boolean "conjunctive queries" to the "model checking" of $\CPDLp(\+G)$: A Boolean "conjunctive query" $q$ holds true in a (finite) "structure@@kripke" $K$ if{f} the "formula" $\tup{C[x,x]}$ is satisfied in some of its "worlds", where
$C$ is the set of "atoms" of $q$, and $x$ is any variable of~$q$. 
\end{proofsketch}
\begin{proof}
    \proofcase{\ref{it:modelch:ptime}.}
The procedure for showing that $\CPDLp(\Tw)$ "model checking" is in \ptime is a classical dynamic programming algorithm. Given a "formula" $\phi$ and a finite "Kripke structure" $k$, we iteratively label all "worlds" of $K$ with "subexpressions" $\psi$ of $\phi$ based on the labeling of "subexpressions" of $\psi$. 
% We assume an arbitrary linear order on the "worlds" of $K$.
% Lo mandé a la intro:
% Concretely, let $sub(\phi)$ be the set of all "subexpressions" of $\phi$ in the sense that it is the smallest set satisfying the following:
% \begin{itemize}
%     \item $\phi \in sub(\phi)$,
%     \item if $\psi \land \psi' \in sub(\phi)$ then $\set{\psi, \psi'} \subseteq sub(\phi)$
%     \item if $\lnot\psi \in sub(\phi)$ then $\psi \in sub(\phi)$
%     % \item if $\psi \in sub(\phi)$ and $\psi$ is not of the form $\lnot\psi'$, then $\lnot\psi \in \subseteq sub(\phi)$
%     \item if $\tup{\pi} \in sub(\phi)$, then $\pi \in sub(\phi)$
%     \item if $\pi \star \pi' \in sub(\phi)$, then $\set{\pi,\pi'} \subseteq sub(\phi)$ for every $\star \in \set{\circ,\cup}$,
%     \item if $\pi^* \in sub(\phi)$, then $\pi \in sub(\phi)$,
%     \item if $[\psi?] \in sub(\phi)$, then $\psi \in sub(\phi)$,
%     \item if $C[x_s,x_t] \in sub(\phi)$ and $\pi(x,y) \in C$, then $\pi \in sub(\phi)$.
% \end{itemize}
We will use a unary relation (\ie\ a set) $U_\psi \subseteq \worlds{K}$ for every "formula" $\psi \in sub(\phi)$ and a binary relation $B_\pi \subseteq \worlds{K} \times \worlds{K}$ for every "program" $\pi \in sub(\phi)$, which are all initialized in the empty relation.
We start by processing all "atomic propositions" $p \in \subexpr(\phi)$ and "atomic programs" $a \in \subexpr(\phi)$: we set $U_p = X_p$ and $B_{a} = {\to_a}$. For the converse of "atomic programs" $\bar a \in sub(\phi)$ we set $B_{\bar a} = (\to_a)^{-1}$.
Now take any "formula" $\psi \in sub(\phi)$ such that all its "subexpressions" from $sub(\psi) \setminus \set\psi$ have already been processed. We process $\psi$ as follows:
\begin{itemize}
    \item If $\psi = \psi_1 \land \psi_2$, then set $U_\psi = U_{\psi_1} \cap U_{\psi_2}$
    \item if $\psi = \lnot\psi'$, then set $U_{\psi} = \worlds{K} \setminus U_{\psi'}$,
    \item if $\psi = \tup{\pi}$, then set $U_{\psi} = \set{w \in \worlds{K} : \exists w' s.t.\ (w,w') \in B_{\pi}}$.
\end{itemize}
It is easy to see that each one of these operations is in polynomial time (even linear with the right data structure).
Now take any "program" $\pi \in sub(\phi)$ such that all its "subexpressions" from $sub(\pi) \setminus \set\pi$ have already been processed. We process $\pi$ as follows:
\begin{itemize}
    \item If $\pi = \pi_1 \star \pi_2$ for $\star \in \set{\circ,\cup}$, then set $B_\pi = B_{\pi_1} \star B_{\pi_2}$, 
    \item if $\pi = (\pi')^*$, then set $B_\pi = B_{\pi'}^*$,
    \item if $\pi = C[x_s,x_t]$, we evaluate $C[x_s,x_t]$ as if it were a "conjunctive query" on the already processed relations to populate $B_\pi$. 
\end{itemize}
It is clear that the first two items can be done in quadratic time. For the last item, it is well-known that the evaluation of "conjunctive queries" of "tree-width" $\leq k$ is in polynomial time \cite[Theorem~3]{DBLP:journals/tcs/ChekuriR00} (more precisely "LOGCFL"-complete \cite[Theorem~6.12]{DBLP:journals/jacm/GottlobLS01}), via a bottom-up processing of the "tree decomposition" of the query.
%, in a generalization of what is known as Yannakakis algorithm for acyclic queries \cite{DBLP:conf/vldb/Yannakakis81}.

All in all, this yields an algorithm which is linear in $\phi$ and polynomial in $K$, where the degree of the polynomial is $k+1$ if $\phi \in \CPDLg{\Tw}$. 
% Otherwise, the algorithm runs in non-deterministic polynomial time.

The lower bound comes from \ptime-hardness of "model checking" for modal logic (see, eg, \cite[Proposition 5]{DBLP:journals/japll/Lange06}).

\proofcase{\ref{it:modelch:notptime}.} This follows from the fact that a similar statement is known for "conjunctive queries": the "model checking" problem  for the class of Boolean "conjunctive queries" on binary relations whose underlying graph is in $\+G$ is not in \ptime unless $\fpt = \wone$ \cite[Corollary~19]{DBLP:conf/stoc/GroheSS01}.\footnote{The fact that here we have connected "conjunctive queries" plays no role, and the result of \cite{DBLP:conf/stoc/GroheSS01} holds also for connected "conjunctive queries"} 
% In fact, this result also holds for the case where the model is connected, in this setting, 
We have that  a Boolean "conjunctive query" $q$ holds true in a "structure@@kripke" if{f} the "formula" $\tup{C[x,x]}$ is satisfied in some of its "worlds", where
$C$ is the set of atoms of $q$, and $x$ is any variable of~$q$.

This \ptime-reduction from "model checking" of "conjunctive queries" to "model checking" of $\CPDLp(\+G)$ implies that the "model checking" of $\CPDLp(\+G)$ cannot be in \ptime, under the hypothesis that $\wone \neq \fpt$.\footnote{A similar parallel can be made for parameterized "model checking", where the parameter is the size of the "formula". In our case we would obtain that parameterized "model checking" of $\CPDLp(\+G)$ is \wone-hard.}
\end{proof}

% \begin{theorem}
%     The "model checking" problem for $\ICPDL$ is in linear time and \ptime-complete.
% \end{theorem}
% To the best of our knowledge, "model checking" problem for $\ICPDL$ had only been shown to be in polynomial (quadratic) time \cite{DBLP:journals/japll/Lange06}, while linear time "model checking" was shown only for $\PDL$ \cite{DBLP:journals/jcss/FischerL79}.

%!TEX root = main.tex
\section{Conclusions}

We have introduced an expressive logic $\CPDLp$, which captures several formalisms known on "Kripke structures" and "graph databases". 
We ignore if the complexity of satisfiability for $\CPDLp$ is complete for $3$\exptime, we leave this issue for future work.
While the "satisfiability" of $\CPDLp$ is decidable, the \emph{finite}-satisfiability is a notable open problem even for $\loopCPDL$ \cite[\S 7]{DBLP:journals/jsyml/GollerLL09}. \ifarxiv See also the discussion in \Cref{sec:appendixconcl}. \else \fi

%//TODO Agregar si/cuando haya espacio:
\begin{toappendix}
\section{Extended discussion}
\label{sec:appendixconcl}

\paragraph{Finite satisfiability} Observe that in the satisfiability problem that we study here, the "Kripke structure"  may be infinite. It is a well-known open problem, even for the case of "loop-CPDL", whether the \emph{finite}-satisfiability problem is decidable \cite[\S 7]{DBLP:journals/jsyml/GollerLL09}.

\paragraph{Infinite state model checking}
Infinite state model checking (\ie, "model checking" for some finitely represented infinite "Kripke structures") has been studied for $\CPDL$ \cite{DBLP:conf/csl/GollerL06} and $\ICPDL$ \cite{DBLP:journals/jsyml/GollerLL09}. While we believe that the upper bound results can be extended to $\CPDLp$, we leave this for future work.

\paragraph{Connectedness} The connectedness restriction that we assume on "conjunctive programs" is, as already mentioned, not essential. All upper bounds hold for non-connected "programs", and the characterization can be easily adapted to capture these more general kind of queries. The reason why we adopted the connectedness condition is to make the comparison with the $\PDL$ family of logics, which can only move inside one connected component of a "structure@@kripke".

\paragraph{Arbitrary arity programs} In our developments we have shown that we can consider not only binary "programs" $C[x,y]$ but actually arbitrary arity "programs" $C[\bar x]$. This can be seen as a query language capturing the whole language of "CRPQ" or even "Regular Queries" for arbitrary arity queries (as opposed to just queries of arity $\leq 2$ as depicted in \Cref{fig:expressive-power}).
One can also consider extending the composition, union, and Kleene star operators to $2n$-arity relations in the obvious way. This yields a more expressive formalism. We currently do not know whether the satisfiability problem is decidable, and we leave this direction for future work. Another direction would be to abandon Kleene star in favor of a more general form of transitive closure allowing ``parameters'', in the same way as in first-order logic with transitive closure (an idea which was explored in \cite{DBLP:conf/dlog/BourhisKR14} in a different context).

\paragraph{Constants} Since this work intends to capture some query languages such as "CQs" and "CRPQs", it would make sense to have access to \emph{constants}, which in this context are usually modelled as \emph{nominals} (\ie, a kind of "atomic propositions" which can hold true in at most one "world"). We believe that nominals can be treated easily in $\CPDLp$ at the expense of increasing the "tree-width" of the tree-like model property. We leave these details for future work.
\end{toappendix}

\section*{Acknowledgements}
This work was partially funded by the French-Argentinian IRP
\href{http://www.lia-sinfin.org/}{SINFIN}. Diego Figueira is partially funded by ANR QUID, grant ANR-18-CE400031. Santiago Figueira is partially funded by UBA, grant UBACyT 20020190100021BA.

\renewcommand{\appendixsectionformat}[2]{Missing details from Section~#1}
% \input{characterization_general}
% \appendix
% \section{Appendix}
% \input{translation_proofs}

\ifarxiv
  \bibliographystyle{ACM-Reference-Format}
  \bibliography{long,references}
\else
  \bibliographystyle{IEEEtran} \bibliography{IEEEabrv,long,references}
\fi

\clearpage
\ifarxiv
\appendix
\fi
% \addtolength{\oddsidemargin}{.7in}
% \addtolength{\evensidemargin}{.7in}
% \addtolength{\textwidth}{-1.4in}
% \large
\end{document}